\documentclass[a4paper]{article}
\usepackage{amsmath}
\usepackage[english]{babel}
\usepackage{latexsym}
\usepackage{amssymb}
\usepackage{amscd}
\usepackage{amsgen,amstext,amsbsy,amsopn}
\usepackage{math rsfs}

\usepackage{amsthm,epsfig,graphicx,graphics}
\usepackage[latin1]{inputenc}
\usepackage{xspace}
\usepackage{amsxtra}

\usepackage{color}

\newcommand{\version}{December 6th, 2010}
\renewcommand{\theequation}{\thesection.\arabic{equation}}
\numberwithin{equation}{section}
\newcommand{\bdm}{\begin{displaymath}}
\newcommand{\edm}{\end{displaymath}}
\newcommand{\bdn}{\begin{eqnarray}}
\newcommand{\edn}{\end{eqnarray}}
\newcommand{\bay}{\begin{array}{c}}
\newcommand{\eay}{\end{array}}
\newcommand{\ben}{\begin{enumerate}}
\newcommand{\een}{\end{enumerate}}
\newcommand{\beq}{\begin{equation}}
\newcommand{\eeq}{\end{equation}}
\newcommand{\bml}[1]{\begin{multline} #1 \end{multline}}
\newcommand{\bmln}[1]{\begin{multline*} #1 \end{multline*}}

\newcommand{\lf}{\left}
\newcommand{\ri}{\right}

\newcommand{\RR}{\mathbb{R}^2}

\newcommand{\rv}{\vec{r}}
\newcommand{\sv}{\vec{s}}
\newcommand{\av}{\vec{a}}
\newcommand{\avj}{\vec{a}_j}
\newcommand{\avi}{\vec{a}_i}

\newcommand{\avti}{\vec{\tilde{a}}_i}
\newcommand{\avtj}{\vec{\tilde{a}}_j}
\newcommand{\diff}{\mathrm{d}}
\newcommand{\eps}{\varepsilon}

\newcommand{\ba}{\mathcal{B}}

\newcommand{\gpf}{\mathcal{E}^{\mathrm{GP}}}
\newcommand{\gpe}{E^{\mathrm{GP}}}
\newcommand{\gpm}{\Psi^{\mathrm{GP}}}
\newcommand{\chem}{\mu^{\mathrm{GP}}}
\newcommand{\gpd}{\rho^{\mathrm{GP}}}

\newcommand{\hgpf}{\hat{\mathcal{E}}^{\mathrm{GP}}}
\newcommand{\hgpe}{\hat{E}^{\mathrm{GP}}}
\newcommand{\hgpm}{g_{\omega}}
\newcommand{\hchem}{\hat{\mu}^{\mathrm{GP}}}

\newcommand{\hgpfa}{\hat{\mathcal{E}}^{\mathrm{GP}}_{\A,\omega}}
\newcommand{\hgpea}{\hat{E}^{\mathrm{GP}}_{\A,\omega}}
\newcommand{\hgpma}{g_{\A,\omega}}
\newcommand{\hchema}{\hat{\mu}^{\mathrm{GP}}_{\A,\omega}}

\newcommand{\Eg}{\mathcal{E}}
\newcommand{\Fg}{\mathcal{F}}

\newcommand{\dg}{{\rm deg}}
\newcommand{\curl}{{\rm curl}}

\newcommand{\tff}{\mathcal{E}^{\mathrm{TF}}}
\newcommand{\tfd}{\mathcal{A}^{\mathrm{TF}}}
\newcommand{\tfe}{E^{\mathrm{TF}}}
\newcommand{\tfm}{\rho^{\mathrm{TF}}}
\newcommand{\rtf}{R_{\mathrm{h}}}
\newcommand{\rd}{R_{>}}
\newcommand{\rt}{R_{<}}

\newcommand{\tfchem}{\mu^{\mathrm{TF}}}

\newcommand{\ttfm}{\tilde{\rho}^{\mathrm{TF}}}

\newcommand{\htff}{\hat{\mathcal{E}}^{\mathrm{TF}}}
\newcommand{\htfe}{\hat{E}^{\mathrm{TF}}}
\newcommand{\htfm}{\hat{\rho}^{\mathrm{TF}}}
\newcommand{\hrtf}{\hat{R}_{\omega}}
\newcommand{\htfchem}{\hat{\mu}^{\mathrm{TF}}}

\newcommand{\ann}{\mathcal{A}}
\newcommand{\at}{\tilde{\A}}

\newcommand{\optphtf}{\omega^{\mathrm{TF}}}
\newcommand{\oopt}{\omega_{{\rm opt}}}
\newcommand{\gaintf}{H^{\mathrm{TF}}}
\newcommand{\rgaintf}{\tilde{H}^{\mathrm{TF}}}
\newcommand{\gain}{H}
\newcommand{\costtf}{F^{\mathrm{TF}}}

\newcommand{\magnp}{\vec{A}}
\newcommand{\rmagnp}{\vec{B}_{\omega}}

\newcommand{\half}{\hbox{$\frac12$}}

\newcommand{\Z}{\mathbb{Z}}
\newcommand{\R}{\mathbb{R}}
\newcommand{\N}{\mathbb{N}}

\newcommand{\C}{\mathbb{C}}
\newcommand{\D}{\mathcal{D}}
\newcommand{\F}{\mathcal{F}}
\newcommand{\E}{\mathcal{E}}
\newcommand{\A}{\mathcal{A}}
\newcommand{\B}{\mathcal{B}}
\newcommand{\T}{\mathcal{T}}
\newcommand{\HH}{\mathcal{H}}

\newcommand{\Q}{\mathcal{Q}}
\newcommand{\OO}{\mathcal{O}}

\newcommand{\al}{\alpha}
\newcommand{\alt}{\tilde{\alpha}}
\newcommand{\ep}{\varepsilon}
\newcommand{\g}{\gamma}
\newcommand{\Om}{\Omega}
\newcommand{\om}{\omega}

\newcommand{\dd}{\partial}

\newtheorem{teo}{Theorem}[section]
\newtheorem{lem}{Lemma}[section]
\newtheorem{pro}{Proposition}[section]

\newtheorem{defi}{Definition}[section]

\newcounter{remark}[section]
\newenvironment{rem}{\stepcounter{remark} \vspace{0,1cm} \noindent \textit{Remark \thesection.\theremark}\,}{\vspace{0,2cm}}

\begin{document}

\markboth{\scriptsize{Giant Vortex - CRY - \version}}{\scriptsize{Giant Vortex - CRY - \version}}

\title{The Transition to a Giant Vortex Phase in a Fast Rotating Bose-Einstein Condensate}

\author{M. Correggi	\\ \normalsize\it CIRM, Fondazione Bruno Kessler,    \\ \normalsize\it Via Sommarive 14, 38123 Trento, Italy. \\ \normalsize\it \hspace{-.5 cm}	\\	N. Rougerie	\\	\normalsize\it	Universit\'{e} Paris 6 \\ \normalsize\it Laboratoire Jacques-Louis Lions \\ \normalsize\it 175 rue du Chevaleret, 75013 Paris, France \\ \normalsize\it \hspace{-.5 cm}	\\ J. Yngvason	\\ \normalsize\it Erwin Schr{\"o}dinger Institute for Mathematical Physics,	\\ \normalsize\it Boltzmanngasse 9, 1090 Vienna, Austria,	\\ \normalsize\it Fakult\"at f\"ur Physik, Universit{\"a}t Wien,	\\ \normalsize\it Boltzmanngasse 5, 1090 Vienna, Austria.}

\date{\version}

\maketitle

\begin{abstract} 

We study the Gross-Pitaevskii (GP) energy functional for a fast rotating Bose-Einstein condensate on the unit disc in two dimensions. Writing the coupling parameter as $1/\eps^2$ we consider the asymptotic regime $\eps\to 0$ with the angular velocity $\Omega$ proportional to $(\eps^2|\log\eps|)^{-1}$.
We prove that if $\Omega=\Omega_0 (\eps^2|\log\eps|)^{-1}$ and $\Omega_0>2(3\pi)^{-1}$ then a minimizer of the GP energy functional has no zeros in an annulus at the boundary  of the disc that contains the bulk of the mass. The vorticity resides in a complementary `hole'  around the center where the density is vanishingly small. Moreover, we prove a lower bound to the ground state energy that matches, up to small errors, the upper bound obtained from an optimal giant vortex  trial function, and  also that the winding number of a GP minimizer around the disc is in accord with the phase of this trial function.
	\vspace{0,2cm}

	MSC: 35Q55,47J30,76M23. PACS: 03.75.Hh, 47.32.-y, 47.37.+q.
	\vspace{0,2cm}
	
	Keywords: Bose-Einstein Condensates, Superfluidity, Vortices, Giant Vortex.
\end{abstract}

\tableofcontents

\section{Introduction}

An ultracold Bose gas in a magneto-optical trap exhibits remarkable phenomena when the trap is set in rotational motion. In the ground state the gas is a superfluid and responds to the rotation by  the creation of quantized vortices whose number increases with the angular velocity. The literature on this subject is vast but there exist some excellent reviews \cite{A,Co,Fe1}. The mathematical analysis is usually carried out in the framework of the (time independent) Gross-Pitaevskii (GP) equation that has been derived from the many-body quantum mechanical Hamiltonian in \cite{LSY} for the non-rotating case and in \cite{LS} for a rotating system at  {\em fixed} coupling constant and rotational velocity. When these parameters tend to infinity  the leading order approximation was established  in \cite{BCPY} but it is still an open problem to derive the full GP  equation rigorously in this regime. The present paper, however,  is not concerned with the many-body problem and  the GP description will be assumed throughout.

When studying gases in rapid rotation a distinction has to be made between the case of  harmonic traps, where the confining external potential increases quadratically with the distance from the rotation axis, and anharmonic traps where the confinement is stronger. In the former case there is a limiting angular velocity beyond which the centrifugal forces drive the gas out of the trap. In the latter case the angular velocity can in principle be as large as one pleases \cite{Fe2}. Using variational arguments it was predicted in \cite{FB} that at sufficiently large angular velocities in such traps a phase transition changing the character of the wave function takes place: Vortices disappear from the bulk of the density and the vorticity resides in a `giant vortex' in the central region of the trap where the density is very low. This phenomenon was also discussed and studied numerically in both earlier and later works \cite{FJS,FZ,KB,KF,KTU}, but proving theorems that firmly establish it mathematically has remained quite challenging. 

The emergence of a giant vortex state at large angular velocity with the interaction parameter {\em fixed} was proved in \cite{R1} for a quadratic plus quartic trap potential. 
The present paper is concerned with the combined effects of a large angular velocity {\em and} a large interaction parameter on the distribution of vorticity in an anharmonic trap. In particular, we shall establish rigorous estimates on the relation between the interaction strength and the angular velocity required for creating a giant vortex.  The physical regime of interest and hence the mathematical problem treated is rather different  from that in \cite{R1} but a common feature is that the annular shape of the condensate is created by the centrifugal forces at fast rotation. By contrast, the paper  \cite{AAB} focuses on a situation where, even at slow rotation speeds, the condensate has a fixed annular shape due to the choice of the trapping potential. The basic methodology of the present paper is related to that of \cite{AAB} and \cite{IM,IM2} but with important  novel aspects that will become apparent in the sequel.
 
As in the papers \cite{CRY1,CY}, that deal mainly with rotational velocities below the giant vortex transition,  the mathematical model we consider is that of a two-dimensional `flat' trap. The angular velocity vector points in the direction orthogonal to the plane and  the energy functional in the rotating frame of reference is\footnote{The notation ``$a:=b$" means that $a$ is {\em by definition} equal to $b $.}
\beq
	\label{GPf0}
	\gpf[\Psi] : = \int_{\B} \diff \rv \: \lf\{ \lf| \nabla \Psi \ri|^2 - 2\vec\Omega\cdot\Psi^*\vec L\Psi + \eps^{-2} |\Psi|^4 \ri\},
\eeq
where we have denoted the physical angular velocity by $2\vec\Omega$, the integral domain is the unit disc 
$ \B = \{ \rv \in \RR \: : \: r \leq 1 \} $,
and $\vec L=-i \rv \wedge \nabla $ is
the angular momentum operator. 
It is also useful to write the functional in the form
\beq
	\label{GPf}
	\gpf[\Psi] = \int_{\B} \diff \rv \: \lf\{ \lf| \lf( \nabla - i \magnp \ri) \Psi \ri|^2 - \Omega^2 r^2 |\Psi|^2 + \eps^{-2} |\Psi|^4 \ri\},
\eeq
where the vector potential $ \magnp $ is given by
\beq
	\magnp : = \vec\Omega\wedge \rv=\Omega r \vec{e}_{\vartheta}.
	\eeq
Here $(r,\vartheta)$ are two-dimensional polar coordinates and $\vec{e}_{\vartheta}$ a unit vector in the angular direction.
The complex valued $ \Psi $ is normalized in $ L^2(\B) $  and the ground state energy is defined as
\beq\label{gpe} 
\gpe : = \inf_{\lf\| \Psi \ri\|_2 = 1} \gpf[\Psi]. 
\eeq
The minimization in \eqref{gpe} leads to Neumann boundary conditions on $\dd \B$. Alternatively we could have imposed Dirichlet boundary conditions, or considered the case of a homogeneous trapping potential as in \cite{CRY2}. Our general strategy applies to these cases too, but with nontrivial modifications and new aspects that are dealt with in separate papers \cite{CPRY}.\\
We denote by $\Psi^{\rm GP}$ a minimizer of \eqref{GPf}, i.e., any normalized function such that $\gpf[\Psi^{\rm GP}]=\gpe$.\\
The minimizer is in general not unique because vortices can break the rotational symmetry \cite[Proposition 2.2]{CRY1} but any minimizer satisfies the variational equation (GP equation)
\beq
	\label{GP variational}
	- \Delta \gpm - 2 \vec\Omega\cdot \vec L\, \gpm + 2 \eps^{-2} \lf| \gpm \ri|^2 \gpm = \chem \gpm,
\eeq
with Neumann boundary conditions and the chemical potential
\beq
	\label{chem}
	\chem : = \gpe + \frac{1}{\eps^2} \int_{\B} \diff \rv \: \lf| \gpm \ri|^4.
\eeq
The subsequent analysis concerns the asymptotic behavior of $\Psi^{\rm GP}$ and $\gpe$ as $\eps\to 0$ with $\Omega$ tending to $\infty$ in a definite way. 

As discussed in \cite{CRY1} the centrifugal term $- \Omega^2 r^2 |\Psi|^2$ in \eqref{GPf} creates for $\Omega\gtrsim  \eps^{-1}$ a \lq hole' around the center  where the density $|\Psi^{\rm GP}|^2$ is exponentially small while the mass is concentrated in an annulus of thickness $\sim\eps\Omega$ at the boundary. Moreover, by establishing upper and lower bounds on $\gpe$, it was shown in \cite{CY} that in the asymptotic parameter regime
\beq\label{lowregime} |\log\eps|\ll \Omega\ll (\eps^2|\log\eps|)^{-1}\eeq the ground state energy is to subleading order correctly reproduced by the energy of a trial function exhibiting a lattice of vortices reaching all the way to the boundary of the disc. It can even be shown that in the whole regime \eqref{lowregime} the vorticity of a true minimizer $\Psi^{\rm GP}$ in the annulus is uniformly distributed\footnote{Cf. \cite[Theorem 3.3]{CY}. This theorem is stated only for the case $\Omega\lesssim \eps^{-1}$ but it holds in fact true in the whole region \eqref{lowregime} \cite{R2}.} 
with density $\Omega/\pi$.

\subsection{Main Results}

In this paper we investigate the case 
\beq
	\label{angular velocity}
	\Omega = \frac{\Omega_0}{\eps^2 |\log\eps|}
\eeq
with $\Om_0$ a fixed constant and prove that for  $\Omega_0$ sufficiently large a phase transition takes place: Vortices disappear from the bulk of the density and all vorticity is contained in a hole where the density is low.

To state this result precisely, we first recall from \cite{CRY1} that the leading order in the asymptotic expansion of the GP ground state energy is given by the minimization of the \lq Thomas-Fermi\rq\ 
(TF) functional obtained from $ \gpf $ by simply neglecting the kinetic term:
\beq
	\label{TFf}
	\tff[\rho] : = \int_{\B} \diff \rv \: \lf\{ - \Omega^2 r^2 \rho + \eps^{-2} \rho^2 \ri\},
\eeq
where $ \rho $ plays the role of the density $ |\Psi|^2 $.  The minimizing normalized density, denoted by $ \tfm $, can be  computed
explicitly (see  Appendix).  For $\Omega>2(\sqrt{\pi}\eps)^{-1}$ it vanishes for $r< \rtf $ with
\beq
	\label{TFann0}
	1-\rtf^2 = (2/\sqrt \pi) (\eps \Omega)^{-1}\sim \Omega_0^{-1}\eps|\log\eps|,
\eeq
while for $\rtf\leq r\leq 1$ it is given by
\beq
\label{TFm0} \tfm(r) =\half \eps^2\Omega^2 (r^2-\rtf^2).
\eeq
Our proof of the disappearance of vortices is based on estimates that require the TF density to be sufficiently large. For this reason we consider an annulus
\beq \label{abulk}\tilde{\mathcal A}:=\{\rv \: :\: \rtf +\eps|\log\eps|^{-1}\leq r\leq 1\}\eeq
where $\rho^{\rm TF}(r)\gtrsim \Omega_0\eps^{-1}|\log\eps|^{-3}$. Our result on the disappearance of vortices from the  essential support of the density is as follows:

\begin{teo}[\textbf{Absence of vortices in the bulk}]\label{theo:vortex}
\mbox{}	\\
If $\Omega$ is given by \eqref{angular velocity} with $\Omega_0>2(3\pi)^{-1}$, then $\Psi^{\rm GP}$ does not have any zero  in $\tilde{\mathcal A}$ for $\eps>0$ small enough. More precisely, for $\rv\in \tilde{\mathcal A}$, 
 \beq
 \label{densitydiff}
 \left||\Psi^{\rm GP}(\rv)|^2-\tfm(r)\right|\leq C\frac{|\log\eps|^3}{\eps^{7/8}}\ll\tfm(r).
 \eeq
\end{teo}

\begin{rem}{\it (Bulk of the condensate).}
\mbox{}	\\
The annulus $\tilde{\mathcal A}$ contains the bulk of the mass.  Indeed, because  $\int_{\tilde{\mathcal A}}\tfm=1-o(1)$ and $\Psi^{\rm GP}$ is normalized, \eqref{densitydiff} implies that also $\int_{\tilde{\mathcal A}}|\Psi^{\rm GP}|^2=1-o(1)$.
\end{rem}

The proof of Theorem 1.1 is based on precise energy estimates that involve a comparison of the energy of the restriction of $\Psi^{\rm GP}$ to $\tilde{\mathcal A}$ with the energy of a giant vortex trial function. By the latter we mean a function of the form

\beq\label{gvgeneral}
\Psi(\rv)=f(\rv)\exp(i\hat\Omega \vartheta)
\eeq
with $f$ {\em real valued} and $\hat\Omega\in \mathbb N$. It is convenient to write\footnote{We use the notation `$ [ \: \: \cdot \: \: ]$' for the integer part of a real number.} 
\beq\label{hatomega}
\hat\Omega =[\Omega]-\omega
\eeq 
with $\omega\in\mathbb Z$ and use $\omega$ as label for the 
$\hat\Omega$-dependent quantities in the sequel.  We define a functional of $f$ by
\begin{equation}\label{hatfunct}
\hat{\mathcal E}^{\rm GP}_{\omega}[f] := {\mathcal E}^{\rm GP}[f\exp(i\hat\Omega \vartheta)]=
\int_{\mathcal B} \diff \rv \:\left\{| \nabla f|^2+(\hat\Omega^2\,r^{-2}-2\Omega\hat\Omega)f^2+ \eps^{-2} f^4 \right\}.
\end{equation}

Minimizing the functional $\hat{\mathcal E}^{\rm GP}_{\omega}[f]$ (see Proposition \ref{htminimization}), that is convex in $ f^2 $, over $L^2$-normalized $f$ for fixed $\hat\Omega$, we obtain an energy denoted by
$\hat{E}^{\rm GP}_{\omega}$ and a minimizer that is rotationally symmetric, without a zero for $r>0$,  and unique up to sign\footnote{Instead of requiring $f$ to be real we could have required $f$ to be radial in which case the minimizer is unique up to a constant phase.}. We denote the unique positive minimizer by $g_\omega$. If $0\leq \omega\leq C\eps^{-1}$ the corresponding density $g_\omega^2$ is close to $\rho^{\rm TF}$ (see Remark 1.3 below).

It is clear that
\beq\label{upperhat}
{E}^{\rm GP}\leq{\mathcal E}^{\rm GP}[g_\omega\exp(i\hat\Omega \vartheta)] =\hat{E}^{\rm GP}_{\omega}
\eeq
for any value of $\omega$ and hence
\beq\label{upperhat1}
{E}^{\rm GP}\leq {\hat E}^{\rm GP}:=\inf_{\omega \in \Z} \hat{E}^{\rm GP}_{\omega}.
\eeq

Our second main result is a lower bound that matches \eqref{upperhat1} up to small errors (see Remark 1.4 below) and an estimate of the phase that optimizes $\hat{E}^{\rm GP}_{\omega}
$.

\begin{teo}[\textbf{Ground state energy}]\label{theo:energy}
\mbox{}	\\
For  $\Omega_0>2(3\pi)^{-1}$ and $\eps>0$ small enough the ground state energy is
\beq
E^{\rm GP}\label{enas1}=\hat E^{\rm GP}- \OO\left(\frac{|\log\eps|^{3/2}}{\eps^{1/2}(\log|\log\eps|)^2}\right).
\eeq
Moreover $\hat E^{\rm GP}=\hat E^{\rm GP}_{\omega_{\rm opt}}$, with $\omega_{\rm opt}\in\mathbb N$ satisfying
\beq\label{optphase}
\omega_{{\rm opt}}=\frac2{3\sqrt\pi \eps} \left(1+ \OO(|\log\eps|^{-1/2})\right).
\eeq
\end{teo}

Because $\Psi^{\rm GP}$ does not have any zeros in the annulus $\tilde{\mathcal A}$ the winding number (degree) of   $\Psi^{\rm GP}$ around the unit disc is well defined.
Our third main result is that the giant vortex phase $[\Omega]-\omega_{\rm opt}$ is, up to possible small errors, equal to this winding number. 

\begin{teo}[\textbf{Degree of a minimizer}] \label{theo:degree}
\mbox{}	\\
For $\Omega_0> 2(3\pi)^{-1}$  and $\eps>0$ small enough, the winding number of $\Psi^{\rm GP}$  is
\beq
	\label{gpm degree}
	 {\rm deg}\; \{\Psi^{\rm GP}, \partial \mathcal B\}=[\Omega]-\omega_{\rm opt}(1+o(1))\eeq
 with $\omega_{\rm opt}$ as in \eqref{optphase}.\end{teo}

The following remarks are intended to elucidate the phase difference $\omega_{\rm opt}$ and to justify the claim that the remainder in \eqref{enas1} is, indeed, a small correction to $\hat E^{\rm GP}$.

\begin{rem} ({\it Giant vortex phase}).
\label{rem: giant vortex phase}
\mbox{}	\\
For a fixed  trial function $f$ the energy \eqref{hatfunct} is minimal for 
\beq\label{hatopt}
\hat\Omega=\Omega \left(\int \diff \rv \:  r^{-2}f^2\right)^{-1}+ \OO(1).
\eeq
If $f^2$ is close to $\rho^{\rm TF}$ and thus essentially concentrated in an annulus of width $\sim(\eps\Omega)= \OO(\eps|\log\eps|)$, it follows that
\beq\label{omegadiff} 0<(\Omega-\hat\Omega)= \OO(\eps^{-1}).\eeq
The leading term $2/(3\sqrt\pi)\eps^{-1}$ in \eqref{optphase} can, indeed, be computed  from \eqref{hatopt} with $f^2=\rho^{\rm TF}$, or more simply, for $f^2(r)$ increasing linearly from 0 at $r=R_{\rm h}$ to its maximal value at $r=1$.
It is worth noting that if $f^2(r)$ were constant in the interval $[R_{\rm h},1]$, the leading term for $\omega$ according to  \eqref{hatopt} would be $2/(\sqrt\pi\eps)$ and the phase $\hat\Omega$ therefore approximately equal to  $(\Omega/\pi)\times\text{area of the hole}$.  The optimal $\hat\Omega$ is closer to $\Omega$ because of the inhomogeneity of the density in the annulus.
\end{rem}

\begin{rem} ({\it Giant vortex density functional}).
\label{rem: giant vortex density}
\mbox{}	\\
The functional \eqref{hatfunct} can also be written
\begin{equation}\label{hatfunct2}
\hat{\mathcal E}^{\rm GP}_{\omega}[f] = \int_{\mathcal B}\diff \rv \: \left\{| \nabla f|^2-\Omega^2 r^{2}f^2+B_\omega^2 f^2+ \eps^{-2} f^4 \right\}
\end{equation}
with
\beq\label{newB0}
B_{\omega}(r) : =  \Omega r - {\hat \Omega }/{r}.
\eeq
If $\omega= \OO(\eps^{-1})$, then $B_{\omega}(r)= \OO(\eps^{-1})$ close to the boundary of the disc. Neglecting the term $B_\omega^2 f^2$ as well as the gradient term in \eqref{hatfunct2} in comparison with the centrifugal and interaction terms leads to the TF functional evaluated at $f^2$.  This makes plausible the assertion above that $g_\omega^2$ is close to $\rho^{\rm TF}$ if  
$\omega= \OO(\eps^{-1})$ and this will indeed be proved in Section \ref{sec:estimates}. 
\end{rem}

\begin{rem} ({\it Composition of the ground state energy}).
\mbox{}	\\
If one drops the gradient term in  \eqref{hatfunct2} but retains the term with $B_\omega^2$, one obtains a modified (and $\omega$-dependent) TF functional
\beq\label{hatTF}
\hat{\mathcal E}^{\rm TF}_{\omega}[\rho] : = \int_{\mathcal B}\diff \rv \: \left\{-\Omega^2 r^{2}\rho+B_\omega^2 \rho+ \eps^{-2} \rho^2 \right\}.
\eeq
Its minimizer and minimizing energy $\hat E_\omega^{\rm TF}$ can be  computed explicitly (see the Appendix) and one sees that the energy is minimal for $\omega=2/(3\sqrt\pi\eps)(1+o(1))$. Denoting the minimal value by $\hat E^{\rm TF}$, we have
\beq\label{hatTFenergy}
\hat E^{\rm TF}=E^{\rm TF}+ \OO(\eps^{-2}),
\eeq
while
\beq E^{\rm TF}=-\Omega^2-4/(3\sqrt\pi)\eps^{-1}\Omega=\OO(\eps^{-4}|\log\eps|^{-2})+\OO(\eps^{-3}|\log\eps|^{-1}).
\eeq
The term $\OO(\eps^{-2})$ in \eqref{hatTFenergy} is the angular kinetic energy corresponding to the third term in \eqref{hatfunct2}.
The difference between $\hat E^{\rm GP}$ and $\hat E^{\rm TF}$ is the radial kinetic energy of the order $(\eps\Omega)^{2}|\log\eps|=\OO(\eps^{-2}|\log\eps|^{-1})$. The remainder in 
\eqref{enas1} is thus much smaller than all terms in $\hat E^{\rm GP}$. It is also smaller than the energy a vortex in $\tilde{\mathcal A}$ would have, which is $|\log\eps|\times\,{\rm density}=\OO(\eps^{-1})$.
\end{rem}

\subsection{Proof Strategy}
\label{sec:proof strategy}

We now explain the general strategy for the proof of the main results. 

The first step is an energy splitting as in \cite{LM}. We write a generic, normalized wave function as $\Psi(\rv)= g_\omega(r)w(\rv)$ with a complex valued function $w$. The identity
\beq\int_{\mathcal B} \diff \rv \: |\nabla\Psi|^2=\int_{\mathcal B} \diff \rv \: (-\Delta g_\omega)g_\omega|w|^2+ \int_{\mathcal B} \diff \rv  \: g_\omega^2|\nabla w|^2\eeq
holds by partial integration and the boundary condition on $g_\omega$. Using the variational equation for $g_\omega$ and writing $w(\rv)=\exp(i\hat\Omega\vartheta) v(\rv)$ and $\hat\Omega$ as in \eqref{hatomega}, we obtain a splitting of the energy functional:
\beq \label{split}\mathcal E^{\rm GP}[\Psi]=\hat E^{\rm GP}_\omega+\mathcal E_\omega[v]\eeq
with
\beq\label{eomega0} \mathcal E_\omega[v] :=\int_{\mathcal B} \diff \rv \lf\{ g_\omega^2|\nabla v|^2-2g_\omega^2\vec B_\omega\cdot (iv,\nabla v)+\frac {g_\omega^4}{\eps^2}(1-|v|^2)^2 \ri \}
\eeq
where
\beq
	\label{vec B}
	\vec B_\omega(r) : =B_\omega(r) \vec e_\vartheta = \lf[ \Omega r - \lf( [\Omega] - \omega \ri) r^{-1} \ri]  \vec e_\vartheta
\eeq 
and we have used the notation 
\beq (iv,\nabla v) : = \half i( v\nabla v^* - v^* \nabla v).
\eeq
Since the divergence of the two-dimensional vector field $g_\omega^2\vec B_\omega$ vanishes we can write 
\beq 2g_\omega^2\vec B_\omega=\nabla^{\perp}F_\omega\eeq with a scalar potential $F_\omega$ and the dual gradient $\nabla^{\perp}=(-\partial_y,\partial_x)$. Stokes theorem then gives
\beq\label{eomega}\mathcal E_\omega[v] = \int_{\mathcal B} \diff \rv \lf\{ g_\omega^2|\nabla v|^2
+F_\omega{\rm curl}(iv,\nabla v) \ri\} -\int_{\partial \mathcal B} \diff \sigma \: F_\omega(iv,\partial_\tau v)
+\int_{\B} \diff \rv \:\frac {g_\omega^4}{\eps^2}(1-|v|^2)^2 .
\eeq
Here and in the rest of the paper \lq curl $(iv,\nabla v)$' stands for the 3-component of $\nabla\wedge (iv,\nabla v)$, $\partial_\tau$ is the tangential derivative and $ \diff \sigma $  the Lebesgue measure on the circle, i.e, given a ball $ \B_R $ of radius $ R $ centered at the origin, $ \partial_{\tau} : = R^{-1} \partial_{\vartheta} $ and $ \diff \sigma : = R \diff \vartheta $ on $ \partial \B_R $.

The giant vortex trial function $g_\omega \exp(i\hat\Omega\vartheta)$ gives an upper bound to the energy and hence we see from \eqref{split} that
\beq\label{simplebound1}
\mathcal E_\omega[u]\leq 0
\eeq
where $u$ is the remaining factor of  $\Psi^{\rm GP}$ after the giant vortex trial function has been extracted, e.g.,
$\Psi^{\rm GP}=g_\omega\exp(i\hat\Omega\vartheta)\, u$. The proof of Theorem \ref{theo:vortex} is based on a lower bound on $\mathcal E_\omega[u]$ that, for $\Omega_0$ large enough, would be positive and hence contradict \eqref{simplebound1} if $u$ had zeros in $\tilde{\mathcal A}$. The main steps are as follows:
\begin{itemize}
\item {\em Concentration of the density:} There is an annulus $\mathcal A$, slightly larger than $\tilde{\mathcal A}$, so that the supremum of $|\Psi^{\rm GP}|^2$ over $\mathcal B\setminus \mathcal A$ is\footnote{We use the symbol $ \OO(\eps^{\infty}) $ to denote a remainder which is smaller than any positive power of $ \eps $, e.g., exponentially small.} $\OO(\eps^\infty)$. The same holds for $g_\omega^2$ provided $\omega\leq C\eps^{-1}$. This implies that for $v=u$ the integrations in \eqref{eomega0} can be restricted to $\mathcal A$ up to errors that are $\OO(\eps^\infty)$.
\item {\em Optimization of the phase:} Choice of $\omega=\OO(\eps^{-1})$ so that $F_\omega$, chosen to vanish on the inner boundary of $\mathcal A$, is small on $\partial \B$.
\item {\em Boundary estimate:} Using the variational equations for $\Psi^{\rm GP}$ and $g_\omega$ and the smallness of $F_\omega$ on the boundary it is proved that the boundary term in \eqref{eomega} is small. This step uses the Neumann boundary conditions on $\dd \B$ in a strong way and its extension to the case of Dirichlet boundary conditions is a nontrivial open problem\footnote{This problem has recently been solved \cite{CPRY}.}.
\item {\em Vortex balls:} Isolation of the possible zeros of $u$ in \lq vortex balls' and an estimate of the first term in \eqref{eomega} from below in terms of the infimum of $g_\omega^2$ on the vortex balls and the winding numbers of $u$ around the centers. For the construction of vortex balls a suitable upper bound on the last term in \eqref{eomega} is essential.
\item {\em Jacobian estimate:} Approximation of the integral of $F_\omega$ against the vorticity ${\rm curl}(iu,\nabla u)$ by a sum of  the values of $F_\omega$ at the centers of the vortex balls multiplied by the corresponding winding numbers.
\item{\em Gradient estimate:} Estimates on $\nabla u$ leading to a lower bound on the last term in \eqref{eomega} that, for $\Omega_0$ large enough,  excludes zeros of $u$.
\end{itemize}
A key point to notice is that the vortex ball construction in combination with the jacobian estimate leads to a {\em cost function} defined
as
\beq
H_\omega(r) :=\half g_\omega^2(r)|\log\eps|+F_\omega(r).
\eeq
The first term is, to leading approximation, the kinetic energy of a vortex of unit strength at radius $r$, the second term is the gain due to the potential energy of the vortex in the field $ 2 g^2_\omega\vec B_\omega$. If $H_\omega$ is negative at some point, the energy may be lowered by inserting a vortex at this point. A positive value of $H_\omega$ means that the cost of the kinetic energy outweighs the gain. Note that $H_\omega$ depends on $\Omega_0$ through $F_\omega$; in fact it is not difficult to see that\footnote{Here, as in the rest of the paper, $C$ stands for a finite, positive constant whose value may vary from one formula to another.}
\beq
|F_\omega(r)|\leq\frac C{\Omega_0}g_\omega^2(r)|\log\eps|.
\eeq 
Hence, if $\Omega_0$ is large enough, the cost function is positive everywhere on the annulus and  vortices are energetically unfavorable. An upper bound, $ 2 (3 \pi)^{-1}$, on the critical value for $\Omega_0$ is computed in the Appendix. This upper bound is in fact optimal, as demonstrated in \cite{R3}. The proof requires additional ingredients.

The construction of vortex balls is a technique introduced independently by Jerrard \cite{J} and Sandier \cite{Sa} in Ginzburg-Landau (GL) theory, whereas the jacobian estimate originates from the work \cite{JS}. Both techniques are described in details in the monograph \cite{SS}.  This method has been applied in \cite{AAB} and \cite{IM,IM2}  to functionals that at first sight look exactly like \eqref{eomega0}. There is, however,  an essential difference: While in \cite{AAB}  and \cite{IM,IM2} the size of the relevant integral domain  is fixed, the weight function $g_\omega^2$ is in our case concentrated in an annulus whose width tends to zero as $\eps\to 0$. This shrinking of the width of the annulus and the large gradient of  $g_\omega^2$ are the main reasons for the complications encountered in the proof of our main theorems.

To appreciate the difficulty the following consideration is helpful. As already noted, the concentration of $g_\omega^2$ on the
annulus $\tilde{\mathcal A}$  of width 
$\OO(\eps|\log\eps|)$ implies that $B_\omega=\OO(\eps^{-1})$ on this annulus. From \eqref{simplebound1} and the Cauchy-Schwarz inequality applied to the second term in \eqref{eomega0}, one obtains the bound
\beq\label{estimate0}
\int_{\tilde{\mathcal A}} \diff \rv \: \frac {g_\omega^4}{\eps^2}(1-|u|^2)^2\leq \frac C{\eps^2}.
\eeq
Now, although $g_\omega^4\sim (\eps\Omega)^2\sim \eps^{-2}|\log\eps|^{-2}$ on $\tilde{\mathcal A}$, the estimate \eqref{estimate0} is not sufficient for the construction of vortex balls on $\tilde{\mathcal A}$. In fact, \eqref{estimate0} is compatible with the vanishing of $u$ on a ball of radius $\OO(\eps|\log\eps|)$, i.e., comparable to the width of $\tilde{\mathcal A}$, while for a construction of vortex balls one must be able to isolate the possible zeros of $u$ in balls of much smaller radius.  

The solution of this problem, elaborated in Section \ref{Sect est reduced energy}, involves a division of the annulus into cells of area $\sim (\eps\Omega)^2$ with the upshot that a local version of \eqref{estimate0} holds in every cell. The local version,
\beq\label{estimate0cell}
\int_{\rm Cell} \diff \rv \: \frac {g_\omega^4}{\eps^2}(1-|u|^2)^2\leq \frac C{\eps^2}\times\text{(number of of cells)}^{-1}\sim \frac{|\log\eps|}\eps
\eeq
means that, in the cell,  $u$ can only vanish in a region of area $\OO(\eps^3|\log\eps|^3)$
that is much smaller than the area of the cell, i.e., $\sim \eps^2|\log\eps|^2$. Hence the vortex ball technique applies in the cells where \eqref{estimate0cell}, or a sufficiently close approximation to it,  holds. The gist of the proof of a global lower bound, that eventually leads to Theorems \ref{theo:vortex} and \ref{theo:energy},  is a stepwise increase in the number of  cells where an estimate close to \eqref{estimate0cell} is valid until all potential zeros in the annulus are included in vortex balls. The proof of Theorem \ref{theo:degree} involves in addition an estimate on the winding number of $u$.

In the sketch of the proof strategy above we have for simplicity deviated slightly from the actual procedure that will be followed in the sequel. Namely, for technical reasons, we find it  necessary to restrict the considerations to a problem on the annulus {\it before} the splitting of the GP energy functional as in \eqref{split}. This means that, instead of the function $g_{\omega_{\rm opt}}$ and the optimal phase $\omega_{\rm opt}$ defined above, we shall work with corresponding quantities for a functional like \eqref{hatfunct2} but with the integration restricted to $\mathcal A$. The main reason for this complication is  lack of precise information about the behavior of $\Psi^{\rm GP}$ in a neighborhood of the origin. In fact,  the distribution of the  `giant vorticity' of $\Psi^{\rm GP}$  in the  hole is unknown. In particular it is not known whether $\Psi^{\rm GP}$ vanishes at the origin like $g_\omega$, which has a zero there of order $\hat\Omega$. 
 
\subsection{Heuristic Considerations}\label{scaling}
A heuristic argument for the giant vortex transition, based on an analogy with an electrostatic problem, has been given in \cite{CY} and goes as follows: Writing a wave function as
$\Psi(\rv)=|\Psi(\rv)|\exp(i\varphi(\rv))$ with a real phase $\varphi$ the kinetic energy term in \eqref{GPf} is
\beq
\int_{\mathcal B} \diff \rv \left\{|\nabla|\Psi||^2+|\Psi|^2|\nabla\varphi-\vec A|^2\right\}
\eeq
with $\vec A(\rv)=\Omega r\vec e_{\vartheta}$. We focus on the second term and consider the situation where the phase $\varphi$ contains,  besides the giant vortex, also a possible contribution from a vortex of unit degree at a point $\rv_0$ in the annulus of thickness $\sim \left(\Omega\eps\right) ^{-1}$. The phase can thus be written
\beq
\varphi(\rv) :=\hat\Omega\vartheta  +\arg(\rv-\rv_0).
\eeq
where $\arg(x,y)=\arctan(y/x)$. The modulus $|\Psi|$ vanishes at $\rv_0$ and is small in a disc (`vortex core') of radius $\sim\sqrt{\eps/\Omega}\sim \eps^{3/2}|\log \ep |^{1/2}$ around $\rv_0$. The phase $\arg(\rv-\rv_0)$ can be regarded as the imaginary part of the complex logarithm if we identify $\mathbb R^2$ with $\mathbb C$. The corresponding real part, i.e., conjugate harmonic function, is $\log|\rv-\rv_0|$ which is the two-dimensional electrostatic potential of a unit point charge localized in $\rv_0$. Likewise, the conjugate harmonic function of $\hat\Omega\vartheta$ is $\hat\Omega\log r$, i.e., the electrostatic potential  of a charge $\hat\Omega$ placed at the origin. By the Cauchy-Riemann equations for the complex logarithm we can write 
\beq\label{CR} \nabla \varphi=\partial_r\varphi \: \vec e_r+r^{-1}\partial_{\vartheta}\varphi \: \vec e_{\vartheta}=
\partial_r\chi \: \vec e_{\vartheta}-r^{-1}\partial_{\vartheta}\chi \: \vec e_{r}
\eeq
with
\beq\label{conjugate} \chi(\rv) :=\hat\Omega \log r+\log|\rv-\rv_0|.\eeq
After a rotation by $\pi/2$, i.e., replacement of $\vec e_{\vartheta}$ by $\vec e_r$ at every point, the
vector potential also has an electrostatic interpretation: $\Omega r\vec e_r$ is the electrostatic field of a uniform charge distribution with charge density $\Omega/\pi$. Employing \eqref{CR} we  now have
\beq\label{dual}
|\nabla\varphi-\vec A|^2=|\nabla\chi-\Omega r\vec e_r|^2
\eeq 
and $\vec E(\rv):=\nabla\chi-\Omega r\vec e_r$ is the electric field generated by the point charges and the uniform background.  

We can now apply {\em Newton's theorem} to argue that the effect of the giant vortex, i.e., the first term in \eqref{conjugate}, is to neutralize in the annulus the field generated by the uniform charge distribution in the `hole', i.e., the second term in \eqref{dual}, provided $\hat\Omega$ is chosen to match the area of the hole times the uniform charge density.  By the same argument the point charge at $\rv_0$ neutralizes, outside a disc of radius $\sim 1/\sqrt{\Omega}$, the effect of one unit of the continuous charge.

The total charge of the continuous distribution in the annulus is $Q\sim \Omega (\Omega\eps)^{-1}\sim\eps^{-1}$.  Inserting the vortex reduces the effective charge by one unit so the corresponding energy gain is  $\sim \eps^{-1}$. On the other hand, the energy associated with the electrostatic field from the point charge outside the vortex core (that is cut off by the modulus of the wave function) is $\sim|\Psi|^2|\log\eps|\sim(\Omega\eps)|\log\eps|$.  The condition for the cost outweighing the gain is $ \Omega \eps |\log\eps| \gtrsim \eps^{-1} $, i.e.,
\beq
 \Omega\gtrsim \frac 1{\eps^2|\log\eps|},
\eeq
marking the transition to the giant vortex phase.

It is also instructive to consider a version of the functional \eqref{eomega0}  where the variables have been scaled so that the width of the annulus becomes $\OO(1)$. We define
\beq\ell : =(\eps\Omega)^{-1}=\Omega_0^{-1}\eps|\log\eps|,\quad \vec s : =\ell^{-1}\vec r, \quad \check g(\vec s) : =g_\omega(\ell \vec s), \quad \check v(\vec s) : =v(\ell\vec s), \quad \vec  {\check B}(\vec s) : =\ell \vec B_\omega(\ell\vec s).
\eeq
Then \eqref{eomega0} is equal to 
\beq\label{scaledfunc}
\check{{\mathcal E}}[\check v] :=\int_{\check {\mathcal B}}\diff\vec s \: \check g^2 \left\{|\nabla \check v|^2-2\vec {\check B}\cdot (i\check v,\nabla\check  v)+\frac {\check g^2 \ell ^2 }{\eps^2}(1-|\check v|^2)^2\right\}
\eeq
with $\check {\mathcal B} : =\mathcal B_{\ell^{-1}}$ a ball of radius $\ell^{-1}=\Omega_0\,\eps^{-1}|\log\eps|^{-1}$. 
For the sake of a heuristic consideration  we now assume that $g^2_\omega$ is constant, $\sim \ell^{-1}$, on the annulus of width $\ell$ and zero otherwise. We define a new small parameter
\beq
\check \eps : =\eps\ell^{-1/2}\sim\eps^{1/2}|\log\eps|^{-1/2}
\eeq
and note that
\beq | \check B|\sim\eps^{-1}\ell\sim \Omega_0^{-1}(|\log\check\eps|+\OO(\log|\log\check\eps|)).
\eeq
We are thus led to consider the functional
\beq
\int_{\check {\mathcal A}} \diff \sv \lf\{ |\nabla \check v|^2-2 \vec {\check B}\cdot (i\check v,\nabla\check  v)+\frac 1{\check\eps^2}(1-|\check v|^2)^2 \ri\} \eeq
on an annulus $\check{\mathcal  A}$ of width $\OO(1)$ and an effective vector potential of strength $\OO(\Omega_0^{-1}|\log\check\eps|)$. This is reminiscent of the situation considered in \cite{AAB} and \cite{IM,IM2} where the  domain is fixed and the rotational velocity is proportional to the logarithm of the small parameter. Moreover, increasing $\Omega_0$ {\em decreases} the coefficient in front of the logarithm. Hence for large $\Omega_0$ this coefficient is small and {\em if} the analysis of \cite{AAB} and \cite{IM,IM2} would apply, one could conclude that vortices are absent. This reduction of the problem to known results is, however, too simplistic because the annulus $\check{\mathcal  A}$ is not fixed: Although its width stays constant, its diameter and hence the area increases as $\check \eps^{-2}|\log\check\eps|^{-2}$. A new ingredient is needed, and in our approach this is the division of the annulus into cells as mentioned in the previous subsection. In the scaled version \eqref{scaledfunc} of the energy functional  these cells are (essentially) of fixed size. When writing the actual proofs we prefer to use the original unscaled variables but the picture provided by the scaling is still helpful, in particular for comparison with \cite{AAB} and \cite{IM,IM2}.

\begin{rem}{\it (Alternative approach)}.	\\
	After the submission of this paper we learned \cite{J3} of a possible alternative approach to prove the absence of vortices in the bulk, relying on \cite[Lemma 8]{J2} (see also \cite[Lemma 4.1]{AJR}). This method could replace some of the arguments in Section \ref{Sect est reduced energy} and lead to a shorter proof of our Theorem \ref{theo:vortex} but is likely to yield worse remainder terms in the energy (Theorem \ref{theo:energy}). We also stress that the tools developed in Section \ref{Sect est reduced energy} of the present paper are an essential input in the paper \cite{R3}, where it is proved that vortices {\it do} appear in the bulk if $\Omega_0 < 2(3\pi)^{-1}$. For this Lemma 8 in \cite{J2} would not be sufficient.
\end{rem}

\subsection{Organization of the Paper}

The paper is organized as follows. In Section \ref{sec:estimates} we gather some definitions and notation that are to be used in the rest of the paper. We also prove useful estimates on the matter densities, both for the actual GP minimizer $|\gpm| ^2$ and for the `giant vortex profiles' $g_{\om}$. In Section \ref{sec:auxiliary} we introduce the auxiliary problem that we are going to study on the annulus $\A$ and prove that it indeed captures the main energetic features of the full problem. 
Section \ref{Sect est reduced energy} is devoted to the study of the auxiliary problem via tools from the Ginzburg-Landau theory. We conclude the proof of our main results in Section \ref{sec:energy asympt}. Finally an Appendix gathers important facts about the TF functionals that we use in our analysis, as well as the analysis of the cost function that leads to an upper bound on the critical speed.

\section{Useful Estimates}
\label{sec:estimates}

In this section we state some useful estimates which are going to be used in the rest of the paper. Some of them are simple consequences of energy considerations and in particular energy upper bounds, whereas others depend in a crucial way on the variational equations solved by $ \gpm $, etc., and apply therefore only to (global) minimizers.

In the first part of the section we investigate the properties of any GP minimizer: Starting from simple energy estimates, we prove the concentration of $ \gpm $ on an annulus of width $ \OO(\eps |\log\eps|) $ close to the boundary of the trap and in particular its exponential smallness inside the hole. This is the major result about $ \gpm $.
\newline
The second part of the section is devoted to the analysis of the densities associated with the giant vortex energy. We shall introduce first some notation and then, exploiting energy considerations, discuss the properties of those minimizers: Most of them are very similar to the ones proven for $ \gpm $, as, e.g., the exponential smallness, but we also need other general properties as, for instance, an a priori bound on the gradient of the densities and an estimate of their difference from the TF density $ \tfm $.

\subsection{Estimates for GP Minimizers}

We briefly recall the main notation: Given the GP energy functional $ \gpf[\Psi] $ defined in \eqref{GPf}, we denote by $ \gpe $ its infimum over $L^2$-normalized wave functions. Any GP minimizer is denoted by $ \gpm $ and solves the variational equation \eqref{GP variational}.
\newline
The TF functional $ \tff[\rho] $ was introduced in \eqref{TFf} and $ \tfe $ and $ \tfm $ stand for its ground state energy and density respectively (see also the Appendix).
\newline
Any further label to the functionals $ \gpf $ and $ \tff $, as, e.g.,  $ \gpf_{\D} $, denotes a restriction of the integration in the functional to the domain $ \D $. The same convention is used for the corresponding ground state energies and minimizers.
\newline
Finally we use the notation $ \B(\rv,\varrho) $ for a ball of radius $ \varrho $ centered at $ \rv $, whereas $ \B_R $ is a ball with radius $ R $ centered at the origin.

The starting point is a simple GP energy upper bound proven for instance in \cite[Proposition 4.2]{CY}, i.e.,
\beq
	\label{CY energy upper bound}
	\gpe \leq \tfe + \Omega |\log\eps| (1 + o(1)) \leq \tfe + C \eps^{-2}.
\eeq
A straightforward consequence of such an upper bound is that the $L^2$ norm of $ \gpm $ is concentrated in the support  $ \tfd := \{\rv \: : \: \rtf \leq r \leq 1\} $ of the TF minimizer. At the same time the bound implies a useful upper bound on $ |\gpm | $:

	\begin{pro}[\textbf{Preliminary estimates for $ \gpm $}]
		\label{GPmin}
		\mbox{}	\\
		As $ \eps \to 0 $,
		\beq
			\label{GPmin estimates}
			\lf\| |\gpm|^2 - \tfm \ri\|_{L^2(\ba)} \leq \OO(1),	\hspace{1,5cm}	\lf\| \gpm \ri\|^2_{L^{\infty}(\ba)} \leq \lf\| \tfm \ri\|_{L^{\infty}(\ba)}  \lf(1 + \OO(\sqrt{\eps|\log\eps|}) \ri).
		\eeq
	\end{pro}

	\begin{proof}
	 	In order to prove the first statement it is sufficient to use the fact that $ \tfm $ is the positive part of the function $(\eps^2/2)(\tfchem + \Omega^2 r^2)$ together with the normalization of $ \tfm $ and $\Psi^{\rm GP}$ and estimate
		\bml{
 			\label{l2 difference}  
			\int_{\B} \diff \rv \: \lf[ |\gpm|^2 - \tfm \ri]^2 = \lf\| \gpm \ri\|^4_4 + \lf\| \tfm \ri\|_2^2 - 2 \int_{\B} \diff \rv \: |\gpm|^2 \: \tfm \leq	\\
			\lf\| \gpm \ri\|^4_4 + \lf\| \tfm \ri\|_2^2 - \eps^2 \tfchem - \eps^2 \Omega^2 \int_{\B} \diff \rv \: r^2 |\gpm|^2 = \eps^2 \lf( \tff\lf[|\gpm|^2\ri] - \tfe \ri) \leq	\\
			\eps^2 \lf( \gpe - \tfe \ri) \leq C,
		}
		by \eqref{CY energy upper bound}.	
		
		The proof of the second inequality is similar to the proof of Lemma 5.1 in \cite{CY} and involves the variational equation \eqref{GP variational}. We define
		\beq
			\label{GP density}
			\gpd : = | \gpm |^2.
		\eeq
		The crucial point is that at any maximum point of $ \gpd $ in the closed ball $ \bar{\B} $, it has to be $ \Delta \gpd \leq 0 $. This is trivially true in the open ball $ \B $ but can be extended to the boundary thanks to Neumann boundary conditions: Since $ \partial_r \gpm = 0 $ at the boundary $ \partial \B $, which implies $ \partial_r \gpd = 0 $ there, $ \gpd $ can have a maximum at $ \rv_0 \in \partial \B $ only if $ \nabla \gpd (\rv_0) = 0 $ and therefore $ \Delta \gpd (\rv_0) \leq 0 $. 
		\newline
		From the variational equation \eqref{GP variational} solved by $ \gpm $, one can estimate 
		\beq
			\label{var eq density}
			- \Delta \gpd \leq 4 \eps^{-2} \lf( \eps^2 \chem + \eps^2 \Omega^2 r^2 - 2 \gpd \ri) \gpd \leq 4 \eps^{-2} \lf( \eps^2 \chem + \eps^2 \Omega^2 - 2 \gpd \ri) \gpd,
		\eeq
		by using the properties
		\bdm
			- \Delta \gpd = - {\gpm}^* \Delta \gpm - \gpm \Delta {\gpm}^{*} - 2 \lf|\nabla \gpm \ri|^2,	
		\edm
		\beq
			\label{delta density est}
			\big| {\gpm}^{*} 2 \vec\Omega \cdot \vec{L} \gpm \big| \leq \lf| \nabla \gpm \ri|^2 + \Omega^2 r^2 \lf| \gpm \ri|^2.
		\eeq
		Now since $ \Delta \gpd \leq 0 $ at any maximum point of $ \gpd $, one immediately has
		\bdm
			\lf\| \gpm \ri\|_{\infty}^2 \leq (\eps^2 \chem + \eps^2 \Omega^2)/2 \leq \tfm(1) + C \eps^{2} |\chem - \tfchem|.
		\edm
		On the other hand
		\beq
			\label{TF sup}
		\tfm(1) = \lf\| \tfm \ri\|_{\infty} = C \eps \Omega = C \eps^{-1} |\log\eps|^{-1},
		\eeq
		and the difference between the chemical potentials (see \eqref{chem} and \eqref{tfchem} for the definitions) can be estimated as follows
		\bmln{
			|\chem - \tfchem| \leq \gpe - \tfe + \eps^{-2} \lf( \lf\| \tfm \ri\|^{1/2}_{\infty} + \lf\| \gpm \ri\|_{\infty} \ri) \lf\| |\gpm|^2 - \tfm \ri\|_2 \leq		\\ 
			C \eps^{-2} \lf[ 1 + \lf\| \tfm \ri\|_{\infty}^{1/2} \lf( 1 + \eps^{3} |\log\eps| |\chem - \tfchem| \ri)^{1/2} \ri] \leq	\\
			C \eps^{-5/2} |\log\eps|^{-1/2} \lf(1 + \eps^{3/2} |\log\eps|^{1/2}  |\chem - \tfchem|^{1/2} \ri),
		}
		which yields
		\beq
			\label{diff GP chemical}
			|\chem-\tfchem| \leq C \eps^{-5/2} |\log\eps|^{-1/2}
		\eeq
		and thus the result.
	\end{proof}

A consequence of the $L^2 $ estimate in \eqref{GPmin estimates} is that the $ L^2 $ norm of $ \gpm $ is concentrated inside the support of $ \tfm $ or, in other words, the mass of $ \gpm $ inside the hole $ \ba_{\rtf} $ is small. Indeed one can easily realize that the first inequality in \eqref{GPmin estimates} implies
\beq
	\lf\| \gpm \ri\|^2_{L^2(\ba_{\rtf})} \leq \OO(\sqrt{\eps |\log\eps|}),
\eeq
since, due to the normalization of $ \tfm $ and its support in $ \tfd $,
\bmln{
 	1 = \lf\| \gpm \ri\|^2_{L^2(\B_{\rtf})} + \lf\| \gpm \ri\|^2_{L^2(\tfd)} = \lf\| \gpm \ri\|^2_{L^2(\B_{\rtf})} + \int_{\tfd} \diff \rv \lf( \lf|\gpm\ri|^2 - \tfm \ri) + 1 \geq	\\
			\lf\| \gpm \ri\|^2_{L^2(\B_{\rtf})} + 1 - \OO(\sqrt{\eps|\log\eps|)}),
} 
where in the last step the Cauchy-Schwarz inequality and the size $ \OO(\eps|\log\eps|) $ of $ \tfd $ is used.

This simple estimate can be refined (see, e.g., \cite[Proposition 2.5]{CRY1}) and one can actually show that $ \gpm $ is exponentially small in $ \eps $ inside $ \ba_{\rtf} $. The next Proposition is devoted to the proof of two pointwise inequalities of this kind:
	
	\begin{pro}[\textbf{Exponential smallness of $ \gpm $ inside the hole}]
		\label{exponential smallness}
		\mbox{}	\\
		As $ \eps \to 0 $ and for any $ \rv \in \B $,
		\beq
			\label{exp small}
			\lf| \gpm (\rv) \ri|^2 \leq C \eps^{-1} |\log\eps|^{-1} \: \exp \lf\{ - \frac{1-r^2}{1 - \rtf^2} \ri\}.
		\eeq 
		Moreover there exists a strictly positive constant $ c $ such that for any $ \OO(\eps^{7/6}) \leq r \leq \rtf - \OO(\eps^{7/6}) $,
		\beq
			\label{improved exp small}
			\lf| \gpm(\rv) \ri|^2 \leq C  \eps^{-1} |\log\eps|^{-1}  \: \exp \lf\{ - \frac{c}{\eps^{1/6}} \ri\}.
		\eeq
	\end{pro}

	\begin{rem}{\it (Comparison between \eqref{exp small} and \eqref{improved exp small})}.
		\mbox{}	\\
		At first sight the pointwise estimates in \eqref{exp small} and \eqref{improved exp small} look very similar. In fact, since $ 1 - \rtf^2 = \OO(\eps |\log\eps|) $, one can easily realize that the first one yields a much better upper bound than the latter as soon as $ 1 - r^2 \gg \OO(\eps^{5/6}) $, i.e., in particular for $ 1 - r^2 = \OO(1) $. The main drawback of the first inequality is however that it becomes much weaker and even useless if one gets closer to the radius $ \rtf $. More precisely as soon as $ 1 - r^2 \lesssim \OO(\eps) $ the bound is no longer exponentially small in $ \eps $. On the opposite the second inequality has no $ r $ dependence and a worse coefficient in the exponential function but it holds true and yields some exponential smallness up to a distance of order $ \eps^{7/6} $ from the boundary $ \partial \B_{\rtf} $. Because of this fact the second inequality will be crucial in the reduction of the original problem to another one on an annulus close to the support of the TF minimizer. The first inequality, on the contrary, will be useful in a region far from the boundary of the hole $ \partial \ba_{\rtf} $ and close to the origin, where the second bound provides a worse estimate.

		Note also that the factor in front of the exponential in both \eqref{exp small} and \eqref{improved exp small} is essentially given by the sup of $ |\gpm|^2 $ over the domain $ \ba $. In the case of the first inequality this is actually needed in order to make it meaningful on the whole of $ \ba $.
	\end{rem}

	\begin{proof}
		The starting point of the proof of \eqref{exp small} is the inequality \eqref{var eq density} together with the estimate \eqref{diff GP chemical}, which yield
		\beq
			\label{var ineq gpd}
			- \Delta \gpd \leq 4 \eps^{-2} \lf[ \ttfm(r) + C \eps^{-\frac{1}{2}} |\log\eps|^{-\frac{1}{2}} - \gpd \ri] \gpd,
		\eeq
		where we have set
		\beq
			\label{ttfm}
			\ttfm(r) : = \frac{\eps^2}{2} \lf[ \tfchem + \Omega^2 r^2 \ri],
		\eeq
		which coincides with the TF density $ \tfm $ inside $ \tfd $ and is negative everywhere else. More precisely, for any $ \rv $ such that $ r^2 \leq \rtf^2 - \eps $, one has 
		\beq
			\label{U subsolution}
			- \Delta \gpd + 2 \eps^{-3} |\log\eps|^{-2} \gpd \leq 0,
		\eeq
		since in that region
		\bdm
			\ttfm(r) + C \eps^{-\frac{1}{2}} |\log\eps|^{-\frac{1}{2}} \leq - \eps^3 \Omega^2 (1 - o(1)) \leq - \frac{1}{2 \eps |\log\eps|^2} .
		\edm
		On the other hand the function 
		\bdm
			W(r) : = \exp \lf\{ \frac{r^2 - 1}{1 - \rtf^2 - \eps} \ri\},
		\edm
		satisfies for any $ r \leq 1 $
		\bdm
			- \Delta W + 2 \eps^{-3} |\log\eps|^{-2} W \geq \frac{C}{\eps^2 |\log\eps|^2} \lf[ - \eps |\log\eps| - r^2 + \eps^{-1} \ri] W \geq 0.
		\edm
		If we then multiply $ W $ by $ \| \gpd \|_{\infty} $ we get then a supersolution to \eqref{U subsolution}, so that by the maximum principle (see, e.g., \cite[Theorem 1 at p. 508]{Evans})
		\bdm
			\lf| \gpm(\rv) \ri|^2 \leq \lf\| \gpm \ri\|^2_{\infty} W(r),
		\edm
		for any $ r \leq \rtf - \eps $. However $ W(r) $ is an increasing function and $ W(\sqrt{\rtf^2-\eps}) = \| \gpd \|_{\infty} $ by construction, which implies that the estimate trivially holds true for any $ \rv \in \B $. The estimate \eqref{exp small} is a consequence of \eqref{GPmin estimates} and the inequality $ 1 - \rtf^2 = \eps |\log\eps| \gg \eps $.

		Concerning the proof of the second inequality \eqref{improved exp small}, we first notice that \eqref{var ineq gpd} implies that, for any $ \rv \in \B_{\rtf} $ such that $ r \leq \rtf - \OO(\eps^{7/6}) $,
		\beq
			- \Delta \gpd \leq - C \eps^{-17/6} |\log\eps|^{-2} \gpd,
		\eeq
		since in that region
		\beq
			\ttfm(r) = \frac{\eps^2 \Omega^2}{2} \lf( r^2 - \rtf^2 \ri) \leq - \frac{C}{\eps^{5/6} |\log\eps|^2}.
		\eeq
		For any $ \rv_0 $ such that $ r_0 \leq \rtf - \OO(\eps^{7/6}) $ we can thus consider the annulus $ \mathcal{I} = \{ \rv \in \B \: : \: r \in [r_0 - \varrho, r_0 + \varrho] \} $ for some $ \varrho = \OO(\eps^{7/6})  $, such that $ r_0 + \varrho \leq \rtf - \OO(\eps^{7/6}) $, i.e., the outer boundary of the annulus is still at a distance of order $ \eps^{7/6} $ from $ \partial \B_{\rtf} $, and it is straightforward to verify that the function
		\beq
			\label{supersolution exp small}
			U(r) : = \lf\| \gpd \ri\|_{\infty} \exp\lf\{ \frac{ \varrho^2 - (r - r_0)^2}{\eps^{5/2}} \ri\},
		\eeq
		satisfies
		\bml{
			- \Delta U = - \lf[ \frac{4 (r-r_0)^2}{\eps^5} + \frac{2}{\eps^{5/2}} + \frac{2(r-r_0)}{\eps^{5/2} r} \ri] U(r) \geq - C \lf[ \frac{\varrho^2}{\eps^5} + \frac{1}{\eps^{5/2}} \ri] U(r) \geq	\\
			 - C \eps^{-8/3} U(r) \gg - C \eps^{-17/6} |\log\eps|^{-2} U(r),
		}
		where we have used the fact that $ |r - r_0| \leq \varrho $ inside $ \mathcal{I} $ and $ \varrho = \OO(\eps^{7/6}) $. Denoting now $ V(\rv) = \gpd(\rv) - U(r) $, we aim at proving that $ V \leq 0 $, but one has, for any $ \rv \in \mathcal{I} $,
		\bdm
			-\Delta V(\rv) \leq - C \eps^{-17/6} |\log\eps|^{-2}  V(\rv).
		\edm	
		Now at any maximum point of $ V $ in the interior of $ \mathcal{I} $ it must be $ \Delta V \leq 0 $, which implies $ V \leq 0 $ because of the above inequality. Hence it  remains only to prove that $ V \leq 0 $ at boundary $ \partial \mathcal{I} $ since the function might have a positive maximum there. However by construction $ V \leq 0 $ at $ \partial \mathcal{I} $ since $ U(r_0 - \varrho) = U(r_0 + \varrho) = \| \gpd \|_{\infty} $ and thus $ V(\rv) $ is negative everywhere. In particular $ \gpd(\rv) \leq U(r)$ for all $\rv$ such that $r=r_0$ and $ \varrho=\OO(\eps^{7/6})$, which yields \eqref{improved exp small}.
	\end{proof}

\subsection{The Giant Vortex Densities}

We now discuss some basic properties of the energy functional \eqref{hatfunct} together with those of an analogous functional where the integration is restricted to an annulus $\A=\{\rv \: : \: R_<\leq r\leq 1\}\subset\B$ with $ \rt < \rtf $ and $ |\rt - \rtf| \ll \eps|\log\eps| $. A precise choice for $ \rt$ will be made at the beginning of Section \ref{sec:auxiliary} (see \eqref{the inner radius}) but the results contained in this section apply to any $ \rt $ satisfying the above conditions.
\newline
As indicated at the end of Section \ref{sec:proof strategy} the restricted functional is for technical reasons actually more useful for the proof of the main results than the original functional. It is defined  in analogy with \eqref{hatfunct} and \eqref{hatfunct2} for {\em real} valued functions $f$ on $\A$ by
\begin{multline}
\label{hGPfa}
\hgpfa[f] : = \int_{\A} \diff \rv \left\{| \nabla f|^2+ ([\Omega]-\omega)^2 r^{-2} f^2 - 2\Omega([\Omega]-\omega) f^2 + \eps^{-2} f^4 \right\} =	\\
\int_{\A} \diff \rv \: \lf\{ \lf| \nabla f \ri|^2 - \Omega^2 r^2 f^2 + B_{\omega}^2 f^2 + \eps^{-2} f^4 \ri\},
\end{multline}
with $B_\omega(r)$ given by \eqref{newB0}.  We  use $\A$  as a subindex to label quantities associated with the annulus and thus denote the infimum of the functional \eqref{hGPfa} by 
\beq
\hgpea: = \inf_{\lf\| f \ri\|_2 = 1, f = f^*} \hgpfa[f]
\eeq
while  $\hat E^{\rm GP}_\omega$ stands for the corresponding infimum of \eqref{hatfunct}.
\newline
We shall use the short-hand notation $ \hgpf_{\star} $, $ \hgpe_{\star} $, $ g_{\star} $, etc., for quantities related either to  \eqref{hatfunct}  or \eqref{hGPfa}, e.g., a statement about $ g_{\star} $ is meant to apply to both $ \hgpm $ {\it and} $ \hgpma $.

	\begin{pro}[\textbf{Minimization of $ \hgpf_{\star} $}]
		\label{htminimization}
		\mbox{}	\\
		For any $ \omega \in \Z $ such that $ |\omega| \leq \OO(\eps^{-1}) $, the ground state energies $ \hgpe_{\star} $ satisfy the estimates
		\beq
			\label{h energy bound}
			\tfe \leq \htfe_{\omega} \leq \hgpe_{\star} \leq \htfe_{\omega} + \OO(\eps^{-2} |\log\eps|^{-1}) \leq \tfe + \OO(\eps^{-2}).
		\eeq
		The minimizers $ \hgpm $ and $ \hgpma $ exist, are radially symmetric and unique up to a global sign, which can be chosen so that they are given by positive functions solving the variational equation
		\beq
			\label{hgpm var}
			- \Delta g_{\star} + \frac{([\Omega] - \omega)^2}{r^2} g_{\star}  - 2 \Omega ([\Omega] - \omega) g_{\star} + 2 \eps^{-2} g_{\star} ^3 = \hchem_{\star} g_{\star} ,
		\eeq
		where the chemical potentials $ \hchem_{\star} $ are fixed by the $L^2$ normalization of $ g_{\star} $, i.e., $ \hchem_{\star} = \hgpe_{\star} + \eps^{-2} \| g_{\star} \|_4^4 $.
		\newline
		In addition $ g_{\star} $ are smooth and increasing and satisfy Neumann conditions at the boundary $ \partial \B $, i.e., $ \partial_r g_{\star}(1) = 0 $; $ \hgpma $ satisfies an identical condition at the inner boundary as well, i.e., $ \partial_r \hgpma(\rt) = 0 $. 
	\end{pro}

	\begin{proof}
		The lower bounds to the ground state energies are simply obtained by neglecting positive terms (the kinetic energies) in the functionals: In the case of $ \hgpf_{\omega} $ (the other case is identical), one has
		\bdm
 			\hgpf_{\omega}[f] \geq \htff[f^2] \geq \htfe \geq \tfe,
		\edm
		where we refer to the Appendix for the simple proof of the last inequality.
		\newline
		The upper bound can be easily obtained by testing the functionals on suitable regularizations of $ \sqrt{\htfm_{\omega}} $ (see \eqref{hTFm} for its definition and \cite{CY}, where such regularizations are performed on $\sqrt{\tfm}$): The main correction to the energy is due to the radial kinetic energy of the regularization and one can easily realize that this energy can be made of order $ \OO(\eps^{-2} |\log\eps|^{-2})$ times $|\log\eps|$. Indeed,  $ \sqrt{\htfm_{\omega}} $ is a monotone function going from $ 0 $ to $ \OO(\eps^{-1/2} |\log\eps|^{-1/2}) $ in an interval of size $ \OO(\eps|\log\eps|) $ and if it were smooth its kinetic energy would hence be  $ \OO(\eps^{-2} |\log\eps|^{-2})$. The exctra factor $|\log\eps|$ is due to the regularization close to $r=R_{\rm h}$. 
		\newline
		The last inequality in \eqref{h energy bound} is also discussed in the Appendix but it is basically due to the estimate 
		\beq
			\label{rmagnp inside TF}
			\big| B_{\omega}(\rv) \big| \leq \Omega \lf| r^{-1} - r \ri| + C |\omega| + \OO(1) \leq C \eps^{-1},
		\eeq
		for any $ \omega $ such that $ |\omega| \leq \OO(\eps^{-1}) $ and $ \rv \in \A $.
		\newline
		We remark that the restriction of the integration to the annulus $ \A $ has no effect because the support of the trial function can be assumed to be contained inside the support of $ \htfm_{\omega} $, i.e., the region where $ r \geq \hrtf $. On the other hand by \eqref{new radius}, $ \hrtf \geq \rtf - \OO(\eps^2|\log\eps|^2) $,  for any $ \omega \in \Z $ such that $ |\omega| \leq C \eps^{-1} $, which implies that the support of $ \htfm_{\omega} $ is contained inside $ \A $.

		Existence and uniqueness of the minimizers trivially follow from strict convexity of the functionals with respect to the density $ f^2 $, which can be made clearer by writing them as
		\beq
			\label{rewrite hGPf}
			\hgpf_{\star}[\sqrt{\rho}] = \int \diff \rv \:  \lf\{ \lf| \nabla \sqrt{\rho} \ri|^2 + ([\Omega] - \omega)^2 r^{-2} \rho + \eps^{-2} \rho^2 \ri\} - 2 \Omega ([\Omega] - \omega),
		\eeq 
		where $ \rho = f^2 $ and we have used the $L^2 $ normalization of $ f $. Similarly the radial symmetry of $ g_{\star} $ can be proven by averaging over the angular variable and  exploiting  the convexity of the functional.
		\newline
		All the other properties of $ g_{\star} $, including the variational equations \eqref{hgpm var}, are trivial consequences of the minimization: Positivity can be proven by noticing that the minimizers $ g_{\star} $ are actually ground states of suitable one-dimensional Schr\"{o}dinger operators and therefore cannot vanish except at the origin (see, e.g., \cite[Theorem 11.8]{LL}). Smoothness follows  from \eqref{hgpm var} by a simple bootstrap argument, etc.
		\newline
		The only property which requires a brief discussion is the monotonicity and we state it in a separate Lemma.
	\end{proof}

	\begin{lem}[\textbf{Monotonicity of the density}]	
		\mbox{}	\\
		Let $\rho(r)\geq 0$ be the $L^1 $-normalized minimizer of
		\bdm
			\int_R^1 \diff r \: r \bigg\{ \left|\frac{\diff\sqrt \rho}{\diff r} \ri|^2 + \frac{a}{r^2} \rho +  b\rho^2 \bigg\}
		\edm
		in $ H^1(\ba\setminus\ba_R) \cap L^2(\ba\setminus \ba_R) $ with $ a,b>0 $ and $ 0 \leq R < 1 $. Then $\rho$ is monotonously increasing in $r$.
	\end{lem}
	
	\begin{proof} 
		After a  transformation of variables $r^2\mapsto s$ and considering  $\rho$ as a function of $s$, the functional takes the form
		\bdm
			\int_{R^2}^1 \diff s \bigg\{s \left|\frac{\diff\sqrt \rho}{\diff s} \ri|^2 + \frac{a}{s}\rho + b \rho^2 \bigg\}
		\edm
		and the normalization condition is 
		\bdm
			\int_{R^2}^1 \diff s \:\rho ={\rm const}.
		\edm	
		Suppose the assertion is false. Then $\rho$ has a maximum at $s=s_1$ for some $ R \leq s_1<1$ and a local minimum at some $s_2$ with $s_1<s_2\leq 1$. For $0<\epsilon<\rho(s_1)-\rho(s_2)$ consider the set $ \mathcal{I}_{\epsilon}=\{s < s_2 \: : \: \rho(s_1)-\epsilon \leq \rho(s) \leq \rho(s_1)\}$. Then, because $\rho$ is continuous, 
		\bdm
			\Phi(\eps) : = \int_{\mathcal{I}_\epsilon} \diff s \: \rho(s) 
		\edm 
		is strictly positive and $ \Phi(\epsilon)\to 0$ as $\epsilon\to 0$. Likewise, for $\delta>0$ we consider  $ \mathcal{J}_\delta=\{s>s_1 \: : \: \rho(s_2)\leq \rho(s)\leq \rho(s_2)+\delta\}$ and the function
		\bdm
			\Gamma(\delta) : =  \int_{\mathcal{J}_\delta} \diff s \: \rho(s)
		\edm
		that has the same properties as $\Phi$. 
		\newline
		Since both $\Phi$ and $\Gamma$ are continuous, strictly positive and tend to zero when their arguments tends to zero, there exist  $\bar\epsilon,\bar\delta>0$ such that $ \rho(s_2)+\bar\delta<\rho(s_1)-\bar\epsilon $ and $\Phi(\bar\epsilon)=\Gamma(\bar\delta)$: Because $\rho$ is continuous one can, for any given $\bar\delta>0$, always find an $\bar \epsilon>0$ such that the equality 
$\Phi(\bar\epsilon)=\Gamma(\bar\delta)$ holds and $\bar\eps \to 0$ as $\bar\delta\to 0$. The inequality $\rho(s_2)+\bar\delta<\rho(s_1)-\bar\epsilon$ is fulfilled if $\bar\delta$, and hence also $\bar\epsilon$, are small enough.
		\newline
		We now define a new function $\bar \rho$ by putting 
		\bdm
			\bar \rho(s) : = \lf\{ 
			\begin{array}{lll}
				\rho(s_1)-\bar\epsilon, 	&	& \mbox{if} \:\: s \in \mathcal{I}_{\bar\epsilon},	\\
				\rho(s_2)+\bar\delta, 	& 	& \mbox{if} \:\: s \in \mathcal{J}_{\bar\delta},	\\
				\rho(s),			& 	& \mbox{otherwise.}
			\end{array}
			\right.
		\edm 
		We now compute the energy of the new density $ \bar{\rho} $. Note that it belongs by definition to the minimization domain and it is normalized in $ L^1 $ because $\Phi(\bar\epsilon)=\Gamma(\bar\delta)$. 
		\newline
		The kinetic energy of $\bar\rho$ vanishes in the intervals $\mathcal{I}_{\bar\epsilon}$ and $\mathcal{J}_{\bar\delta}$ and the potential term for $\bar\rho$ is strictly smaller than for $\rho$ because $1/s$ is strictly decreasing and the value of $\bar\rho$ on $\mathcal{I}_{\bar\epsilon}$ is larger than on $\mathcal{J}_{\bar\delta}$. Finally 
		\bdm
			\int_{R^2}^1 \diff s \: {\bar\rho}^2 < \int_{R^2}^1 \diff s \: \rho^2 
		\edm
		because when modifying $\rho$ to define $\bar\rho$, mass is moved from $\mathcal{I}_{\bar\epsilon}$ to $\mathcal{J}_{\delta}$ where the density is lower. 
		\newline
		Altogether the functional evaluated on $\bar\rho$ is strictly smaller than on $\rho$ and this contradicts the assumption that $\rho$ is a minimizer.
	\end{proof}

	As next we compare the densities $ g_{\star}^2$ with the TF density and prove exponential smallness in the hole. The analogue of Proposition \ref{GPmin} is

	\begin{pro}[\textbf{Preliminary estimates for $ g_{\star} $}]
		\mbox{}	\\
		As $ \eps \to 0 $ and for any $ \omega \in \Z $ such that $ |\omega| \leq \OO(\eps^{-1}) $,
		\beq
			\label{g estimates}
			\lf\| g_{\star}^2 - \tfm \ri\|_{L^2(\B)} = \OO(1),	\hspace{1,5cm}	\lf\| g_{\star} \ri\|_{L^{\infty}(\ba)}^2 \leq \lf\| \tfm \ri\|_{L^{\infty}(\ba)} \lf(1 + \OO(\sqrt{\eps|\log\eps|} \ri).
		\eeq
	\end{pro}

	\begin{proof}
		The proof of the $L^2$ estimate is exactly the same as for the first inequality in \eqref{GPmin estimates}. The only difference occurs in the energy remainder on the r.h.s. of \eqref{l2 difference} which can now be bounded by $ \eps^2(\hgpe_{\star} - \tfe) \leq C $ because of \eqref{h energy bound}. Here the condition on $ \omega $, i.e., $ |\omega| \leq C \eps^{-1} $, is used as in Proposition \ref{htminimization}.
		
		The sup estimate can be proven by applying the same argument used to prove the second inequality in \eqref{GPmin estimates} to the variational equations \eqref{hgpm var} solved by $ g_{\star} $, exploiting as well Neumann boundary conditions: Indeed thanks to the monotonicity of the potential $  ([\Omega]-\omega)^2 r^{-2} $, one obtains the inequality
		\bdm
			- \Delta g_{\star}^2 \leq 4 \eps^{-2} \lf( \eps^{2} \hchem_{\star} + 2\eps^2\Omega ([\Omega]-\omega) - \eps^2 ([\Omega] - \omega)^2 - 2 g_{\star}^2 \ri) g_{\star}^2,
		\edm
		which as in the proof of Proposition \ref{GPmin} implies
		\beq
			\label{apriori bound}
			\lf\| g_{\star} \ri\|^2_{\infty} \leq \frac{\eps^{2} \hchem_{\star} + 2\eps^2\Omega ([\Omega]-\omega) - \eps^2 ([\Omega] - \omega)^2}{2}.
		\eeq
		The difference between the chemical potential can be estimated as in \eqref{diff GP chemical}:	
		\beq
			\label{diff g chemical}
			\lf| \hchem_{\star} - \tfchem \ri| \leq \OO(\eps^{-5/2} |\log\eps|^{-1/2}),
		\eeq 
		so that  $ \| g_{\star} \|^2_{\infty} \leq \tfm(1) (1 + o(1)) $.
	\end{proof}

As for $ \gpm $ the above estimates imply the exponential smallness of the densities $ g_{\star} $ inside the hole. Moreover the $L^2$ estimate \eqref{GPmin estimates} can be refined to a pointwise estimate inside $ \tfd $ and one can actually prove that there is a region where the difference between $ g_{\star} $ and $ \tfm $ is much smaller than $ \tfm $:
	 
	\begin{pro}[\textbf{Exponential smallness of $ g_{\star} $ inside the hole}]
		\label{g exponential smallness}
		\mbox{}	\\
		As $ \eps \to 0 $ and for any $ \rv \in \B $
		\beq
			\label{g exp small}
			 g^2_{\star}(r) \leq C \eps^{-1} |\log\eps|^{-1} \exp \lf\{ - \frac{1-r^2}{1 - \rtf^2} \ri\}.
		\eeq 
		Moreover there exists a strictly positive constant $ c $ such that for any $ r \leq \rtf - \OO(\eps^{7/6}) $,
		\beq
			\label{g improved exp small}
			g^2_{\star}(r) \leq C \eps^{-1} |\log\eps|^{-1} \exp \lf\{ - \frac{c}{\eps^{1/6}} \ri\}.
		\eeq
	\end{pro}

	\begin{proof}
		The proof can be easily reduced to the proof of Proposition \ref{exponential smallness} by using \eqref{diff g chemical}.
	\end{proof}

	\begin{pro}[\textbf{Pointwise estimate for $ g_{\star} $}]
		\label{point GP dens}
		\mbox{}	\\
		As $ \eps \to 0 $ and for any $ \omega \in \Z $ with $ |\omega| \leq \OO(\eps^{-1}) $
		\beq
			\label{pointwise bounds}
			\lf| g_{\star}^2(r) - \tfm(r) \ri| \leq C \eps^{2} |\log\eps|^{2} (r^2 - \rtf^2)^{-3/2} \tfm(r)
		\eeq
		for any $ \rv \in \tfd $ such that $ r \geq \rtf + \OO(\eps^{3/2} |\log\eps|^2) $.
	\end{pro}
	
	\begin{proof}	
		The proof is similar to the proof of Proposition 1 in \cite{AAB} but for the sake of completeness we bring here all the details.

		The variational equation \eqref{hgpm var} can be rewritten in the following form
		\beq
			\label{var eq g}
			-\Delta g_{\star} = \frac{2}{\eps^2} \lf[ \tilde{\rho}_{\star}(r) - g_{\star}^2 \ri] g_{\star},
		\eeq
		where the function $ \tilde{\rho}_{\star} $ is given by
		\beq
			\label{def tilderho}
			\tilde{\rho}_{\star}(r) : = \half \lf( \eps^2 \hchem_{\star} + \eps^2 \Omega^2 r^2 - \eps^2 B_{\omega}^2(r) \ri).
		\eeq
		Moreover by \eqref{rmagnp inside TF} and \eqref{diff g chemical}
		\beq
			\label{tilderho}
			\lf\| \tilde{\rho}_{\star} - \tfm \ri\|_{L^{\infty}(\tfd)} \leq \eps^2 \lf| \hchem - \tfchem \ri| + \OO(1) \leq \OO(\eps^{-1/2} |\log\eps|^{-1/2}).
		\eeq
		On the other hand, for any $ \rv \in \tfd $ such that $ r - \rtf \geq \OO(\eps^{3/2} |\log\eps|^{2}) $, $ \tfm(r) \geq \OO(\eps^{-1/2}) $ and therefore
		\beq
			\label{lb hatdensity}
			\tilde{\rho}_{\star}(r) \geq \tfm(r) (1 - C |\log\eps|^{-1/2}) > C \eps^{-1/2}
		\eeq
		and in particular $ \tilde{\rho}_{\star} $ is strictly positive for such $ \rv $, which is crucial in order to apply the maximum principle. The pointwise estimates are indeed proven by providing local super- and subsolutions to \eqref{var eq g}.

		For the upper bound, we consider an interval $ [r_0 - \delta, r_0 + \delta] $, where $ \rtf + C \eps^{3/2} |\log\eps|^{2} + \delta < r_0 < 1 - \delta $ with $ \delta \ll 1 $, and the function
		\beq
			W(r) : = \sqrt{\tilde{\rho}_{\star}(r_0+\delta)} \coth \lf[ \coth^{-1} \lf( \sqrt{\frac{\tilde{\rho}_{\star}(1)}{\tilde{\rho}_{\star}(r_0 + \delta)}} \ri) + \frac{\delta^2 - \lf| r - r_0 \ri|}{3 \delta \eps} \sqrt{2\tilde{\rho}_{\star}(r_0 + \delta)} \ri].
		\eeq
		One has (see \cite[Proof of Proposition 2.1]{AS}) for any $ r \in [r_0 - \delta, r_0 + \delta] $
		\beq
			- \Delta W \geq \frac{2}{\eps^2} \lf( \tilde{\rho}_{\star}(r_0 + \delta) - W^2 \ri) W \geq  \frac{2}{\eps^2} \lf( \tilde{\rho}_{\star}(r) - W^2 \ri) W,
		\eeq
		where we have used the fact that $ \tilde{\rho}_{\star}(r) $ is an increasing function of $ r $. Moreover at the boundary of the interval $ W(r_0-\delta) = W(r_0 + \delta) = \sqrt{\tilde{\rho}_{\star}(1)} $ which is not smaller than $ g_{\star} $ thanks to the upper bound \eqref{apriori bound}, which reads $ g_{\star}^2 \leq \tilde{\rho}_{\star}(1) $. Therefore $ W(r) $ is a supersolution to \eqref{var eq g} in the interval $ [r_0 - \delta, r_0 + \delta] $ and by the maximum principle
		\beq
 			\label{pointwise ub}
 			g_{\star}(r_0) \leq W(r_0) = \sqrt{\tilde{\rho}_{\star}(r_0+\delta)} \coth \lf[ \frac{\delta}{3 \eps} \sqrt{2\tilde{\rho}_{\star}(r_0 + \delta)} \ri] 
		\eeq
		where we have used the fact that $ \coth(x) $ is a non-increasing function.
		\newline
		By the explicit expression \eqref{def tilderho} and the inequality \eqref{lb hatdensity}, one has
		\beq
			\lf| \tilde{\rho}_{\star}(r_0+\delta) - \tilde{\rho}_{\star}(r_0) \ri| \leq C \eps^2 \Omega^2 \delta,
		\eeq
		\beq
			\label{lb at r_0}
			\tilde{\rho}_{\star}(r_0) \geq \tfm(r_0)  \lf( 1 - \OO(|\log\eps|^{-1/2} \ri) \geq C \eps^2 \Omega^2 (r_0^2 - \rtf^2),
		\eeq
		so that \eqref{pointwise ub} becomes
		\beq
			\label{pointwise ub 1}
			g_{\star}(r_0) \leq \sqrt{\tilde{\rho}_{\star}(r_0)}  \lf( 1 + \frac{C \delta}{r_0^2 - \rtf^2} \ri)  \coth \lf[ \frac{\delta}{3 \eps} \sqrt{2\tilde{\rho}_{\star}(r_0 + \delta)} \ri].
		\eeq
		When the argument of $ \coth $ tends to $ \infty $, i.e., 
		\bdm
			\frac{\delta}{\eps} \sqrt{\tilde{\rho}_{\star}(r_0 + \delta)} \gg 1,
		\edm
		we can bound
		\bdm
			\coth(x) = \frac{1 + e^{-2x}} {1 - e^{-2x}} \leq (1 + C e^{-2x}),
		\edm
		and obtain the inequality
		\beq
			\label{pointwise ub 2}
			g_{\star}(r_0) \leq \sqrt{\tilde{\rho}_{\star}(r_0)} \lf( 1 + \frac{C \delta}{r_0^2 - \rtf^2} \ri) \lf( 1 + C\exp \lf\{ - \frac{2 \delta}{3 \eps} \sqrt{2\tilde{\rho}_{\star}(r_0 + \delta)} \ri\} \ri).
		\eeq
		By \eqref{lb at r_0} the second error term on the r.h.s. of the above expression is bounded from above by
		\bdm
			\exp \lf\{ - C \delta\Omega \sqrt{r_0^2 - \rtf^2} \ri\},
		\edm
		so that by taking 
		\beq
			\label{delta choice}
			\delta = C \eps^2 |\log\eps|^2 \lf(r_0^2 - \rtf^2\ri)^{-1/2},
		\eeq
		such an error can be made smaller than any power of $ \eps $, since one can choose the constant coefficient in $ \delta $ arbitrarily large. With such a choice the other error term becomes
		\bdm
			C \eps^2 |\log\eps|^2 \lf( r_0^2 - \rtf^2 \ri)^{-3/2} > C \eps^{1/2} |\log\eps|^{1/2},
		\edm
		since by definition $ r^2 - \rtf^2 \leq 1 - \rtf^2 \leq \OO(\eps|\log\eps|) $. Hence the second factor on the r.h.s. of \eqref{pointwise ub 2} can be absorbed in the above remainder. 
		\newline
		Moreover for any $ r \geq \rtf $, $ \delta \leq \eps^{5/2} |\log\eps|^{1/2} \ll \eps^{3/2} $ and one can extend the estimate to any $ r \geq \rtf + \OO(\eps^{3/2} |\log
\eps|^2) $. For the same reason the estimate applies also to the region $ [1 - 2\delta, 1] $: There one can use \eqref{apriori bound} and the fact that $ \tilde{\rho}_{\star}(1) - \tilde{\rho}_{\star}(1-2\delta) \leq C \Omega^2 \eps^2 \delta \leq C \eps^{1/2} $, which is much smaller than the error term above. 
		\newline
		The final estimate is then
		\beq
			\label{pointwise ub 3}
			g_{\star}(r) \leq \sqrt{\tilde{\rho}_{\star}(r)}  \lf[ 1 + C \eps^2 |\log\eps|^2 \lf( r^2 - \rtf^2 \ri)^{-3/2} \ri],
		\eeq
		for any $ \rtf + C \eps^{3/2} |\log\eps|^2 \leq r \leq 1 $. 
		\newline
		The next step is the replacement of $ \tilde{\rho}_{\star} $ with $ \tfm $ by means of \eqref{tilderho}, which  yields an additional remainder given by
		\bml{
			\tfm(r)^{-1} \lf|  \tilde{\rho}_{\star}(r) - \tfm(r) \ri| \leq C (\eps\Omega)^{-2} \lf( r^2 - \rtf^2 \ri)^{-1} \lf\| \tilde{\rho}_{\star} - \tfm \ri\|_{L^{\infty}(\tfd)} \leq	\\
			C \eps^{2} |\log\eps|^{2} \lf( r^2 - \rtf^2 \ri)^{-1},
		}
		which can however be absorbed in the error term in \eqref{pointwise ub 1}, since $ r^2 - \rtf^2 \leq \eps |\log\eps| $.		

		In order to prove a corresponding lower bound, we fix some $ r_0 $ and $ \delta $ as before, i.e., such that $ \rtf + C \eps^{3/2} |\log\eps|^2 + \delta < r_0 < 1 - \delta $ and $ \delta \ll 1 $. Since $ \tilde{\rho}_{\star} $ is an increasing function of $ r $ and $  g_{\star} $ is positive
		\beq
			\label{g var inequation}
			-\Delta g_{\star} \geq \frac{2}{\eps^2} \lf[ \tilde{\rho}_{\star}(r_0-\delta) - g_{\star}^2 \ri] g_{\star}.	
		\eeq
		Then we denote by $ h(r) $ the function solving for $ \rv \in \B $
		\bdm
			-\Delta h = \tilde{\eps}^{-2} \lf( 1 - h^2 \ri) h,
		\edm
		with Dirichlet boundary condition $ h(1) = 0 $ and $ \tilde{\eps} \to 0 $. In \cite{Se} it was proven that $ h $ satisfies the bound
		\bdm
			1 - c \exp \lf\{ - \frac{1 - r^2}{2 \tilde{\eps}} \ri\} \leq h(r) \leq 1.
		\edm
		If we now set 
		\beq
			\tilde{h}(r) : = \sqrt{\tilde{\rho}_{\star}(r_0-\delta)} \: h\lf(\frac{|r-r_0|}{\delta}\ri),
		\eeq 
		\beq
			\label{tilde eps}
			\tilde{\eps} : = \frac{\eps}{\delta \sqrt{2 \tilde{\rho}_{\star}(r_0-\delta)}},
		\eeq	
		then $ \tilde{h} $ solves in $ [r_0-\delta,r_0+\delta] $ the equation
		\bdm
			-\Delta \tilde{h} = \frac{2}{\eps^2} \lf[ \tilde{\rho}_{\star}(r_0-\delta) - \tilde{h}^2 \ri] \tilde{h},
		\edm
		with Dirichlet conditions at the boundary $ r = r_0 \pm \delta $. Thanks to \eqref{g var inequation}, $ g_{\star} $ is a supersolution for the same problem, so that by the maximum principle $ g_{\star}(r) \geq \tilde{h}(r) $ inside the interval and in particular
		\beq
			\label{pointwise lb}
			g_{\star}(r_0) \geq \tilde{h}(r_0) \geq \sqrt{\tilde{\rho}_{\star}(r_0-\delta)} \lf[ 1 - c \exp \lf\{ - \frac{1}{2 \tilde{\eps}} \ri\} \ri],
		\eeq
		for any $  \rtf + C \eps^{3/2} |\log\eps|^2 + \delta < r_0 < 1 - \delta $. 
		\newline
		Note that by choosing $ \delta $ as in \eqref{delta choice} the remainder in the above expression can be made smaller than any power of $ \eps $. However the estimate of $ \tilde{\rho}_{\star}(r-\delta) $ in terms of $ \tilde{\rho}_{\star}(r) $ provides the same remainder as in the upper bound proof, i.e.,
		\bdm
			g_{\star}(r) \geq \sqrt{\tilde{\rho}_{\star}(r)}  \lf[ 1 - C \eps^2 |\log\eps|^2 \lf( r^2 - \rtf^2 \ri)^{-3/2} \ri].
		\edm
		The extension of the estimate to the whole region $ [\rtf + \OO(\eps^{3/2}|\log\eps|^2), 1] $ as well as the replacement of $ \tilde{\rho}_{\star} $ with $ \tfm $ can be done exactly as in the upper bound and the remainders included in the above error term.
		\end{proof}
	
We conclude the section with an useful estimate of the gradient of the densities $ g_{\star} $:

	\begin{pro}[\textbf{Gradient estimate for $ g_{\star} $}]
		\mbox{}	\\
		As $ \eps \to 0 $ and for any $ R > \rtf + \OO(\eps^{3/2} |\log\eps|^2) $, one has
		\beq
			\label{gen grad est}
			\lf\| \nabla  g_{\star} \ri\|_{L^{\infty}(\ba\setminus\ba_R)} \leq  C \eps^{-1} \lf( R^2 - \rtf^2 \ri)^{-1/4} \lf\| g_{\star} \ri\|_{L^{\infty}(\ba)},
		\eeq
		which inside $ \at = \{\rv \: :\: \rtf +\eps|\log\eps|^{-1}\leq r\leq 1\} $ becomes
		\beq
			\label{grad est}
			\lf\| \nabla  g_{\star} \ri\|_{L^{\infty}(\at)} \leq  C \eps^{-7/4} |\log\eps|^{-3/4}.
		\eeq	
	\end{pro}

	\begin{proof}
		We first exploit the fact that any $ g_{\star} $ is radial to rewrite \eqref{hgpm var} as
		\beq
			\label{var eq g radial}
			- g^{\prime\prime}_{\star} - r^{-1} g_{\star}^{\prime} + ([\Omega] - \omega)^2 r^{-2} g_{\star}  + 2 \eps^{-2} g_{\star} ^3 = \lf( \hchem_{\star} + 2 \Omega [\Omega - \omega] \ri) g_{\star}
		\eeq
		and take the $ L^{\infty} $ norm inside $ \ba\setminus\ba_R $, for some $ R > \rtf + \OO(\eps^{3/2} |\log\eps|^2) $, of both sides to obtain (see \eqref{def tilderho})
		\bml{
			\lf\| g^{\prime\prime}_{\star} \ri\|_{L^{\infty}(\ba\setminus\ba_R)} \leq C \lf\| r^{-1} g_{\star}^{\prime} \ri\|_{L^{\infty}(\ba\setminus\ba_R)} + C \eps^{-2} \lf\| (\tilde{\rho}_{\star} - g_{\star}^2) g_{\star} \ri\|_{L^{\infty}(\ba\setminus\ba_R)} \leq \\ 
 			C \lf[ \lf\| g_{\star}^{\prime} \ri\|_{L^{\infty}(\ba\setminus\ba_R)} + \eps^{-2} \lf( R^2 - \rtf^2\ri)^{-1/2} \lf\| g_{\star} \ri\|_{L^{\infty}(\ba)} \ri],
		}
		by \eqref{tilderho} and the pointwise estimate \eqref{pointwise bounds}:
		\bmln{
			\lf\| \tilde{\rho}_{\star} - g_{\star}^2 \ri\|_{L^{\infty}(\ba\setminus\ba_R)} \leq \lf\| \tfm - g_{\star}^2 \ri\|_{L^{\infty}(\ba\setminus\ba_R)} + \OO(\eps^{-1/2}|\log\eps|^{-1/2}) \leq	\\
			C \lf( R^2 - \rtf^2 \ri)^{-1/2} + \OO(\eps^{-1/2}|\log\eps|^{-1/2}) \leq C \lf( R^2 - \rtf^2 \ri)^{-1/2},
		}
		since $ R^2 - \rtf^2 \leq C \eps |\log\eps| $.
		\newline
		On the other hand by the Gagliardo-Nirenberg inequality (see, e.g., \cite[Theorem at p. 125]{N})
		\bml{
			\lf\| g^{\prime}_{\star} \ri\|_{L^{\infty}(\ba\setminus\ba_R)} \leq C \lf\| g^{\prime\prime}_{\star} \ri\|_{L^{\infty}(\ba\setminus\ba_R)}^{1/2} \lf\| g_{\star} \ri\|_{L^{\infty}(\ba)}^{1/2} \leq	\\
			C \lf[ \eps^{-1}  \lf( R^2 - \rtf^2 \ri)^{-1/4} \lf\| g_{\star} \ri\|_{L^{\infty}(\ba)} + \lf\| g^{\prime}_{\star} \ri\|_{L^{\infty}(\ba\setminus\ba_R)}^{1/2} \lf\| g_{\star} \ri\|_{L^{\infty}(\ba)}^{1/2} \ri],
		}
		which implies the result.
	\end{proof}

\section{Reduction to an Auxiliary Problem on an Annulus}
\label{sec:auxiliary}

The first step towards the proof of the main results is the reduction of the original GP energy functional to an analogous functional on a suitable annulus. The main ingredients of such a reduction are, on the one hand, the exponential smallness of the GP minimizers proven in the last section (see, e.g., \eqref{exp small}, \eqref{improved exp small}, \eqref{g exp small}, etc.), which intuitively implies that all the $L^2$ mass and therefore the energy are concentrated in a annulus close to the boundary, and on the other a decoupling of the GP energy functional, which allows the extraction of the giant vortex energies introduced in \eqref{hatfunct}.
\newline
The second part of the section is devoted to the discussion of the optimal giant vortex phase: Whereas both the energy splitting and upper bound hold true for any reasonable phase $ \omega $, they are useful only for specific choices of the phase. As we are going to see, the optimal phase $ \omega_0 $ can be defined as the minimizer of a suitable coupled problem in the annulus.
\newline
We also discuss the existence of the analogous minimizer associated with the original problem in the ball $ \B $, i.e., the phase $ \oopt $ occurring in Theorems \ref{theo:energy} and \ref{theo:degree}, as well as some relevant properties of it.

\subsection{Energy Decoupling}

Before stating the main result of this section, we first recall and introduce some notation.
\newline
The main object of the reduction is an annulus 
\beq
	\label{the annulus}
	\A : = \lf\{ \rv  \: : \: \rt \leq r \leq 1 \ri\},
\eeq 
where the inner radius $ \rt < \rtf $ has to be chosen in a proper way so that two conditions are simultaneously fulfilled: The radius $ \rt $ should be sufficiently far from $ \rtf $ in such a way that the exponential smallness proven in Propositions \ref{exponential smallness} applies inside the complement of $ \A $. At the same time, $ \rt $ must not be too far from $ \rtf $; in fact we shall in Section \ref{Sect est reduced energy} need that $|\rt - \rtf | \ll \eps|\log\eps|^{-1} $ (see \eqref{BL 1} and the use of \eqref{Fbound2} in the proof of Proposition \ref{pro:refinedbound}).
\newline
All these conditions are satisfied if one chooses
\beq
	\label{the inner radius}
	\rt : = \rtf - \eps^{8/7}.
\eeq
The apparently strange power of $ \eps $ occurring in the above expression is essentially motivated by the pointwise estimates \eqref{improved exp small} and \eqref{g improved exp small}, which hold true for any $ r \leq \rtf - \OO(\eps^{7/6}) $: In the definition of the inner radius of the annulus we could have chosen any $ \rt $ satisfying such a condition in addition to $|\rt - \rtf | \ll \eps|\log\eps|^{-1} $. In particular any remainder given by a power of $ \eps $ smaller than $ 7/6 $ but larger than $1$ would have been all right. However for the sake of simplicity we make an explicit choice among the allowed remainders and pick $ \eps^{8/7}$.

The auxiliary problem on the annulus is associated with the estimate of the energy functional
\beq
	\label{reduced energy}
	\E_{\A,\omega}[v] := \int_{\A} \diff \rv \: \hgpma^2 \lf\{ \lf| \nabla v \ri|^2 - 2 \rmagnp \cdot \lf( i v, \nabla v \ri) + \eps^{-2} \hgpma^2 \lf( 1 - |v|^2 \ri)^2 \ri\},
\eeq
where $ \rmagnp $ is defined in \eqref{vec B}. According to the convention used in the rest of the paper, $ \E_{\omega} $ without the label $ \A $ stands for the same energy as above but with the integral extended to the whole of $ \ba $.

A key ingredient in the proof of the GP energy asymptotics is a lower estimate of \eqref{reduced energy}. In fact in the next Proposition, which is the main result proven in this section, we show that the energy \eqref{reduced energy} provides a lower bound to $E^{\rm {GP}} - \hat{E}^{\rm {GP}}_{\A,\om}$ when $ v $ is suitably linked to $ \gpm $:

	\begin{pro}[\textbf{Reduction to an annulus}]
		\label{reduction}
		\mbox{}	\\
		For any $ \omega \in \Z $ such that $ |\omega|\leq \OO(\eps^{-1}) $ and for $ \eps $ sufficiently small
		\beq
			\label{reduction est} 
			\hgpea + \E_{\A,\omega}[u_{\omega}] - \OO(\eps^{\infty}) \leq \gpe \leq \hgpea + \OO(\eps^{\infty}),
		\eeq
		where the function $ u_{\omega} \in H^1(\A) $ is given by the decomposition
		\beq
			\label{function u}
			\gpm(\rv) = : \hgpma(r) u_{\omega}(\rv) \exp\lf\{ i([\Omega] - \omega) \vartheta\ri\}.
		\eeq
	\end{pro}

	\begin{proof}
		The proof of \eqref{reduction est} is done by proving suitable upper and lower bounds to the GP ground state energy.
		
		The upper bound is obtained by testing $ \gpf $ on a trial function of the form $ \tilde{g}(r) \exp\{ i ([\Omega] - \omega) \vartheta\} $, where $ \tilde{g} $ is an appropriate regularization of $ \hgpma $, and using the definition \eqref{hGPfa}. Since $ \hgpma $ does not vanish at $ \rt $, it is not in the minimization domain of $ \gpf $ and one has to regularize it at the boundary $ \partial \rt $, e.g., taking 
		\beq
			\tilde{g}(r) : = c
			\lf\{
			\begin{array}{lll}
				\hgpma(r)							&	\mbox{}	&	\mbox{if} \:\:\: \rv \in \A,	\\
				\eps^{-n} \hgpma(\rt) (\rt - \eps^{n} - r)		&	\mbox{}	&	\mbox{if} \:\:\: \rt - \eps^{n} \leq r \leq \rt, \\
				0								&	\mbox{}	&	\mbox{if} 	\:\:\:	r \leq \rt - \eps^{n},
			\end{array}
			\ri.
		\eeq
where $ c  $ is a normalization constant and $ n $ some arbitrary power greater than $ 0 $. Thanks to the exponential smallness \eqref{g improved exp small}, which implies $ \hgpma(\rt) = \OO(\eps^{\infty}) $, the normalization constant is $ c = 1 - \OO(\eps^{\infty}) $ and the energy of $ \tilde{g} $ in $ [\rt - \eps^n, \rt] $ is exponentially small as well. The upper bound trivially follows.

		The lower bound is mainly a consequence of a classical result of energy decoupling, which has been already used in different contexts, e.g., in \cite{LM}. 
		\newline
		The starting point is however a reduction to the annulus $ \A $ of the GP energy functional: Exploiting the exponential smallness \eqref{improved exp small} of $ \gpm $, it is very easy to show that
		\beq
			\label{annulus energy}
			\gpe = \gpf \lf[\gpm\ri] \geq \gpf_{\A} \lf[\gpm\ri] - \OO(\eps^{\infty}),	
		\eeq
		where we have used the symbol $ \gpf_{\A} $ to denote the restriction of the integration to $ \A $. Note that in the above inequality one can avoid the estimate of the gradient of $ \gpm $ by simply rewriting the functional as in \eqref{GPf} and neglecting all the positive terms given by the integration over $ \ba \setminus \A $; the only negative term (centrifugal energy) can then be estimated by means of \eqref{improved exp small}. 
		\newline
		Thanks to the exponential smallness, one also has that the $L^2$ mass of $ \gpm $ outside $ \A $ is very small, i.e.,
		\beq
			\label{annulus mass}
			\lf\| \gpm \ri\|_{L^2(\A)} \geq 1 - \OO(\eps^{\infty}).
		\eeq
		Now inside $ \A $ the decomposition \eqref{function u} is well defined and therefore one can calculate
		\bmln{
			\gpf_{\A}\lf[\gpm\ri] = \E_{\A,\omega}[u_{\omega}] + \int_{\A} \diff \rv \lf\{ \lf| u_{\omega} \ri|^2 \lf| \nabla \hgpma \ri|^2 +  \hgpma \nabla \hgpma \cdot \nabla \lf| u_{\omega} \ri|^2 + \lf( [\Omega]- \omega \ri)^2 r^{-2} \hgpma^2 \lf| u_{\omega} \ri|^2 - 	\ri.	\\
			\lf. 2\Omega\lf([\Omega] - \omega\ri) \hgpma^2 \lf| u_{\omega} \ri|^2 + \eps^{-2} \hgpma^4 \lf( 2 \lf|u_{\omega}\ri|^2 - 1 \ri) \ri\}.
		}
		A simple integration by parts then yields
		\bdm
 			 \int_{\A} \diff \rv \: \hgpma \nabla \hgpma \cdot \nabla \lf| u_{\omega} \ri|^2 = - \int_{\A} \diff \rv \: \lf\{ \lf|u_{\omega} \ri|^2 \lf| \nabla \hgpma \ri|^2 + \hgpma  \lf| u_{\omega} \ri|^2 \Delta \hgpma \ri\},
		\edm
		since the boundary terms vanish because of the Neumann conditions satisfied by $ \hgpma $ at the boundaries $ \partial \ba_{\rt} $ and $ \partial \ba $. Then one can replace in the above expression $ \Delta \hgpma $ by means of the variational equation \eqref{hgpm var} and the result is
		\bml{
			\gpf_{\A}\lf[\gpm\ri] = \E_{\A,\omega}[u_{\omega}] + \hchema \int_{\A} \diff \rv \: \hgpma^2 \lf| u_{\omega} \ri|^2 - \eps^{-2} \int_{\A} \diff \rv \: \hgpma^4 =	\\
			\E_{\A,\omega}[u_{\omega}] + \hchema \lf\| \gpm \ri\|^2_{L^2(\A)} - \eps^{-2} \lf\| \hgpma \ri\|^4_{L^4(\A)} \geq \E_{\A,\omega}[u_{\omega}] + \hgpea - \OO(\eps^{\infty}),
		}
		by \eqref{annulus mass} and the definition of the chemical potential $ \hchema $ (see Proposition \ref{htminimization}).
	\end{proof}

\subsection{Optimal Phases and Densities}

The idea behind the decomposition \eqref{function u} is that, after the extraction from $ \gpm $ of a density $ \hgpm $ and a giant vortex phase, i.e., the phase factor $ \exp \{ i([\Omega] - \omega) \vartheta \}$, what is left is a function $ u_{\omega} $, which contains all the remaining vorticity in $ \gpm $. Therefore in the giant vortex regime, one would like to prove that $ |u_{\omega}| \sim 1 $ inside a suitable annulus at the boundary of the trap, and the key tool to proving such a behavior is a detailed analysis of the reduced energy $ \E_{\A,\omega} $.
\newline
However, in order to prove such a result, both the phase and the associated density $ \hgpma $ have to be chosen in an appropriate way: The result stated in Proposition \ref{reduction} is basically independent of $ \omega $, i.e., it applies to any reasonable giant vortex phase $ \omega $. Nevertheless the leading order term in the GP energy asymptotics is given in \eqref{reduction est} by $ \hgpea $ which depends in crucial way on $ \omega $ and it is clear that in order to extract some delicate information about $u_\omega$ like the absence of zeros, the estimate (upper bound) of the reduced energy $ \E_{\A,\omega}[u_{\omega}] $ through \eqref{reduction est} has to be very precise. This leads to the definition of a giant vortex {\it optimal phase} (associated with the annulus $ \A $), which is nothing but the minimizer of the energy $ \hgpea $ with respect to $ \omega $.
\newline
In the next Proposition we show that there exists at least one $ \omega_0 $ minimizing $ \hgpea $ as well as some properties which are going to be crucial in the rest of the proof.

	\begin{pro}[\textbf{Properties of the optimal phase $ \omega_0 $ and density $ g_{\A,\omega_0} $}]
		\label{optimal phase pro}
		\mbox{}	\\
		For every $ \eps > 0 $ there exists an $ \omega_0 \in \Z $ minimizing $ \hgpea $. It satisfies
		\beq
			\label{est omega_0}
			\omega_0 = \frac{2}{3\sqrt{\pi} \eps} \lf(1 + \OO(|\log\eps|^{-1/2}) \ri).
		\eeq
		Moreover the minimizer $ g_{\A,\omega_0} $ of $ \hgpf_{\A,\omega_0} $ satisfies the bound
		\beq
			\label{compatibility}
			\int_{\A} \diff \rv \: g_{\A,\omega_0}^2 \lf( \Omega - \frac{[\Omega] - \omega_0}{r^2} \ri) =  \OO(1).
		\eeq
	\end{pro}

	\begin{proof}
		The existence of a minimizing $ \omega_0 \in \Z $ can be easily proven by noticing that $ \lim_{\omega \to \pm \infty} \hgpea = + \infty $ since by \eqref{hGPfa}
		\bdm
			\hgpea \geq \lf( [\Omega] - \omega \ri)^2 - 2 \Omega \lf( [\Omega] - \omega \ri),
		\edm
		which implies that, for given $ \eps > 0 $, only finitely many $ \omega \in \Z $ can minimize the energy.

		The main ingredient for the proof of \eqref{est omega_0} is the energy bound \eqref{h energy bound}, which implies, for any $ \omega \in \Z $ such that $ |\omega| \leq C\eps^{-1} $,
		\beq
			\label{energy inequalities}
			\htfe_{\omega_0} \leq \hgpe_{\A,\omega_0} \leq \hgpe_{\A,\omega} \leq \htfe_{\omega} + \OO(\eps^{-2}|\log\eps|^{-1}),
		\eeq
		by definition of $ \omega_0 $. Choosing now $ \omega = [2/(3\sqrt{\pi}\eps)] $ in the r.h.s. of the above expression and using \eqref{hTFe} for $ \htfe_{\omega} $, we get the inequality
		\beq
			\tfe + \lf[ \omega_0 - \frac{2}{3 \sqrt{\pi} \eps} \ri]^2 + \frac{2}{9\pi \eps^2} - 2 \Omega \left( [\Omega] - \Omega \right) \leq \htfe_{\omega_0} 	\leq \tfe + \frac{2}{9\pi \eps^2} -2 \Omega \left( [\Omega] - \Omega \right) + C\eps^{-2} |\log\eps|^{-1}
		\eeq
		which yields the result.
		
		We now prove \eqref{compatibility}. A simple estimate yields
		\bml{
 			\label{Taylor exp}
 			\hgpf_{\A,\omega_0 \pm 1}\lf[g_{\A,\omega_0}\ri] = \hgpf_{\A,\omega_0}\lf[g_{\A,\omega_0}\ri] \pm 2 \int_{\A} \diff \rv \: \lf[ \Omega - \lf( [\Omega] - \omega_0 \ri) r^{-2} \ri] g_{\A,\omega_0}^2 + \int_{\A} \diff \rv \: r^{-2} g_{\A,\omega_0}^2 \leq	\\
			\hgpe_{\A,\omega_0} \pm 2 \int_{\A} \diff \rv \: \lf[ \Omega - \lf( [\Omega] - \omega_0 \ri) r^{-2} \ri] g_{\A,\omega_0}^2 + \rt^{-2}.
		}
		Now suppose that 
		\bdm
			\lf| \int_{\A} \diff \rv \: \lf[ \Omega - \lf( [\Omega] - \omega_0 \ri) r^{-2} \ri] g_{\A,\omega_0}^2 \ri| > \frac{2}{3},
		\edm
		then, since $ \rt = 1 - o(1) $, \eqref{Taylor exp} would imply that there exists some $ \bar\omega \in \Z $ equal to $ \omega_0 \pm 1 $ such that
		\bdm
			\hgpe_{\A,\bar\omega} \leq \hgpf_{\A,\bar\omega}\lf[g_{\A,\omega_0}\ri] < \hgpe_{\A,\omega_0},
		\edm
		which contradicts the fact that $ \omega_0 $ minimizes $ \hgpe_{\A,\omega} $. Note that the above argument implies that the constant on the r.h.s. of \eqref{compatibility} is actually smaller than 1, i.e.,
		\beq
			\label{compatibility2}
			\lf| \int_{\A} \diff \rv \: \lf[ \Omega - \lf( [\Omega] - \omega_0 \ri) r^{-2} \ri] g_{\A,\omega_0}^2 \ri| < 1.
		\eeq
\end{proof}

As anticipated above and in the Introduction, there is another giant vortex phase that naturally emerges in the study of the above functionals, i.e., the one associated with the energy in the whole ball $ \B $: More precisely it is defined as the minimizer of 
\bdm
	\hgpe_{\omega} : = \inf_{\lf\| f \ri\|_2 = 1} \hgpf_{\om}[f]
\edm
with respect to $ \omega \in \Z $ and denoted by $ \oopt $:

\begin{pro}[\textbf{Optimal phase $ \oopt $}]
		\label{optimal phase omega_{opt} pro}
		\mbox{}	\\
		For every $ \eps > 0 $ there exists an $ \oopt \in \Z $ minimizing $ \hgpe_{\omega} $. It satisfies
		\beq
			\label{est omega_{opt}}
			\oopt = \frac{2}{3\sqrt{\pi} \eps} \lf(1 + \OO(|\log\eps|^{-1/2}) \ri).
		\eeq
\end{pro}

\begin{proof}
	The proof is basically identical to the one of \eqref{est omega_0}: For instance it is  sufficient to replace \eqref{energy inequalities} with the corresponding version in $ \B $ (see \eqref{h energy bound}).
\end{proof}

\section{Estimates of the Reduced Energy}
\label{Sect est reduced energy}

In this section we study the auxiliary problem introduced in Section \ref{sec:auxiliary}. From now on we shall simplify notation by dropping some subscripts:
\begin{equation*}
g_{\A,\om_0} := g ,
\end{equation*}
also
\begin{equation}\label{defi B}
\vec{B}(r) := \vec{B}_{\omega_0}(r) = \lf( \Omega r - \frac{[\Omega] - \omega_0}{r} \ri) \vec{e}_{\vartheta},
\end{equation}
and 
\[
 \E [v] := \int_{\ann} \diff \rv \lf\{ g ^{2} \left| \nabla v \right|^2 - 2 g ^2 \vec{B} \cdot (iv,\nabla v) + \frac{g ^4 }{\ep ^2} \left(1-|v|^2 \right)^2 \ri\}.
\]
The following energy is crucial in our analysis: 
\begin{equation}\label{defiF}
\F [v]:= \int_{\A} \diff \rv \lf\{ g ^{2} \left| \nabla v \right|^2 + \frac{g ^4 }{\ep ^2} \left(1-|v|^2 \right)^2 \ri\}.
\end{equation}
We recall that $u :=u_{\om_0}$ is defined by
\begin{equation}\label{defi u}
\gpm(\rv) = : g(r) u(\rv) \exp\lf\{ i([\Omega] - \omega_0) \vartheta\ri\}
\end{equation}
and that from Proposition \ref{reduction} we have
\begin{equation}\label{starting bound}
 \E [u] \leq \OO (\ep ^{\infty}).
\end{equation}
We also need to define a reduced annulus
\begin{equation}\label{defi annt}
 \at := \left\lbrace \rv  \: : \: \rd \leq r \leq 1 \right\rbrace
\end{equation}
with
\begin{equation}\label{defi R large}
\rd := \rtf + \ep |\log \ep|^{-1}. 
\end{equation}
An important point is that from \eqref{pointwise bounds} we have the lower bound
\begin{equation}\label{g low bound}
 g^2(r) \geq \frac{C}{\ep |\log \ep|^{3}} \mbox{ on } \at.
\end{equation}

The main result of this section is
 
\begin{pro}[\textbf{Bounds for the reduced energies}]\label{pro:refinedbound}
	\mbox{}	\\
	Let $u$ be defined by (\ref{defi u}). If $\Om_{0}> 2(3 \pi)^{-1} $ we have for $\ep$ small enough
		\begin{eqnarray}\label{borne Fg final}
			\Fg [u] \leq C \frac{|\log \ep|^{5/2}}{\ep ^{1/2} (\log |\log \ep|)^2} 
			\\ \Eg [u] \geq -C \frac{|\log \ep|^{3/2}}{\ep ^{1/2} \log |\log \ep|} \label{borne Eg final}.
		\end{eqnarray}
\end{pro}

Since the proof of Proposition \ref{pro:refinedbound} is rather involved we sketch the main ideas before going into the details.\\
It is useful to introduce a potential function $F$ defined as follows: 
\begin{equation}\label{F}
F(r):= 2 \int_{R_{<}} ^r \diff s \: g ^2  (s) \left(\Om  s- \left([\Om] - \om_0 \right)\frac{1}{s}\right) = 2 \int_{R_{<}}^r \diff s \: g ^2(s) \vec{B}(s) \cdot \vec{e}_{\vartheta}.
\end{equation}
We have
\begin{equation}\label{F2}
\nabla ^{\perp} F  = 2 g ^2 \vec{B}, \hspace{1,5cm} F(R_{<}) = 0,
\end{equation}  
i.e., $F$ is the ``primitive'' of $ 2 g^2 \vec{B}$ vanishing at $R_{<}$. We refer to Subsection \ref{subsec:preliminaries} for further properties of $F$. Integrating by parts we have
\begin{equation}\label{integ part}
	\int_{\ann} \diff \rv \lf\{ g ^{2} \left| \nabla u \right|^2 - 2 g ^2 \vec{B} \cdot (iu,\nabla u) \ri\} = \int_{\ann} \diff \rv \left\{ g ^{2} \left| \nabla u \right|^2 + F \curl (iu,\nabla u) \right\} - \int_{\partial \B} \diff \sigma  \: F(1) (iu , \partial_{\tau}u) 
\end{equation} 
and thus the energy $\Eg [u]$ can be rewritten as follows
\begin{equation}\label{integ energy}
\Eg [u] = \int_{\ann} \diff \rv \: g ^{2} \left| \nabla u \right|^2 + \int_{\ann} \diff \rv \: F \curl (iu,\nabla u)  - \int_{\partial \B} \diff \sigma \: F(1) (iu , \partial_{\tau}u) + \int_{\ann} \diff \rv \: \frac{g^4}{\ep ^2}\left( 1-|u|^2 \right)^2.
\end{equation}
It is in this form that the energy is best bounded from below.\\
The boundary term (third term in \eqref{integ energy}) is estimated in the following way: The property of $\om_0$ given in \eqref{compatibility} implies that $F(1)=\OO (1)$. We combine this fact with an estimate of the circulation of $u$ on the boundary of the unit ball that we provide in Subsection \ref{subsec:boundary}. Proving this estimate requires to derive a PDE satisfied by $u$ and use it in much the same way as in the proof of the Pohozaev identity \cite{P}.\\
The first two terms can be estimated in terms of the vorticity of $u$. Indeed, suppose that $|u|\sim 1$ except in some balls (that we identify with vortices) whose radii are much smaller than the width of $\A$. Let us denote these balls by $ \lf\{ \B(\avj,t) \ri\}_{j\in J}$ with $ J \subset \N $ and $t\ll \ep |\log \ep|$, i.e., much smaller than the width of the annulus. Then by Stokes theorem, if the degree of $u$ around $\avj$ is $d_j$ (we systematically neglect remainder terms in this sketch)
\begin{equation}\label{vort heur 2}
\int_{\ann} \diff \rv \: F \curl (iu,\nabla u)\simeq \sum_{j \in J} 2 \pi F(a_j) d_j.    
\end{equation}
Minimizing the sum of the first and the last term with respect to $t$ yields $t\propto \ep^{3/2} |\log \ep|^{1/2}$ (see also \cite[p. 6]{CY} for heuristics about the optimal size of the vortex core), which implies an estimate of the form  
\begin{equation}\label{vort heur 1}
\int_{\ann} \diff \rv \: g ^{2} \left| \nabla u \right|^2 \gtrapprox \sum_{j \in J} 2\pi g^2 (a_j) |d_j| \log \left(\frac{\ep |\log \ep|}{t}\right)  \gtrapprox \sum_{j \in J} \pi g^2 (a_j) |d_j| |\log \ep |. 
\end{equation}
In Subsections \ref{vortex balls sec} and \ref{jacobian estimate sec} we give a rigorous version of this heuristic analysis. The main tools were originally introduced in the context of Ginzburg-Landau (GL) theory (we refer to \cite{BBH2,SS} and references therein). The method of growth and merging of vortex balls introduced independently by Sandier \cite{Sa} and Jerrard \cite{J} provides a lower bound of the form \eqref{vort heur 1} (see Subsection \ref{vortex balls sec}). On the other hand, the jacobian estimate of Jerrard and Soner \cite{JS} allow to deduce in Subsection \ref{jacobian estimate sec} that 
\begin{equation}\label{vort heur 3}
 \curl (iu,\nabla u) \simeq \sum_{j \in J} 2 \pi  d_j \delta(\rv - \avj),
\end{equation}
i.e., the vorticity measure of $u$ is close to a sum of Dirac masses. This implies an estimate of the form \eqref{vort heur 2}. 
\newline
Having performed this analysis, the role of the critical velocity becomes clear: We have essentially (note that the boundary term is negligible in a first approach)
\begin{equation}\label{vort heur 4}
 \E [u] \gtrapprox \sum_{j \in J} 2 \pi  |d_j| \left( \frac{1}{2} g^2 (a_j) |\log \ep | +  F(a_j) \right). 
\end{equation}
If $\Om_0> 2(3\pi) ^{-1}$ the sum in the parenthesis is positive for any $\avj$ in the bulk (see the Appendix), which means that vortices become energetically unfavorable. 
\\
So far we have assumed that the zeros of $u$ were isolated in small `vortex' balls. One can show that this
is indeed the case by exploiting upper bounds to $\F [u]$ (this energy controls, via the co-area formula the size of the set where $|u|$ is no close to $1$). We derive a first bound from our a priori estimate on $\E [u]$ \eqref{starting bound} in Subsection \ref{subsec:preliminaries}. We emphasize however that this first bound is not strong enough to construct vortex balls in the whole annulus $\A$. To get around this point we split in Subsection \ref{vortex balls sec} the annulus into cells and distinguish between two type of cells. In `good' cells we have the proper control and perform locally the vortex balls construction. On the other hand we are able to show that there are relatively few `bad' cells where the construction is not possible. Thus the analysis in the good cells allows to get lower bounds for $\E [u]$. These bounds can in turn be used to improve our control on $\F [u]$ and reduce the number of bad cells. The analysis can then be repeated, but now on a larger set. Finally we see that there are no bad cells. Such an induction process is `hidden' at the end of the proof of Proposition \ref{pro:refinedbound} in Subsection \ref{proof completion sec} (see in particular the discussion after \eqref{lowbound 5}).

\subsection{Preliminaries}
\label{subsec:preliminaries}

We first give some elementary estimates on $\vec{B}$ and $F$ that we need in our analysis:

\begin{lem}[\textbf{Useful properties of $\vec{B}$ and $F$}]\label{propertiesBF}\mbox{}\\
Let $\vec{B}$ and $F$ be defined in \eqref{defi B} and \eqref{F} respectively. We have
\begin{eqnarray}
\lf\| B \ri\|_{L^{\infty} (\A)} &\leq& C \ep ^{-1}, \label{Bbound}
\\ \lf\| F \ri\|_{L^{\infty} (\A)} &\leq& C \ep ^{-1}, \label{Fbound}
\\ \lf\| \nabla F \ri\|_{L^{\infty} (\A)} &\leq& C \ep ^{-2} |\log \ep|^{-1},\label{gradFbound}
\\ \label{compatibility3}
 |F (1)| = 2 \left|\int_{R_{<}}^1 \diff s \: g ^2(s) \vec{B}(s) \cdot \vec{e}_{\vartheta} \right| &\leq& C.
\end{eqnarray}
Moreover there is a constant $C$ such that
\begin{equation}\label{Fbound2}
|F(r)| \leq C \min \left( \frac{ |r-R_< | }{\ep} g^2 (r), \: 1 + \frac{C}{\ep^2 |\log \ep| } |r-1| \right)
\end{equation}
for any $\rv \in \A$. 
\end{lem}

\begin{proof}
We have 
\[
 \vec{B}(r) = \left( (\Om - [\Om] )r + [\Om] (r-r^{-1}) + \omega_0 r^{-1}\right)\vec{e}_{\vartheta}
\]
so \eqref{Bbound} is a consequence of $|\omega_0| \leq C\ep ^{-1}$ (see \eqref{est omega_0}) and the fact that $|\A| \propto \ep |\log \ep|$ (see \eqref{the annulus} and \eqref{the inner radius}) and $\Om \propto \ep^{-2} |\log \ep|^{-1}$.\\
We obtain \eqref{gradFbound} from \eqref{g estimates}, \eqref{F2} and \eqref{Bbound}. On the other hand $ F(\rt) =0 $, which implies that  \eqref{Fbound} follows from \eqref{gradFbound} and $ |\A| \propto \ep |\log \ep|$. Moreover \eqref{compatibility3} is exactly the same as \eqref{compatibility} in Proposition \ref{optimal phase pro}.
\newline
To prove \eqref{Fbound2} we prove that $|F(r)|$ is smaller than both terms on the right-hand side. From equation \eqref{F} we have
\[
 |F(r)| \leq 2 \lf\| B \ri\|_{L^{\infty} (\A)} \int_{R_<}^r \diff s \: g^2(s)  \leq C \ep ^{-1} \int_{R_<}^r \diff s \: g^2(s)
\]
but $g^2$ is increasing so
\[
\int_{\rt} ^r \diff s \: g^2(s)  \leq g^2 (r) |r-\rt|. 
\]
We deduce that $|F(r)|$ is smaller than the first term on the right-hand side of \eqref{Fbound2}. On the other hand, using \eqref{gradFbound} and \eqref{compatibility3} we have immediately
\[
 |F(r)| \leq C (1 + \ep^{-2} |\log \ep|^{-1} |r-1|)
\]
for some finite constant $C$. Thus \eqref{Fbound2} is proven. 
\end{proof}

We now prove the energy bounds that are the starting point for the vortex balls construction in Subsection \ref{vortex balls sec}. 

\begin{lem}[\textbf{Preliminary energy bounds}]\label{lem:initialbound}\mbox{}	\\
We have, for $ \eps > 0  $ small enough,
\begin{eqnarray}\label{Fgfirstbound}
\F [u] &\leq& \frac{C}{\ep ^2}, 
\\ \E [u] & \geq & - \frac{C}{\ep ^2}. \label{Egfirstbound}
\end{eqnarray}
\end{lem}

\begin{proof}
We have, using (\ref{Bbound}) together with the normalization of $\gpm$,
\begin{equation}\label{trick momentum}
2 \left| \int_{\ann } \diff \rv \: g^2 \vec{B} \cdot (iu,\nabla u) \right| \leq \frac{1}{2} \int_{\ann } \diff \rv \: g ^2 \left|\nabla u \right|^2 + 2 \int_{\ann} \diff \rv \: g ^2 B^2 |u|^2 \leq  \frac{1}{2} \int_{\ann} \diff \rv \: g ^2 \left|\nabla u \right|^2 + \frac{C}{\ep ^2}. 
\end{equation}
Equation (\ref{Egfirstbound}) immediately follows. We obtain (\ref{Fgfirstbound}) from (\ref{starting bound}): 
\bml{\label{calcul}
\F [u] = \E [u] + 2 \int_{\ann } \diff \rv \: g^2 \vec{B} \cdot (iu,\nabla u) \leq \OO (\ep ^{\infty}) + 2 \left| \int_{\ann } \diff \rv \: g^2 \vec{B} \cdot (iu,\nabla u) \right| \leq	\\
	 \frac{1}{2} \int_{\ann} \diff \rv \: g ^2 \left|\nabla u \right|^2 + \frac{C}{\ep ^2} \leq \frac{1}{2} \F[u] + \frac{C}{\eps^2}.
}
Subtracting $\frac{1}{2} \F[u] $ from both sides of the last inequality we get \eqref{Fgfirstbound}.
\end{proof}

\subsection{Equation for $u$ and Boundary Estimate}
\label{subsec:boundary}

We first derive from the equations for $\gpm$ and $g$ an equation satisfied by $u$:

\begin{lem}[\textbf{Equation for $u$}]\label{lem:equation u}\mbox{}\\
Let $u$ be defined by \eqref{defi u}. We have on $\A$
\begin{equation}\label{vareq u 2}
 -\nabla (g ^2 \nabla u) -2i g^2 \vec{B} \cdot \nabla u +2\frac{g ^4}{\ep ^2} \left( |u|^2 -1 \right)u =\lambda g ^2 u,
\end{equation} 
where $\lambda\in \R$ satisfies the estimate
\begin{equation}\label{scalelambda}
| \lambda | \leq C \left( \left| \Eg [u] \right| + \frac{1}{\ep ^{3/2}|\log \ep|^{1/2}}\Fg [u]^{1/2}\right).
\end{equation}
Moreover $u$ satisfies the boundary condition
\begin{equation}\label{Neuman u}
\dd _r u =0 \mbox{ on } \dd \B . 
\end{equation}
\end{lem}

\begin{proof}
In this proof we use the short hand notation
\begin{equation}\label{phase short}
 \vec{\hat{\Om}} = \hat{\Om} \vec{e}_r := \left([\Om]-\om_0\right)\vec{e}_r.
\end{equation}
The equation for $u$ is a consequence of the variational equation \eqref{GP variational} satisfied by $ \gpm $,
\begin{equation}\label{recall GP variational}
	- \Delta \gpm - 2 \vec{\Omega}\cdot \vec{L} \gpm + 2 \eps^{-2} \lf| \gpm \ri|^2 \gpm = \chem \gpm,
\end{equation}
and that solved by $g$ (see \eqref{hgpm var}),
\begin{equation}\label{recall equation g}
 - \Delta g + \bigg( \Om r - \frac{\hat{\Om}}{r} \bigg)^2 g -\Om ^2 r ^2 + 2 \ep^{-2} g^3 = \hchem g.
\end{equation}
Direct computations starting from (\ref{defi u}) show that
\begin{equation}\label{Delta gpm}
 -\Delta \gpm = \bigg( -u \Delta g - g \Delta u -2 \nabla u \cdot \nabla g +\frac{\hat{\Om}^2}{r^2} g u + \frac{2}{r^2}g \vec{\hat{\Om}}\cdot \vec{L} u \bigg) e^{i \hat{\Om} \vartheta}
\end{equation}
and
\begin{equation}\label{L gpm}
-2 \vec{\Omega}\cdot \vec{L} \gpm = \bigg( -2 g \vec{\Omega}\cdot \vec{L} u -2 \Om \hat{\Om} gu\bigg) e^{i \hat{\Om} \vartheta}.
\end{equation}
Plugging the equation for $g$ in (\ref{Delta gpm}) we obtain
\begin{equation}\label{Delta gpm 2}
 -\Delta \gpm = \left( - g \Delta u -2 \nabla u \cdot \nabla g +\hchem g u + 2\Om \hat{\Om} gu - 2 \ep^{-2}g^3 u + 2 g \vec{\hat{\Omega}}\cdot \vec{L} u  \right) e^{i \hat{\Om} \vartheta}.
\end{equation}
Next, combining (\ref{recall GP variational}), (\ref{L gpm}) and (\ref{Delta gpm 2}), we have
\begin{equation}\label{vareq u}
 - g \Delta u - 2\nabla u \cdot \nabla g + \frac{2}{\ep ^2} g^3 \left( |u|^2-1 \right) u + \frac{2}{r^2}g \vec{\hat{\Omega}}\cdot \vec{L} u -2 g \vec{\Omega}\cdot \vec{L} u = \lambda gu
\end{equation}
where 
\begin{equation}\label{defi lambda}
\lambda = \chem - \hchem. 
\end{equation}
There only remains to multiply (\ref{vareq u}) by $g$ and reorganize the terms to obtain (\ref{vareq u 2}).\\ 
The Neumann condition (\ref{Neuman u}) is a straightforward consequence of the corresponding boundary conditions for $g$ and $\gpm$.\\
We now turn to the proof of (\ref{scalelambda}). Multiplying (\ref{vareq u 2}) by $u^*$ and integrating over $\ann$, we obtain, recalling (\ref{annulus mass}),
\begin{equation}\label{lambdaidentity}
 |\lambda| \lf( 1 - \OO(\eps^{\infty}) \ri) \leq |\lambda| \int_{\A} \diff \rv  \: |\gpm |^2  = |\lambda| \int_{\ann} \diff \rv \: g^2 |u| ^2 = \Eg [u] + \int_{\ann} \diff \rv \: \frac{g^4}{\ep ^2}  \left(|u|^4 - 1\right).
\end{equation}
Using Cauchy-Schwarz inequality
\[
\left| \int_{\ann} \frac{g^4}{\ep ^2}  \left(|u|^4 - 1\right) \right| \leq \left( \int_{\ann} \frac{g^4 }{\ep ^2} \left(|u|^2 - 1\right)^2\right)^{1/2} \left(\int_{\ann} \frac{ 1 }{\ep ^2} \left(g^2|u|^2 + g^2\right)^2 \right)^{1/2}
\]
so (\ref{scalelambda}) follows using the upper bounds  \eqref{GPmin estimates} and \eqref{g estimates} on  $g^2 |u|^2 = |\gpm |^2 $ and $g^2$ respectively together with $|\ann| \propto \ep |\log \ep|$.
\end{proof}

We remark that the equation satisfied by $u$ is exactly of the form that we would have obtained if $u$ had been a minimizer of $\E$ under a mass constraint for $g u$.

We are now able to estimate the circulation of $u$ at the boundary of the unit ball:

\begin{lem}[\textbf{Boundary estimate}]\label{lem:boundary}\mbox{}\\
We have, for $ \eps > 0 $ small enough,
\begin{equation}\label{boundary}
\int_{\dd \B} \diff \sigma \: g ^2 \left| \dd_{\tau} u \right|^2 +  \int_{\dd \B} \diff \sigma \: \frac{g ^4}{\ep^2}\left(1-|u|^2 \right)^2 \leq  \frac{C}{\ep^{3/2}|\log \ep|^{1/2}  } \Fg [u]+  \frac{C}{\ep} |\Eg [u]|.
\end{equation}
\end{lem}

\begin{proof}
The main idea is a Pohozaev-like trick, like in \cite[Theorem 3.2]{BBH2}. We multiply both sides of \eqref{vareq u 2} by $ \left(r- \rd \right)\vec{e}_r \cdot \nabla u^*  $, then integrate over $\at \subset \ann$ (see \eqref{defi annt} for the definition of $\at$) and take the real part. Recall that $1- \rd \propto \ep |\log \ep|$. \\
The first term is given by
\bml{
\label{calc5}
- \frac{1}{2} \int_{\at}  \diff \rv \lf\{ \left(r- \rd \right)\vec{e}_r \cdot \nabla u^* \nabla \lf( g^2 \nabla u \ri) +  \left(r- \rd \right)\vec{e}_r \cdot \nabla u \nabla \lf( g^2 \nabla u^* \ri) \ri\} =	\\
\int_{\at} \diff \rv \lf\{ g^2 \left(1-\frac{\rd }{r} \right)\lf| \nabla u \ri|^2 + g^2 \frac{\rd}{r} \lf| \vec{e}_r \cdot \nabla u \ri|^2 + \frac{1}{2}  g^2 \left(r- \rd \right)\vec{e}_r \cdot \nabla |\nabla u|^2 \ri\} = \\
\frac{1}{2} \int_{\partial \B} \diff \sigma \: g^2 \left(1-\rd \right)\lf| \dd_{\tau} u \ri|^2 - \int_{\at}   \diff \rv \lf\{ g \: \left(r- \rd \right)\vec{e}_r \cdot \nabla g \: |\nabla u|^2 + \frac{\rd}{2r} g^2 \lf[ 2 \lf| \vec{e}_r \cdot \nabla u \ri|^2 - \lf| \nabla u \ri|^2 \ri] \ri\},
} 
where we have integrated by parts twice, remembering the Neumann boundary conditions for $g$ and $u$ on $\partial \B$.
\newline
The third terms yields
\bml{
\label{calc6}
\frac{1}{\eps^2} \int_{\at}  \diff \rv \lf\{ \left(r- \rd \right)\vec{e}_r \cdot \nabla u^* g^4 (|u|^2 - 1) u +  \left(r- \rd \right)\vec{e}_r \cdot \nabla u \: g^4 (|u|^2 - 1) u^* \ri\} =   \\ 
\frac{1}{2\eps^2} \int_{\at} \diff \rv \: g^4 \left(r- \rd \right)\vec{e}_r \cdot \nabla \lf( 1 - |u|^2 \ri)^2 =
\frac{1}{2\eps^2} \int_{\partial \B} \diff \sigma \: \left(1-\rd \right)g^4 \lf( 1 - |u|^2 \ri)^2 - \\ 
 \frac{1}{\eps^2} \int_{\at} \diff \rv \lf\{ \left(1-\frac{ \rd }{2 r} \right) g^4 \lf( 1 - |u|^2 \ri)^2 + 2 g^3 \left(r- \rd \right)\vec{e}_r \cdot \nabla g \: \lf( 1 - |u|^2 \ri)^2 \ri\}.
} 
Altogether we thus obtain
\bml{
\label{bdest4}
\frac{1}{2} \left(1-\rd\right) \int_{\dd \B} \diff \sigma \lf\{ g ^2 \left| \dd_{\tau} u \right|^2 + \frac{g ^4}{\ep^2}\left(1-|u|^2 \right)^2  \right\} = 	\\
	\int_{\at}  \diff \rv \lf\{ g \: \left(r- \rd\right)\vec{e}_r \cdot \nabla g \: |\nabla u|^2 + \frac{\rd}{2r} g^2 \lf[ 2 \lf| \vec{e}_r \cdot \nabla u \ri|^2 - \lf| \nabla u \ri|^2 \ri] +  2 i g^2 \left(r- \rd\right)\vec{e}_r \cdot \nabla u^* \: \vec{B} \cdot \nabla u \ri\} +	\\
\frac{1}{\eps^2} \int_{\at} \diff \rv \lf\{ \left(1-\frac{\rd}{2 r} \right) g ^4 \lf( 1 - |u|^2 \ri)^2 + 2 g ^3 \left(r- \rd \right)\vec{e}_r \cdot \nabla g \: \lf( 1 - |u|^2 \ri)^2 \ri\} +	\\
	\frac{\lambda}{2} \int_{\at} \diff \rv  \: g^2 \left(r- \rd \right)\vec{e}_r \cdot \nabla |u|^2.
}
We then estimate the moduli of the terms of the r.h.s.. Obviously
\begin{equation}\label{bdbound1}
\int_{\at} \diff \rv \lf\{ \frac{\rd}{2r} g^2 \lf[ 2 \lf| \vec{e}_r \cdot \nabla u \ri|^2 - \lf| \nabla u \ri|^2 \ri] + \frac{g^4}{\ep^2}\left(1-\frac{\rd}{2 r} \right)\left(1-|u|^2 \right)^2 \ri\} \leq \Fg [u].
\end{equation}
Using \eqref{defi R large} and \eqref{Bbound} we obtain
\begin{equation}\label{bdbound2}
\bigg| \int_{\at} \diff \rv \: i g^2 \left(r- \rd\right)\vec{e}_r \cdot \nabla u^* \: \vec{B} \cdot \nabla u \bigg|\leq C|\log \ep | \int_{\at} \diff \rv \: g ^2 |\nabla u |^2 \leq C|\log \ep | \Fg [u],	
\end{equation}
whereas, using the normalization of $\gpm$, \eqref{Fgfirstbound} and \eqref{scalelambda},
\bml{
 	\label{bdbound3}
\left|\frac{\lambda}{2} \int_{\at} \diff \rv \:  g^2 \left(r- \rd\right)\vec{e}_r \cdot \nabla |u|^2 \right| =  \left| \lambda \int_{\at} \diff \rv \: g^2 u  \left(r- \rd\right)\vec{e}_r \cdot \nabla u^* \right|  \leq \\
C \ep |\log \ep |\left( \left| \Eg [u] \right| + \frac{1}{\ep ^{3/2}|\log \ep|^{1/2}}\Fg [u]^{1/2}\right) \left(\int_{\at}  \diff \rv \: g^2 |u|^2\right)^{1/2}\left(\int_{\at}  \diff \rv \: g^2 |\nabla u|^2\right)^{1/2}  \leq  \\
 C   |\log \ep | \left| \Eg [u] \right| +  \frac{C|\log \ep |^{1/2} }{\ep ^{1/2}}\Fg [u].
}
Combining \eqref{grad est} with \eqref{g low bound}, we obtain
\begin{multline} \label{bdbound4}
\left| \frac{1}{\eps^2} \int_{\at}  \diff \rv \: g^3 \left(r- \rd \right)\vec{e}_r \cdot \nabla g \: \lf( 1 - |u|^2 \ri)^2 \right| \leq	\\
\frac{C |\log \ep |^{9/4}}{\ep ^{1/4}}  \int_{\at}  \diff \rv \: \frac{ 1}{\ep ^2} g^4 (1-|u|^2 )^2 \leq  \frac{C |\log \ep |^{9/4}}{\ep ^{1/4}}\Fg [u] .
\end{multline}
By similar arguments
\begin{equation}\label{bdbound5}
\left| \int_{\at}   \diff \rv \: g \: \left(r- \rd\right)\vec{e}_r \cdot \nabla g \: |\nabla u|^2 \right| \leq \frac{C |\log \ep |^{9/4}}{\ep ^{1/4}}\Fg [u].
\end{equation}
Gathering equations \eqref{bdest4} to \eqref{bdbound5}, we have 
\[
\frac{1}{2}\left(1-\rd\right)\int_{\dd \B} \diff \sigma \left\{ g ^2 \left| \dd_{\tau} u \right|^2 + \frac{g ^4}{\ep^2}\left(1-|u|^2 \right)^2  \right\} \leq C   |\log \ep | \left| \Eg [u] \right| +  \frac{C|\log \ep |^{1/2} }{\ep ^{1/2}}\Fg [u]
\]
and there only remains to divide by $1-\rd \propto \ep |\log \ep|$ to get the result.
\end{proof}

\subsection{Cell Decomposition and Vortex Ball Construction}
\label{vortex balls sec}

In this subsection we aim at constructing vortex balls for $u$. Namely, we want to construct a collection of balls whose radii are much smaller than the width of $\A$ and that cover the set where $|u|$ is not close to $1$. The usual of way of performing this task is to exploit bounds on $\F [u]$. Unfortunately, the bound \eqref{Fgfirstbound} is not sufficient for our purpose. It only implies that the area of the set where $u$ can possibly vanish is of order $\ep^2 |\log \ep|^2$, whereas the vortex balls method requires to cover it by balls of radii much smaller than the width of $\ann$, which is $\OO(\ep |\log \ep|)$. However, the estimates of Lemma \ref{lem:initialbound} are definitely not optimal and can be improved by using a procedure of local vortex balls construction.\\
The idea is the following: If the bound \eqref{Fgfirstbound} could be localized (in a sense made clear below), we could construct vortex balls. As this is not the case we split the annulus $\ann$ into regions where the bound can be localized and therefore vortex balls can be constructed as usual, and regions where this is not the case. 

\begin{defi}[\textbf{Good and bad cells}]
\mbox{}	\\
We cover $\ann$ with (almost rectangular) cells of side length $C \ep |\log \ep|$, using a corresponding division of the angular variable. We denote by $N\propto \ep^{-1} |\log \ep|^{-1} $ the total number of cells and label the cells as $\A_n, n\in \left\lbrace 1,...,N \right\rbrace$. Let $0\leq \alpha< \frac{1}{2}$ be a parameter to be fixed later on. 
\begin{itemize}
\item We say that $\A_n$ is an $\alpha$-good cell if 
\begin{equation}\label{defgc}
\int_{\A _n}  \diff \rv \lf\{ g ^{2} \left| \nabla u \right|^2+ \frac{g ^4 }{\ep ^2} \left(1-|u|^2 \right)^2 \ri\} \leq \frac{|\log \ep|}{\ep} \ep^{-\alpha}.
\end{equation}
We denote by $N_{\alpha} ^{\mathrm{G}}$ the number of $\alpha$-good cells and $GS_{\alpha}$ the (good) set they cover.
\item We say that $\A_n$ is an $\alpha$-bad cell if
\begin{equation}\label{defbc}
\int_{\A _n}  \diff \rv \lf\{ g ^{2} \left| \nabla u \right|^2+ \frac{g ^4 }{\ep ^2} \left(1-|u|^2 \right)^2 \ri\} > \frac{|\log \ep|}{\ep} \ep^{-\alpha}.
\end{equation}We denote by $N_{\alpha} ^{\mathrm{B}}$ the number of $\alpha$-bad cells and $BS_{\alpha}$ the (bad) set they cover.
\end{itemize}
\end{defi}

Note that the annulus $\ann$ has a width $\ell \propto \ep|\log \ep|$ (which implies that $N\propto \ep^{-1} |\log \ep|^{-1}$) so that we are actually dividing it into bad cells where there is much more energy that what would be expected from the localization of the bound \eqref{Fgfirstbound} (namely $ C \ell \ep ^{-2} \propto \eps^{-1}|\log \ep|$) and regions (good cells) of reasonably small energy. A first consequence of this is, neglecting the good cells and  using \eqref{Fgfirstbound}, 
	\begin{equation}\label{numberbad}
		N_{\alpha} ^{\mathrm{B}} \leq \frac{\ep}{|\log \ep|} \ep^{\alpha} \F [u] \leq \frac{C}{\ep |\log \ep|}\ep^{\al} \ll N, 	 
	\end{equation}
i.e., there are very few $\alpha$-bad cells. A consequence of the refined bound \eqref{borne Fg final} that we are aiming at is that there are actually no $\al$-bad cells at all.

We now construct the vortex balls in the good set. The proof is merely sketched because it is an adaptation of well-established methods (see \cite{SS} and references therein). Note that the construction is possible only in the subdomain $\at$ where the density is large enough.

\begin{pro}[\textbf{Vortex ball construction in the good set}]\label{pro:vortexballs}
 	\mbox{}\\
	Let $0\leq \al< \frac{1}{2}$. There is a certain $\ep_0$ so that, for $\ep \leq \ep_0$ there exists a finite collection $ \{ \B_i \}_{i \in I} := \left\lbrace \B (\avi, \varrho_i)\right\rbrace_{i\in I}$ of disjoint balls with centers $ \avi $ and radii $ \varrho_i $ such that
	\begin{enumerate}
		\item $\left\lbrace \rv \in GS_{\al} \cap \at \: : \: \left| |u| - 1  \right| > |\log \ep| ^{-1}  \right\rbrace \subset \bigcup_{i \in I} \B_i$,
		\item for any $\al $-good cell $\A_n$, $\sum_{i, \B_i \cap \A_n \neq \varnothing } \varrho_i = \ep |\log \ep |^{-5} $.		   
	\end{enumerate}
	Setting $d_i:= \dg \{ u, \partial \B_i \} $, if $ \B_i \subset \at \cap GS_{\al} $, and $d_i=0$ otherwise, we have the lower bounds
		\begin{equation}\label{lowboundballs}
		 	\int_{\B_i}  \diff \rv \: g ^2 \left|\nabla u\right|^2 \geq 2\pi \left(\frac{1}{2} -\al \right) |d_i|   g^2 (a_i) \left| \log \ep \right| \left(1-C \frac{\log \left| \log \ep \right|}{\left|\log \ep\right|}\right).
		\end{equation}
\end{pro}

\begin{proof}
	We begin by covering the subset of $GS_{\al}\cap \at$ where $u$ can possibly vanish by balls of much smaller radii than those announced in the Proposition. Let $\A_n$ be an $\alpha$-good cell. We have, using \eqref{g low bound} and \eqref{defgc},
	\bmln{
	 	\int_{\A_n\cap \at}  \diff \rv \lf\{ \left| \nabla u \right|^2 + \frac{1}{\ep ^3 |\log \ep|^3} \left(1-|u|^2 \right)^2 \ri\} \leq \\
		 C \ep  |\log \ep|^3 \int_{\A_n}  \diff \rv \left\{ g^2 \left| \nabla u \right|^2 + \frac{g^4}{\ep ^2 } \left(1-|u|^2 \right)^2 \right\} \leq C |\log \ep|^4 \ep^{-\al}.
	}
	Then, using the Cauchy-Schwarz inequality, 
	\[
	 	\int_{\A_n  \cap \at}  \diff \rv \: \frac{\left| \nabla |u| \right| \left| 1-|u|^2 \right|}{\ep^{3/2}|\log \ep|^{3/2}}\leq C|\log \ep|^4 \ep^{-\al}
	\]
	and the coarea formula implies
	\[
	 	\int_{t\in \R^{+}} \diff t \: \frac{\left|1-t^2\right|}{\ep^{3/2}|\log \ep|^{3/2}} \HH ^1 \left(\left\lbrace \rv \in \A_n\cap \at \: : \: |u|(\rv) = t\right\rbrace\right) \leq C  |\log \ep|^4 \ep^{-\al},
	\]
	where $\HH ^1$ stands for the one-dimensional Hausdorff measure. We argue as in \cite[Propositions 4.4 and 4.8]{SS} to deduce from this that the set $\{ \rv \in \A_n\cap \at \: : \: \left||u| - 1 \right| >|\log \ep| ^{-1}   \}$ can be covered by a finite number of disjoint balls $\{ \B(\avtj, \tilde{\varrho}_{j})\}_{j \in J}$ with $\sum_{j \in J} \tilde{\varrho}_j \leq C \ep^{3/2-\al}|\log \ep|^{13/2}$. Doing likewise on each $\al$-good cell, we get a collection $\{ \B (\avti, \tilde{\varrho}_i)\}_{i \in I}$, covering $\{ \rv \in GS_{\al}\cap \at \: : \: \left||u| - 1 \right| >|\log \ep| ^{-1}  \}$ and satisfying 
	\begin{equation}\label{seeds2}
		\sum_{i \in I, \B(\avti, \tilde{\varrho}_i) \cap \A_n \neq \varnothing } \tilde{\varrho}_i \leq C\ep^{3/2-\al} |\log \ep |^{13/2}, \mbox{ for any }\al \mbox{-good cell }\A_n. 
	\end{equation}
	The balls in this collection may overlap because we have constructed them locally in the cells. However, by merging the balls that intersect as in \cite[Lemma 4.1]{SS}, we can construct a finite collection of disjoint balls (still denoted by $\B(\avti, \tilde{\varrho}_i)$) satisfying the same bounds on their radii and still covering $\{ \rv \in GS_{\al}\cap \at \: : \: \left| |u| - 1  \right| > |\log \ep| ^{-1}  \}$.\\
	The collection $ \B(\avti, \tilde{\varrho}_i)$ is represented in Figure 1. The gray regions are those were $u$ could vanish, i.e., bad cells and vortex balls. The dashed circle is the inner boundary of $\at$.\\ 
	\begin{figure}\label{fig1}
	\begin{center}
	 \includegraphics[width=200pt,height=200pt]{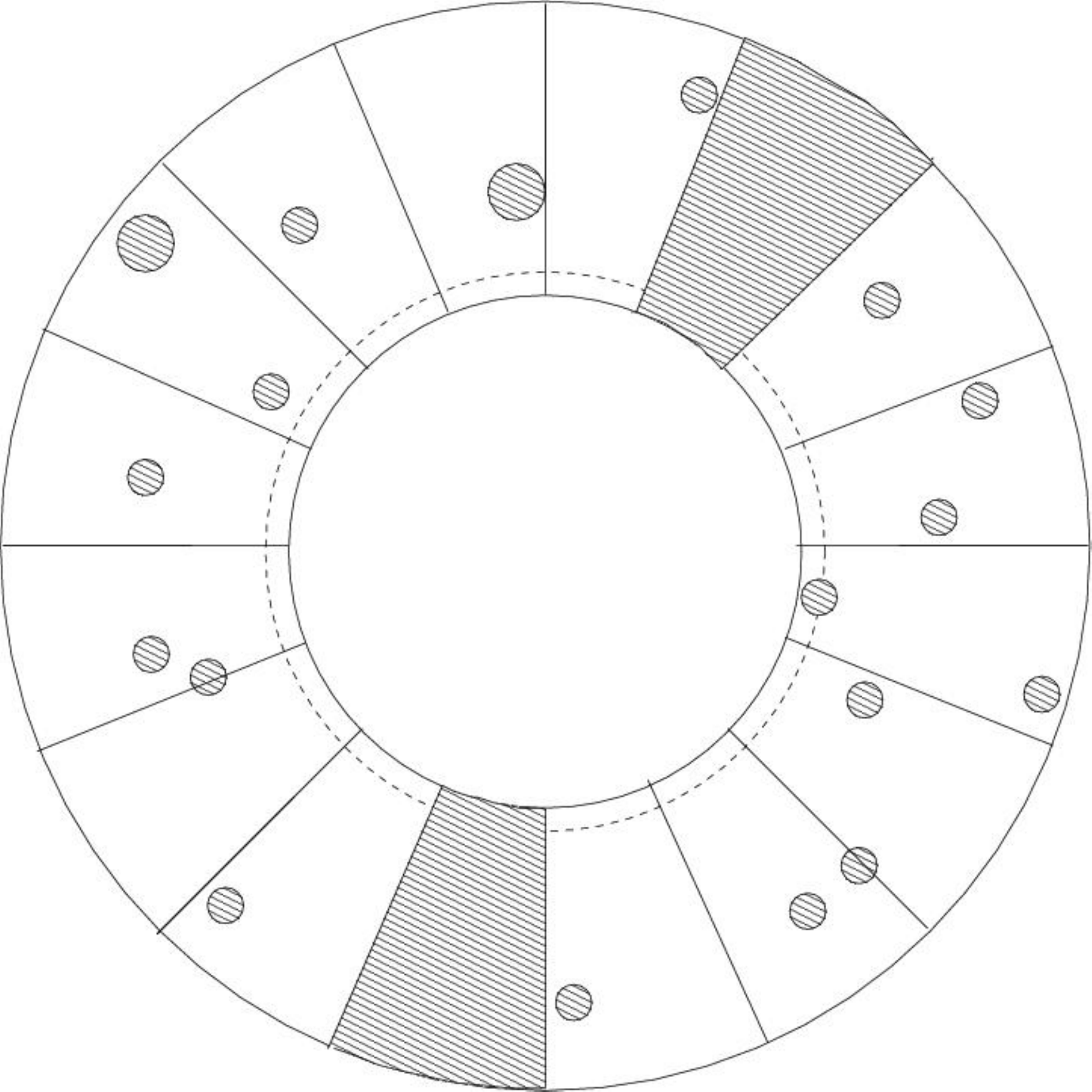}
	\caption{Initial collection of vortex balls}
	\end{center}
	\end{figure}
	We now let the balls in our initial collection grow and merge using the method described in \cite[Section 4.2]{SS}, adding lower bounds on the conformal annuli constructed. In brief, the idea is to introduce a process parametrized by a variable $t$ (interpreted as `time') so that, as $t$ increases, the balls will grow all with the same dilation factor. There are two phases in the process: When balls grow independently their radii get multiplied by some factor, say $m(t)$, and we add lower bounds to the kinetic energy of $u$ on the annuli between the initial and grown balls. The first time two (or more) balls touch we merge them into larger balls (see \cite[Lemma 4.1]{SS}), and continue with the growing phase. Figure 2 shows this process. 
	\newline
	The lower bounds on the annuli are obtained by integrating the kinetic energy over circles centered at the balls centers $ \avi $ during the growth phases:
	\begin{equation}\label{growth}
	 \int _{\B(\avi,m(t) \varrho_i) \setminus \B(\avi,\varrho_i)} \diff \rv \: |\nabla u |^2  \geq 2\pi d_i ^2 \log \left(m(t)\right) \left(1-|\log \ep| ^{-1} \right),
	\end{equation}
	where $d_i : = \dg \{ u,\partial \B(\avi, \varrho_i) \}$. \\
	The main point in this computation is that $||u|-1| < |\log \ep|^{-1}$ (in particular $u$ cannot vanish) between the circles $ \partial \B(\avi,\varrho_i)$ and $\partial \B(\avi,m(t) \varrho_i)$. Thus the degree of $u$ is $d_i$ on any circle $\partial \B(\avi,\varrho)$ with $ \varrho_i<\varrho<m(t) \varrho_i$. 
	\newline
	To add the lower bounds obtained during the different growth phases we use that if two balls, say $\B(\av_p, \varrho_p)$ and $\B(\av_q, \varrho_q)$ merge into a new ball $ \B(\av_s, \varrho_s)$ one always has $ d_p ^2 + d_q ^2 \geq |d_p|+|d_q| \geq |d_p+d_q|= |d_s|$. This accounts for the fact that we get a factor $|d_i|$ in \eqref{lowboundballs}, whereas the factor in \eqref{growth} is $d_i ^2$.\\ 
	\begin{figure}
	\begin{center}
 	\includegraphics[width=200pt,height=200pt]{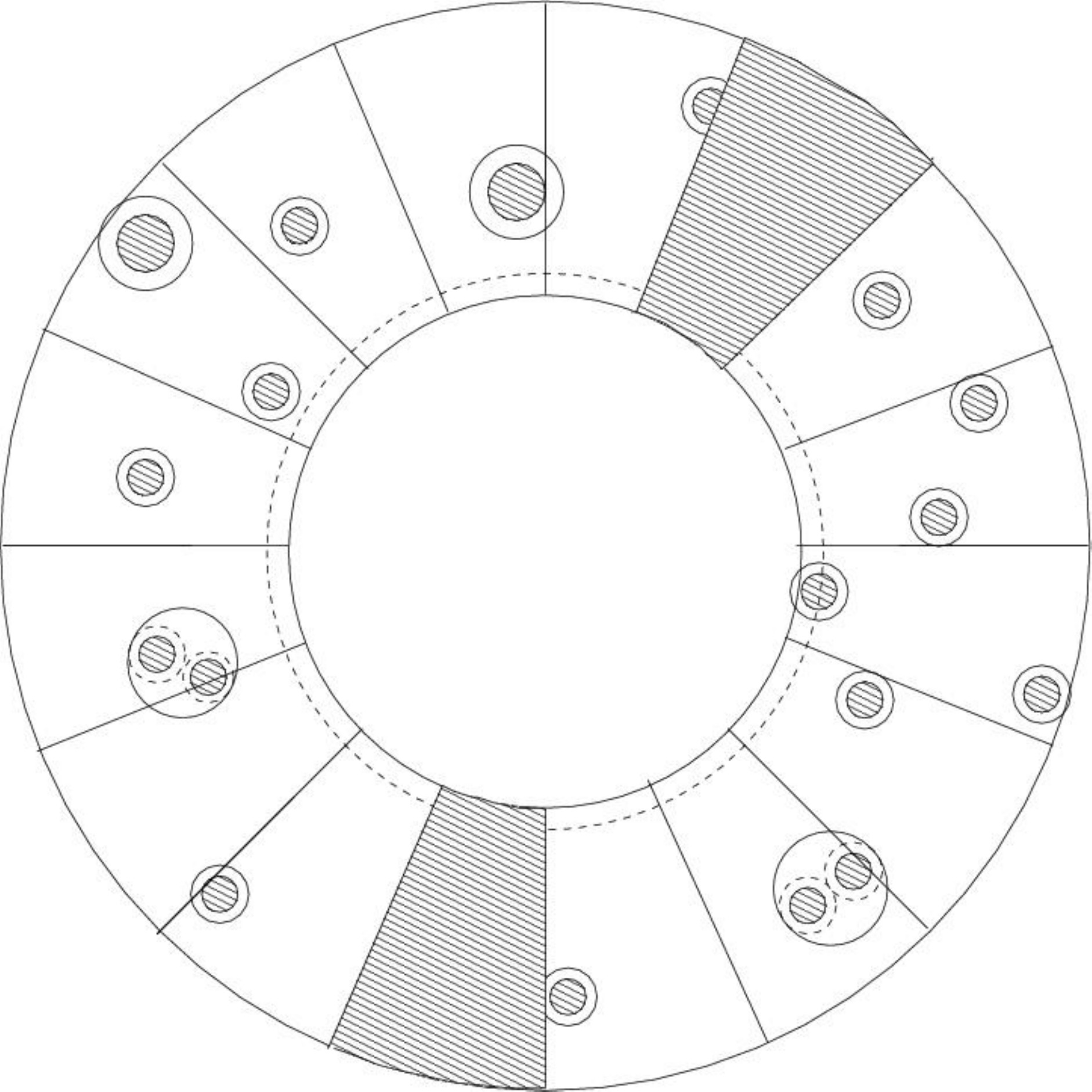}
	\caption{Growth and merging process}
	\end{center}
	\end{figure}
	We stop the process when we have a collection of disjoint balls $\{ \B_i \}_{i \in I} : = \left\lbrace \B (\avi, \varrho_i)\right\rbrace_{i \in I} $ satisfying condition 2 of Proposition \ref{pro:vortexballs} (property 1 is satisfied by construction, since it holds true for the initial family of balls), i.e., the balls are large enough, so that we have
	\begin{equation}\label{lowboundballs2}
		 \int_{\B (\avi,\varrho_i)} \diff \rv \: \left|\nabla u\right|^2 \geq 2\pi \left(\frac{1}{2} -\al \right) |d_i| \left| \log \ep \right|\left(1-C \frac{\log \left| \log \ep \right|}{\left|\log \ep\right|}\right),
	\end{equation}
	where $d_i:=\dg \{u, \partial\B_i\}$, if $ \B_i \subset GS_{\al}\cap \at $, and $d_i =0$ otherwise. The final configuration is drawn in Figure 3. Note that the radii of the final vortex balls are still much smaller than the diameter of the cells. This is important for \eqref{densityballs} below and the jacobian estimate in the next section.
	\newline
	The logarithmic factor $\left(\frac{1}{2} -\al \right) \left| \log \ep \right|\left(1-C \frac{\log \left| \log \ep \right|}{\left|\log \ep\right|}\right)$ comes from the logarithm of the dilation factor of the collections of balls, i.e.,
	\[
		\frac{\sum_{i \in I} \varrho_i}{\sum_{i \in \tilde{I}} \tilde{\varrho}_i}\geq C \ep ^{\alpha-1/2} |\log \ep|^{-23/2}.
	\]
	On the other hand, using \eqref{pointwise bounds} and the fact that $|\nabla \tfm |\leq C \ep^{-2} |\log \ep |^{-2}$ (see \eqref{TFm0}), we have 
	\begin{multline}\label{densityballs}
	 \left| \min_{\rv\in \B_i} g^2(r) -g^2 (a_i)\right| \leq C\left( \ep^{1/2}|\log \ep|^{3/2} \lf\| \tfm \ri\|_{L^{\infty} (\ann)}+ \left| \min_{\rv\in \B_i}\tfm(r) -\tfm(a_i)\right| \right) \leq 
	\\ C \left( \ep^{-1/2}|\log \ep|^{1/2}  + \ep^{-2} |\log \ep |^{-2} \varrho_i \right) \leq C \ep^{-1} |\log \ep|^{-7} \leq C |\log \ep |^{-4} g^2 (a_i),
	\end{multline}
	because $ \avi \in \at$ and $ \varrho_i \leq C \ep |\log \ep|^{-5}$. We conclude from \eqref{lowboundballs2} and \eqref{densityballs} that the lower bound \eqref{lowboundballs} holds on each ball we have constructed by bounding below $g^2$ with its minimum on the ball $ \B_i$. The error $\min_{\rv\in \B_i} g^2(\rv) -g^2 (a_i)$ can then be absorbed into the $g^2 (a_i) \log |\log \ep|$ term.
\end{proof}
\begin{figure}
	\begin{center}
 	\includegraphics[width=200pt,height=200pt]{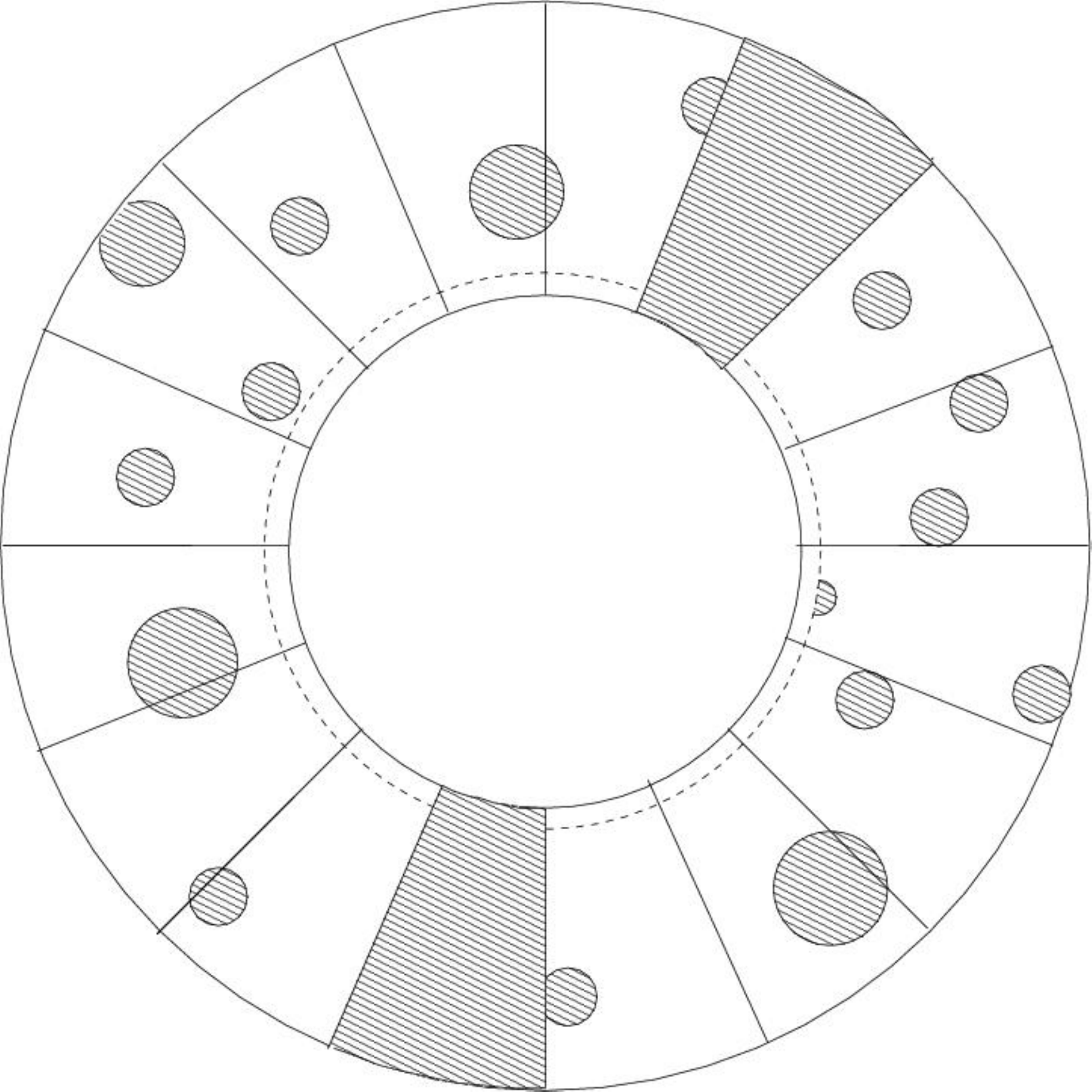}
	\caption{Final configuration of cells and vortex balls.}
	\end{center}
\end{figure}

\subsection{Jacobian Estimate} 
\label{jacobian estimate sec}

We now turn to the jacobian estimate. With the vortex balls that we have constructed, it is to be expected that in the $\al$-good set the vorticity measure of $u$ will be close to a sum of Dirac masses 
	\[
		\curl(iu,\nabla u)\simeq \sum_{i \in I} 2 \pi d_i \delta(\rv - \avi),
	\]
where $\delta(\rv - \avi)$ stands for the Dirac mass at $ \avi $. Indeed, outside of the balls $|u|\simeq 1$, so $\curl(iu,\nabla u) \simeq 0$, and the balls have very small radii compared to the size of $\ann$.\\
Proposition \ref{pro:jacest} gives a rigorous statement of this fact. Note that we can quantify the difference between $\curl(iu,\nabla u)$ and $\sum_i 2 \pi d_i \delta(\rv - \avi) $ in terms of the energy only in the domain $\at$ where the density is large enough. 

\begin{pro}[\textbf{Jacobian estimate}]\label{pro:jacest}
 	\mbox{}\\
	Let $0\leq \al<\frac{1}{2}$ and $\phi$ be any piecewise-$C^1$ test function with compact support 
	\[
	 {\rm supp}(\phi) \subset \at \cap GS_{\al}.  
	\]
	Let $\left\lbrace \B_i \right\rbrace_{i\in I} : = \lf\{ \B(\avi,\varrho_i) \ri\}_{i \in I} $ be a disjoint collection of balls as in Proposition \ref{pro:vortexballs}. Setting $d_i:= \dg \{ u, \partial \B_i \} $, if $ \B_i \subset \at \cap GS_{\al} $, and $d_i=0$ otherwise, one has
	\begin{equation}\label{jacobianestimate}
		\bigg|\sum_{i\in I}  2 \pi d_i \phi (\avi)- \int_{GS_{\al}\cap \at} \diff \rv \: \phi \:  \curl (iu,\nabla u) \bigg| \leq  C \left\Vert \nabla \phi \right\Vert_{L^{\infty}(GS_{\al})}\ep ^{2} |\log \ep|^{-2} \Fg[u] .  
	\end{equation}	
\end{pro}

\begin{proof}
	We argue as in \cite[Chapter 6]{SS}. We first introduce a function $\xi:\R^+ \rightarrow \R ^+$ as follows: $\xi (x)=2x$, if $x\in[0,1/2]$, and $\xi (x)=1$, if $x\in[1/2,+\infty[$. This function satisfies
	\begin{itemize}
	 	\item $\xi(t)\leq 2t$ and $\xi ' (t) \leq 2$,
		\item $\left| \xi (t) - t \right|\leq \left|1-t \right|$ and $\left| \xi (t) - 1 \right|\leq \left|1-t \right|$,
		\item $\left| \xi (t)^2 - t^2 \right|\leq 3t \left|1-t \right|$,
	\end{itemize}
	and we define $w$ as a regularization of $u$ (in the sense that $|w|=1$ and therefore $\curl (iw , \nabla w) = 0$ when $|u|$ is far enough from $0$):
	\begin{equation}\label{uregular}
		w= \frac{\xi(|u|)}{|u|}u.
	\end{equation}
	Remark that we have $(iw,\nabla w) = \frac{|w|^2}{|u|^2}(iu,\nabla u)$ and this has a meaning even if $u$ vanishes.
	By integrating by parts on $GS_{\al}\cap \at$ and using the assumptions on $\phi$, one has
	\begin{equation}\label{proofJE1}
		\int_{GS_{\al}\cap \at} \diff \rv \lf[ \curl (iu,\nabla u)-\curl (iw,\nabla w) \right] \phi =  - \int_{GS_{\al}\cap \at} \diff \rv \left[ (iu,\nabla u)- (iw,\nabla w) \right] \nabla ^{\perp}\phi. 
	\end{equation}
	Now,
	\bml{\label{proofJE3} 
		\bigg|\int_{GS_{\al}\cap \at} \diff \sigma \left[ (iu,\nabla u)- (iw,\nabla w) \right] \nabla ^{\perp}\phi \bigg| \leq \left\Vert \nabla \phi \right\Vert_{L^{\infty}(GS_{\al})} \int_{ \at} \diff \rv \left|(iu,\nabla u)- (iw,\nabla w) \right| \leq \\ 
		\left\Vert \nabla \phi \right\Vert_{L^{\infty}(GS_{\al})} \int_{\at} \diff \rv \: \frac{\left||u|^2-|w|^2 \right|}{|u|} \left| \nabla u \right| \leq \left\Vert \nabla \phi \right\Vert_{L^{\infty}(GS_{\al})} \left\Vert 1-|u| \right\Vert_{L^2(\at)} \left\Vert \nabla u \right\Vert_{L^2( \at)} \leq \\
		C \left\Vert \nabla \phi \right\Vert _{L^{\infty}(GS_{\al})} \ep ^{5/2} | \log \ep |^{9/2} \Fg [u] 
	}
	using the properties of $\xi$, the Cauchy-Schwarz inequality, the definition \eqref{defiF} of $\Fg [u]$, $g ^2 \geq C \ep^{-1} |\log \ep|^{-3}$ on $\at$ and $\left(1-|u|\right)^2 \leq \left(1-|u|^2\right)^2$.\\
	We now evaluate 
	\begin{equation}\label{proofJE4}
		\int_{ GS_{\al}\cap \at} \diff \rv \: \curl (iw,\nabla w) \phi = \sum_{i\in I} \int_{\B_i\cap \at} \diff \rv \: \curl (iw,\nabla w)\phi,
	\end{equation}
	which follows from the fact that $\left| |u| - 1 \right| < |\log \ep |^{-1}$ outside $\cup_{i\in I} \B_i$, so that $|w|=1$ and $\curl(iw,\nabla w)=0$ outside $\cup_{i\in I} \B_i$. If $\B_i \subset \at \cap GS_{\al}$, we have
	\begin{equation}\label{proofJE5}
		\int_{\B_i} \diff \rv \: \left| \curl (iw,\nabla w)\right| \left| \phi(\rv)-\phi(\avi)\right| \leq C \left\Vert \nabla \phi \right\Vert_{L^{\infty}(B_i)} \varrho_i \int_{\B_i} \diff \rv \: |\nabla u |^2
	\end{equation}
	and
	\begin{equation}\label{proofJE6}
		\int_{\B_i} \diff \rv \: \curl(iw,\nabla w ) = \dg  \{u, \partial \B_i \} = 2 \pi d_i
	\end{equation}
	by definition of $w$ and $d_i$. On the other hand, if $\B_i \nsubseteq \at \cap GS_{\al}$ then $\B_i \cap \dd \left(\at \cap GS_{\al}\right)\neq \varnothing$ and thus
	\begin{multline}\label{proofJE7}
	 \bigg|\int_{\B_i\cap \at} \diff \rv \: \curl (iw,\nabla w) \phi \bigg| \leq \int_{\B_i\cap \at \cap GS^{\al} } \diff \rv \: |\phi | |\nabla w| ^2 \leq 4 \int_{\B_i\cap \at \cap GS^{\al}} \diff \rv \: |\phi | |\nabla u| ^2  \leq
	\\ C \left\Vert \nabla \phi \right\Vert_{L^{\infty}(\B_i)} \varrho_i \int_{\B_i\cap \at \cap GS_{\al} } \diff \rv \: |\nabla u |^2,
	\end{multline}
	because $|\nabla w| \leq 2 |\nabla u| $ and $\phi$ is supported in the interior of $\at \cap GS_{\al}$. Gathering equations \eqref{proofJE1} to \eqref{proofJE7}, we obtain
	 \begin{multline}\label{proofJE8}
	  \bigg|\sum_{i\in I}  d_i \phi (\avi) - \int_{GS_{\al}\cap \at} \diff \rv \: \phi \: \curl (iu,\nabla u) \bigg| \leq  C \left\Vert \nabla \phi \right\Vert_{L^{\infty}(GS_{\al})}\ep ^{5/2} |\log \ep|^{9/2} \Fg[u] +
	\\ C \sum_i \left\Vert \nabla \phi \right\Vert_{L^{\infty}(\B_i)} \varrho_i \int_{\B_i\cap \at \cap GS_{\al} } \diff \rv \: |\nabla u |^2  
	 \end{multline}
	and the result follows because, using $ \varrho_i \leq \ep |\log \ep|^{-5}$ and \eqref{g low bound},
	\begin{multline}\label{proofJE9}
	 \sum_i \left\Vert \nabla \phi \right\Vert_{L^{\infty}(\B_i)} \varrho_i \int_{\B_i\cap \at \cap GS_{\al} } \diff \rv \: |\nabla u |^2 \leq C \left\Vert \nabla \phi \right\Vert_{L^{\infty}(GS_{\al})} \ep^2 |\log \ep |^{-2} \sum_i \int_{\B_i\cap \at \cap GS_{\al} } \diff \rv \: g^2 |\nabla u |^2 \leq
	\\  C \left\Vert \nabla \phi \right\Vert_{L^{\infty}(GS_{\al})} \ep^2 |\log \ep |^{-2} \F [u].
	\end{multline}
\end{proof}

Note that (\ref{jacobianestimate}) is equivalent to saying that the norm of $\curl(iu,\nabla u)- \sum_i d_i \delta(\rv - \avi) $ in $ ( C^1_c (\at \cap GS_{\al}) )^*$, i.e., the dual space of $ C^1_c (\at \cap GS_{\al}) $, is controlled by the energy.

\subsection{Completion of the Proof of Proposition \ref{pro:refinedbound}}
\label{proof completion sec}

We now complete the proof of Proposition \ref{pro:refinedbound}, collecting the estimates of the preceding subsections.\\
We want to avoid any unwanted boundary term when performing integrations by parts in the proof below. Indeed, our radial frontiers between the good set and the bad set are somewhat artificial and have no physical interpretation. Therefore it is difficult to estimate integrals on these boundaries.\\
To get around this point we need to introduce an azimuthal partition of unity on the annulus in order to `smooth' the radial boundaries appearing in our construction. This requires new definitions:
	\begin{defi}[\textbf{Pleasant and unpleasant cells}]
	\mbox{}	\\
		Recall the covering of the annulus $\ann$ by cells $\A_n,\: n\in \left\lbrace 1,..,N \right\rbrace$. We say that $\A_n$ is 
		\begin{itemize}
		 \item an $\al$-pleasant cell if $\A_n$ and its two neighbors are good cells. We note $PS_{\al}$ the union of all $\al$-pleasant cells and $N_{\alpha} ^{\mathrm{P}}$ their number,
		\item an $\al$-unpleasant cell if either $\A_n$ is a bad cell, or $\A_n$ is a good cell but its two neighbors are bad cells. We note $UPS_{\al}$ the union of all $\al$-unpleasant cells and $N_{\al}^{\mathrm{UP}}$ their number,
		\item an $\al$-average cell if $\A_n$ is a good cell but exactly one of its neighbors is not. We note $AS_{\al}$ the union of all $\al$-average cells and $N_{\al}^{\mathrm{A}}$ their number.		
		\end{itemize}
	\end{defi}
	
Remark that one obviously has, recalling \eqref{numberbad},
\begin{equation}\label{numberunpleasant}
 N_{\al} ^{\mathrm{UP}} \leq \frac{3}{2} N_{\al} ^{\mathrm{B}} \ll N
\end{equation}
and
\begin{equation}\label{numberaverage}
N_{\al} ^{\mathrm{A}} \leq 2 N_{\al} ^{\mathrm{B}} \ll N.
\end{equation}
The average cells will play the role of transition layers between the pleasant set, where we will use the tools of Subsections \ref{vortex balls sec} and \ref{jacobian estimate sec}, and the unpleasant set, where we have little information and therefore have to rely on more basic estimates (like those we used in the proof of Lemma \ref{lem:initialbound}). To make this precise we now introduce the azimuthal partition of unity we have announced.\\  
Let us label $UPS_{\al}^l, l\in\left\lbrace1,\ldots,L\right\rbrace$, and $PS_{\al}^m,m\in\left\lbrace1,\ldots,M\right\rbrace$, the connected components of the $\al$-unpleasant set and  $\al$-pleasant set respectively. We construct azimuthal positive functions, bounded independently of $\ep$, denoted by $\chi_l ^{\mathrm{U}}$ and $\chi_m ^{\mathrm{P}}$ (the labels U and P stand for ``pleasant set" and ``unpleasant set") so that
\begin{eqnarray}\label{partition}
	\chi_l ^{\mathrm{U}} &:=& 1 \mbox{ on } UPS_{\al}^l, \nonumber \\
	\chi_l ^{\mathrm{U}} &:=& 0 \mbox{ on } PS_{\al}^m, \mbox{ } \forall m\in\left\lbrace1,\ldots,M\right\rbrace, \mbox{ and on } UPS_{\al}^{l'}, \mbox{ } \forall  l' \neq l, \nonumber\\
	\chi_m ^{\mathrm{P}} &:=& 1 \mbox{ on } PS_{\al}^m, \nonumber\\
	\chi_m ^{\mathrm{P}} &:=& 0 \mbox{ on } UPS_{\al}^l, \mbox{ } \forall  l \in\left\lbrace1,\ldots,L\right\rbrace, \mbox{ and on } PS_{\al}^{m'}, \mbox{ } \forall m' \neq m, \nonumber \\
	\sum_m \chi_m ^{\mathrm{P}} + \sum_l \chi_l ^{\mathrm{U}} &=& 1 \mbox{ on } \ann. 
\end{eqnarray}
It is important to note that each function so defined varies from $0$ to $1$ in an average cell. A crucial consequence of this is that we can take functions satisfying
\beq\label{gradientchi}
 	\lf|\nabla \chi_l ^{\mathrm{U}} \ri| \leq \frac{C}{\ep |\log \ep |},	\hspace{1,5cm}	\lf|\nabla \chi_m ^{\mathrm{P}}\ri| \leq \frac{C}{\ep |\log \ep |},
\eeq
because the side length of a cell is $\propto \ep |\log \ep|$. For example one can choose this partition of unity to be constituted of piecewise affine functions of the angle. 
\newline
We will use the short-hand notation
\begin{eqnarray}\label{chi in}
 \chi_{\mathrm{in}} &: =& \sum_{m=1} ^M \chi_m ^{\mathrm{P}}, \\
\chi_{\mathrm{out}} &: =& \sum_{l=1} ^L \chi_l ^{\mathrm{U}} \label{chi out}.
\end{eqnarray}
The subscripts `in' and `out' refer to `in the pleasant set' and `out of the pleasant set' respectively.
 
We would like to use the jacobian estimate of Proposition \ref{pro:jacest} with $\phi = \chi_{\mathrm{in}} F$, whose support is not included in $\at$ but only in $\A$ (moreover it does not vanish on $\dd \B$). We will thus need a radial partition of unity: We introduce two radii $R_{\mathrm{cut}}^+$ and $R_{\mathrm{cut}}^-$ as
\begin{eqnarray}\label{Rcutplus}
 R_{\mathrm{cut}}^+ &: =& 1 -\ep |\log \ep|^{-1}, \\
R_{\mathrm{cut}}^- &: =& \rd + \ep |\log \ep|^{-1}. \label{Rcutminus}
\end{eqnarray}
Let $\xi_{\mathrm{in}}(r)$ and $\xi_{\mathrm{out}}(r)$ be two positive radial functions satisfying
\begin{eqnarray}\label{radial partition}
\xi_{\mathrm{in}} (r) &: =& 1 \mbox{ for } R_{\mathrm{cut}}^- \leq r \leq R_{\mathrm{cut}}^+, \nonumber \\
\xi_{\mathrm{in}} (r) &: =& 0 \mbox{ for }  \rt \leq r \leq \rd  \mbox{ and for } r=1, \nonumber \\
\xi_{\mathrm{out}} (r) &: =& 1 \mbox{ for } \rt \leq r \leq \rd, \nonumber \\
\xi_{\mathrm{out}} (r) &: =& 0 \mbox{ for } R_{\mathrm{cut}}^- \leq r \leq R_{\mathrm{cut}}^+, \nonumber \\
\xi_{\mathrm{in}} + \xi_{\mathrm{out}} &=& 1 \mbox{ on } \ann.
\end{eqnarray}
For example $\xi_{\mathrm{in}}$ and $\xi_{\mathrm{out}}$ can be defined as piecewise affine functions of the radius. Moreover, because of \eqref{Rcutplus} and \eqref{Rcutminus}, we can impose
\beq\label{gradientxi}
 	\lf|\nabla \xi_{\mathrm{in}} \ri| \leq \frac{C|\log \ep|}{\ep},	\hspace{1,5cm}	\lf|\nabla \xi_{\mathrm{out}} \ri| \leq \frac{C|\log \ep|}{\ep}.
\eeq
The subscripts `in' and `out' refer to `inside $\at$' and `outside of $\at$' respectively.

In the sequel $\lf\{ \B_i \ri\}_{i \in I} : = \left\lbrace \B(\avi, \varrho_i)\right\rbrace_{i\in I}$ is a collection of disjoint balls as in Proposition \ref{pro:vortexballs}. For the sake of simplicity we label $\B_j$, $j\in J \subset I $, the balls such that $\B_j \subset \at \cap GS_{\al}$. 

\begin{proof}[Proof of Proposition \ref{pro:refinedbound}]\mbox{} \\
	Recall the properties of $F$ (\ref{F2}). By integration by parts, we have
	\begin{equation}\label{IBP}
	 	\int_{\ann} \diff \rv \lf\{ g ^{2} \left| \nabla u \right|^2 - 2g ^2 \vec{B}\cdot (iu,\nabla u) \ri\} = \int_{\ann} \diff \rv \lf\{ g ^{2} \left| \nabla u \right|^2 + F \curl (iu,\nabla u) \ri\} - \int_{\partial \B} \diff \sigma \: F(1) (iu , \partial_{\tau}u) 
	\end{equation}
	and we are going to evaluate the three terms using our previous results. We begin with a lower bound on the kinetic term in \eqref{IBP}, using Proposition \ref{pro:vortexballs}. We introduce a parameter $\gamma$ to be fixed later in the proof and estimate
	\begin{equation}\label{Kinetic1}
		\int_{\ann} \diff \rv \: g ^{2} \left| \nabla u \right|^2 \geq \left(1-\gamma \right) \sum_{j\in J}  \int_{\B_j } \diff \rv \: \xi_{\mathrm{in}}  g^2 |\nabla u| ^2 +\left(1-\gamma \right) \int_{\ann } \diff \rv \: \xi_{\mathrm{out}} g^2 |\nabla u|^2   + \gamma \int_{\ann} \diff \rv \: g ^{2} \left| \nabla u \right|^2.
	\end{equation}
	Using the lower bound \eqref{lowboundballs}, we have 
	\bml{\label{Kinetic2}
	\int_{\B_j} \diff \rv \: \xi_{\mathrm{in}}  g^2 |\nabla u| ^2 \geq \xi_{\mathrm{in}}(a_j) \int_{\B_j}  \diff \rv \: g^2 |\nabla u| ^2 + \left(\inf_{\rv \in \B_j} \xi_{\mathrm{in}}(r) - \xi_{\mathrm{in}}(a_j)\right) \int_{\B_j } \diff \rv \: g^2 |\nabla u| ^2 \geq
	\\ 2\pi \left(\frac{1}{2} -\al \right) |d_j|  \xi_{\mathrm{in}}(a_j) g^2 (a_j) \left| \log \ep \right| \left(1-C \frac{\log \left| \log \ep \right|}{\left|\log \ep\right|}\right) - \frac{C}{|\log \ep| ^4} \int_{\B_j} \diff \rv \: g^2 |\nabla u| ^2.
	}
	The estimate of $\inf_{\B_j} \xi_{\mathrm{in}}  - \xi_{\mathrm{in}}(a_j)$ is a consequence of \eqref{gradientxi} combined with $ \varrho_j \leq C \ep |\log \ep|^{-5} $.\\
	We now compute
	\begin{equation}\label{Momentum1}
	\int_{\ann} \diff \rv \: F \: \curl (iu,\nabla u) =  \int_{\ann}\diff \rv \lf[ \xi_{\mathrm{in}} \chi_{\mathrm{in}} F \: \curl (iu,\nabla u) + \xi_{\mathrm{out}} \chi_{\mathrm{in}}  F \: \curl (iu,\nabla u) + \chi_{\mathrm{out}} F \: \curl (iu,\nabla u) \ri].
	\end{equation}
	We can use Proposition \ref{pro:jacest} to estimate the first term because $\xi_{\mathrm{in}}\chi_{\mathrm{in}} F$ is a piecewise-$C^1$ function with support included in $\at \cap GS_{\al}$. We obtain
	\begin{equation}\label{Momentum2}
	 \int_{\ann} \diff \rv \: \xi_{\mathrm{in}}\chi_{\mathrm{in}} F \: \curl (iu,\nabla u) \geq   2 \pi \sum_{j\in J } d_j F(a_j) \xi_{\mathrm{in}}(a_j)\chi_{\mathrm{in}}(\avj) -C \left\Vert \nabla (\xi_{\mathrm{in}} \chi_{\mathrm{in}} F) \right\Vert_{L^{\infty}(GS_{\al})}\ep ^{2} |\log \ep|^{-2} \Fg[u].
	\end{equation}
 	Now, using \eqref{Fbound}, \eqref{gradFbound}, \eqref{gradientchi} and \eqref{gradientxi}, we have
	\[ 
	\Vert \nabla (\xi_{\mathrm{in}} \chi_{\mathrm{in}} F) \Vert_{L^{\infty} (\ann )} \leq \frac{C|\log \ep |}{\ep ^2 },
	 \]
	so that
	\begin{equation}\label{Momentum3}
	\int_{\ann} \diff \rv \: \xi_{\mathrm{in}}\chi_{\mathrm{in}} F \: \curl (iu,\nabla u)\geq  2\pi \sum_{j\in J } d_j F(a_j) \xi_{\mathrm{in}}( a_j)\chi_{\mathrm{in}}(\avj) - C |\log \ep|^{-1} \Fg[u].
	\end{equation}
	The second term in the r.h.s. of \eqref{Momentum1} is simply bounded below as follows  
	\begin{equation}\label{MomentumOut}
	 \int_{\ann} \diff \rv \: \xi_{\mathrm{out}} \chi_{\mathrm{in}} F \: \curl (iu,\nabla u) \geq - \int_{\ann} \diff \rv \: \xi_{\mathrm{out}} |F| |\nabla u| ^2.
	\end{equation}
	We now estimate the third term in the r.h.s. of \eqref{Momentum1}: We integrate by parts back to get 
	\begin{equation}\label{Momentum4}
		\int_{\ann} \diff \rv \: \chi_{\mathrm{out}} F \: \curl (iu,\nabla u) \geq - \int_{\ann} \diff \rv \:  \nabla^{\perp}( \chi_{\mathrm{out}} F) \cdot (iu,\nabla u)  - C \int_{\dd \B} \diff \sigma \: |F(1)| \lf| (iu,\partial_{\tau} u) \ri|,
	\end{equation}
	but
	\begin{equation}\label{Momentum41}
	 \int_{\ann} \diff \rv \: \nabla^{\perp}( \chi_{\mathrm{out}} F) \cdot (iu,\nabla u)  =  \int_{\ann} \diff \rv \lf\{ F \nabla^{\perp} (\chi_{\mathrm{out}}) \cdot (iu,\nabla u ) + 2 \chi_{\mathrm{out}} g^2 \vec{B}\cdot (iu,\nabla u) \ri\} 
	\end{equation}
	and the second term can be bounded using the same computations as in the proof of Lemma \ref{lem:initialbound}:
	\[
	 2 \int_{\ann} \diff \rv \: \chi_{\mathrm{out}} g^2 \vec{B} \cdot (iu,\nabla u) \leq \delta \int_{\ann} \diff \rv \: \chi_{\mathrm{out}} g^2 |\nabla u|^2  + C \delta^{-1} \int_{\ann} \diff \rv \: \chi_{\mathrm{out}} g^2 B^2 |u|^2,
	\]
	where $\delta$ is a parameter to be fixed later. For the first term in \eqref{Momentum41} we use \eqref{Fbound2}: 
	\begin{multline*}
	 \left| \int_{\ann} \diff \rv \: \nabla^{\perp}( \chi_{\mathrm{out}}) F  \cdot (iu,\nabla u) \right|  \leq C \ep ^{-1}\int_{\ann} \diff \rv \: \left| \nabla^{\perp} \chi_{\mathrm{out}} \right| \left| r- \rt \right| g ^2|u| |\nabla u| \leq 
	\\  \delta \int_{\left\lbrace \nabla \chi_{\mathrm{out}} \neq 0\right\rbrace} \diff \rv \: g^2 |\nabla u|^2  + \frac{C}{\delta\ep ^2} \int_{\left\lbrace \nabla \chi_{\mathrm{out}} \neq 0\right\rbrace} \diff \rv \: g ^2 |u|^2. 
	\end{multline*}
	The second inequality uses $|1 - \rt | \propto \ep |\log \ep|$ and \eqref{gradientchi}. 
	\newline
	We inject the preceding computations in \eqref{Momentum4}, taking into account that $ B \leq C \ep ^{-1}$. We also note that we have $\chi_{\mathrm{out}} \neq 0 $ and/or $\nabla \chi_{\mathrm{out}} \neq 0 $ only in the unpleasant set and the average set, so 
	\bml{\label{Momentum5}
	  \int_{\ann} \diff \rv \: \chi_{\mathrm{out}} F \: \curl (iu,\nabla u) \geq	\\
	- C \delta \int_{UPS_{\al}\cup AS_{\al}} \diff \rv \: g^2 |\nabla u| ^2 - \frac{C}{\delta \ep ^2} \int_{UPS_{\al}\cup AS_{\al}} \diff \rv \: g^2 |u|^2 - C \int_{\dd \B} \diff \sigma \: |F(1)| \lf| (iu,\partial_{\tau} u) \ri|.
	}
	We gather equations \eqref{Kinetic1}, \eqref{Kinetic2}, \eqref{Momentum1}, \eqref{Momentum3}, \eqref{MomentumOut} and \eqref{Momentum5} to obtain (recalling that $|\chi_{\mathrm{in}}|\leq 1$)
	\begin{multline}\label{lowbound1}
	\int_{\ann} \diff \rv \lf\{ g ^{2} \left| \nabla u \right|^2 + F \: \curl (iu,\nabla u) \ri\} \geq 
	\\  2 \pi \sum_{j\in J} \xi_{\mathrm{in}} (a_j) |d_j| \left[ \left(1-\gamma \right)\left(\frac{1}{2} -\al \right)  g^2 (a_j) \left| \log \ep \right| \left(1-C \frac{\log \left| \log \ep \right|}{\left|\log \ep\right|}\right) -  |F(a_j)|   \right] +
	\\ \left(1-\gamma \right) \int_{\ann} \diff \rv \: \xi_{\mathrm{out}} g^2 |\nabla u |^2 -  \int_{\ann } \diff \rv \: \xi_{\mathrm{out}}  |F| |\nabla u|^2  - C \int_{\dd \B} \diff \sigma \: |F(1)| \lf| (iu,\partial_{\tau} u) \ri| +
	\\ \left(\gamma -\delta \right) \int_{\ann} \diff \rv \: g^2 | \nabla u | ^2 - \frac{C}{\delta \ep ^2} \int_{UPS_{\al}\cup AS_{\al}} \diff \rv \: g^2 |u|^2  - C |\log \ep|^{-1} \Fg[u].
	\end{multline}
	We now choose the parameters in \eqref{lowbound1} as follows:
	\begin{equation}\label{parameters}
	 	\gamma = 2\delta = \frac{\log |\log \ep|}{|\log \ep|}, \hspace{1,5cm} \al = \alt \frac{\log |\log \ep|}{|\log \ep|},
	\end{equation}
	where $\alt$ is a large enough constant (see below). This choice allows to bound the terms in \eqref{lowbound1} from below: Indeed, if $ \Om_0 > 2 (3 \pi )^{-1}$, we have from Proposition \ref{critical ang vel pro}
	\[
	\half g^2 (a_j)\left| \log \ep \right| - |F(a_j)| \geq \frac{C}{\ep |\log \ep| ^{2}}  
	\]
	for any $ \avj \in \at $ and thus	
	\begin{multline}\label{lowbound11}
	 \left(1-\gamma \right)\left(\frac{1}{2} -\al \right)  g^2 (a_j) \left| \log \ep \right| \left(1-C \frac{\log \left| \log \ep \right|}{\left|\log \ep\right|}\right) -  |F(a_j)|   \geq  \\  \frac{1}{2} g^2 (a_j)\left| \log \ep \right|\left(1 - C \frac{\log \left| \log \ep \right|}{|\log \ep|} \right) - |F(a_j)| \geq\frac{C}{\ep |\log \ep| ^{2}}  > 0
	\end{multline}
	where we have used \eqref{g estimates}.\\
	On the other hand, by the definition of $\xi_{\mathrm{out}}$, for any $ \rv \in \mathrm{supp}(\xi_{\mathrm{out}}) $, we have either
	\begin{equation}\label{BL 1}
	 |r-\rt| \leq C \ep |\log \ep|^{-1}
	\end{equation}
	or
	\begin{equation}\label{BL 2}
	 |r -1 | \leq C \ep |\log \ep|^{-1}.
	\end{equation}
	Therefore, using \eqref{Fbound2}, we have in the first case
	\[
	 |F(r)| \leq C \frac{g^2 (r)}{|\log \ep|}.
	\]
	In the second case \eqref{Fbound2} yields
	\[
	 |F(r)| \leq \frac{C}{\ep |\log \ep| ^2} 
	\]
	but \eqref{pointwise bounds} shows that, if $ \rv $ satisfies \eqref{BL 2},
	\[
	 g ^2 (r)\geq  \frac{C}{\ep |\log \ep|}.
	\]
	We conclude that 
	\begin{equation}\label{BL 3}
	 g^2 (r) \geq C |\log\eps| |F(r)| \gg |F(r)|
	\end{equation}
	for any $ \rv\in \mathrm{supp}(\xi_{\mathrm{out}})$ and thus
	\begin{equation}\label{boundarylayer}
	 \left(1-\gamma \right)\int_{\ann} \diff \rv \: \xi_{\mathrm{out}} g^2 |\nabla u |^2 - \int_{\ann} \diff \rv \: \xi_{\mathrm{out}} |F| |\nabla u|^2 \geq 0.
	\end{equation}
	Finally we have from \eqref{lowbound1}, \eqref{lowbound11} and \eqref{boundarylayer}
	\begin{multline}\label{lowbound2}
	\int_{\ann} \diff \rv \lf\{ g ^{2} \left| \nabla u \right|^2 + F \: \curl (iu,\nabla u) \ri\} \geq C \frac{\log \left| \log \ep \right|}{\left|\log \ep\right|} \int_{\ann} \diff \rv \: g^2 | \nabla u | ^2 -
	\\C  \frac{|\log \ep|}{ \ep ^2 \log \left| \log \ep \right|} \int_{UPS_{\al}\cup AS_{\al}} \diff \rv \: g^2 |u|^2  - C \int_{\dd \B} \diff \sigma \: |F(1)| \lf| (iu,\partial_{\tau} u) \ri| - C | \log \ep | ^{-1} \Fg[u].
	\end{multline}
	Adding 
	\bdm
		\int_{\ann} \diff \rv \: \frac{g^4}{\ep^2}\left(1-|u|^2 \right)^2 - \int_{\partial \B} \diff \sigma F(1) (iu , \partial_{\tau}u) 
	\edm
	to both sides of \eqref{lowbound2} and using \eqref{IBP}, we get the lower bound
	\begin{equation}\label{lowbound3}
	 	\Eg [u] \geq C \lf\{ \frac{\log \left| \log \ep \right|}{\left|\log \ep\right|}  \Fg [u] - \frac{|\log \ep|}{ \ep ^2 \log \left| \log \ep \right|} \int_{UPS_{\al}\cup AS_{\al}} \diff \rv \: g^2 |u|^2 -  \int_{\dd \B} \diff \sigma \: F(1) (iu,\partial_{\tau} u) \ri\},
	\end{equation}
	valid for $\ep$ small enough and $\Om_0 > 2 (3\pi)^{-1} $. But $g^2 |u|^2 = |\gpm|^2 \leq C \ep^{-1} |\log \ep|^{-1}$, whereas the side length of a cell is $\OO (\ep |\log \ep|)$, thus
	\begin{equation}\label{bad term 0}
	 \int_{UPS_{\al}\cup AS_{\al}} \diff \rv \: g^2 |u|^2 \leq  \frac{C}{\ep |\log \ep|} \left| UPS_{\al}\cup AS_{\al} \right|\leq C \ep |\log \ep| \left( N_{\al} ^ {\mathrm{UP}} + N_{\al} ^ {\mathrm{A}} \right). 
	\end{equation}
 	Using \eqref{numberbad}, \eqref{numberunpleasant} and \eqref{numberaverage}, we deduce
	\begin{equation}\label{bad term}
	 \int_{UPS_{\al}\cup AS_{\al}} \diff \rv \: g^2 |u|^2 \leq C \ep |\log \ep|  N_{\al} ^ {\mathrm{B}} \leq C \ep ^2 \ep ^{\al} \F [u].
	\end{equation}
	On the other hand \eqref{pointwise bounds} implies that $g^2 (1) \geq C \ep^{-1} |\log \ep|^{-1}$. Combining this fact with the upper bound \eqref{GPmin estimates} yields
	\begin{equation}\label{sup u boundary}
	 |u|\leq C \mbox{ on } \dd \B
	\end{equation}
	and thus, using \eqref{compatibility3} and Cauchy-Schwarz inequality,
	\begin{equation}\label{boundary term 1}
	 \left| \int_{\partial \B} \diff \sigma \: F(1) (iu , \partial_{\tau}u) \right| \leq C \left( \int_{\partial \B} \diff \sigma \: |\dd_{\tau} u|^2 \right)^{1/2}.
	\end{equation}
	Using Lemma \ref{lem:boundary} we conclude
	\begin{equation}\label{boundary term 2}
	 \left| \int_{\partial \B} \diff \sigma \: F(1) (iu , \partial_{\tau}u) \right| \leq C \left( |\log \ep|^{1/2}|\E [u]|^{1/2}+\frac{|\log \ep|^{1/4}}{\ep ^{1/4}} \F[u] ^{1/2}\right). 
	\end{equation}
	Combining \eqref{lowbound3}, \eqref{bad term} and \eqref{boundary term 2}, we have
	\begin{equation}\label{lowbound 4}
	 \Eg [u] \geq C \lf\{ \frac{\log \left| \log \ep \right|}{\left|\log \ep\right|}  \Fg [u] - \frac{|\log \ep|}{\log|\log \ep|}\ep ^{\al} \F[u] - |\log \ep|^{1/2} \lf|\E [u] \ri|^{1/2} - \frac{|\log \ep|^{1/4}}{\ep ^{1/4}} \F[u] ^{1/2}\right\}.
	\end{equation}
	Recall the choice of $\al$ in \eqref{parameters}: We choose now a constant $\alt > 2$. Then
	\[
	  \frac{|\log \ep|}{\log|\log \ep|}\ep ^{\al} = \frac{|\log \ep|^{1-\alt}}{\log|\log \ep|} \ll \frac{\log \left| \log \ep \right|}{\left|\log \ep\right|}
	\]
	and thus there exists a finite constant $ c $ such that
	\begin{equation}\label{lowbound 5}
	 \OO(\ep^{\infty}) \geq \Eg [u] \geq c \left(\frac{\log \left| \log \ep \right|}{\left|\log \ep\right|}  \Fg [u] -   |\log \ep|^{1/2}|\E [u]|^{1/2} - \frac{|\log \ep|^{1/4}}{\ep ^{1/4}} \F[u] ^{1/2}\right)
	\end{equation}
	where the upper bound comes from \eqref{starting bound}. 
	\newline
	Since the sign of $ \Eg [u] $ is not known, we might have two possible cases: If $ \E[u] \geq 0 $, \eqref{starting bound} implies that $ |\E[u]| \leq \OO(\eps^{\infty}) $, which can be plugged in \eqref{lowbound 5} yielding
	\begin{equation}\label{cas 11}
	 \OO(\ep ^{\infty}) \geq \E[u] \geq c \left(\frac{\log \left| \log \ep \right|}{\left|\log \ep\right|}  \Fg [u] -\frac{|\log \ep|^{1/4}}{\ep ^{1/4}} \F[u] ^{1/2}\right) 
	\end{equation}
	 This implies 
	\[
	 \F [u] \leq C \frac{|\log \ep|^{5/2}}{\ep ^{1/2} \log |\log \ep| ^2} 
	\]
	which concludes the proof of Proposition \ref{pro:refinedbound}, if $ \E[u] \geq 0 $.
	\newline
	On the opposite, if $ \E[u] < 0 $, either
	\bdm
		0 \geq \E[u] + c |\log\eps|^{1/2} \lf|\E[u] \ri|^{1/2} \geq c \left(\frac{\log \left| \log \ep \right|}{\left|\log \ep\right|}  \Fg [u] -\frac{|\log \ep|^{1/4}}{\ep ^{1/4}} \F[u] ^{1/2}\right),
	\edm
	which implies the result as before, or $ \left| \E[u] \right| \leq C |\log\eps| $, which gives
	\bdm
		C |\log\eps|^{3/2} \geq c \left(\frac{\log \left| \log \ep \right|}{\left|\log \ep\right|}  \Fg [u] -\frac{|\log \ep|^{1/4}}{\ep ^{1/4}} \F[u] ^{1/2}\right),
	\edm
	and thus again \eqref{borne Fg final} and \eqref{borne Eg final}.
\end{proof}

As already noted, the end of the proof could be formulated as an induction. Plugging the estimates of Lemma \ref{lem:initialbound} in \eqref{lowbound 4} and using the upper bound on $\E [u]$ would yield improved estimates of $\F [u]$ and $|\E [u]|$, thus reducing the number of bad cells and improving the boundary estimate. The process could then be repeated a large number of times, proving that there are no bad cells at all. The second term in the lower bound \eqref{lowbound 4} would then vanish and the process stop when the first term would reach the order of magnitude of the last one (coming from the boundary estimate), thus giving the results of Proposition \ref{pro:refinedbound}.

\section{Energy Asymptotics and Absence of Vortices}
\label{sec:energy asympt}

In this section we conclude the proofs of our main results. The proof of the energy asymptotics is a straightforward combination of the results of Sections \ref{sec:auxiliary} and \ref{Sect est reduced energy}:

\begin{proof}[Proof of Theorem \ref{theo:energy}]
\mbox{}	\\
From Proposition \ref{reduction} we have, using the simplified notation of Section \ref{Sect est reduced energy},
\[
 \gpe \geq \hat{E}^{\mathrm{GP}}_{\A,\om_0} + \E[u] - \OO(\eps^{\infty}) 
\]
which reduces to 
\[
 \gpe \geq \hat{E}^{\mathrm{GP}}_{\A,\om_0}  -C \frac{|\log \ep|^{3/2}}{\ep ^{1/2} \log |\log \ep|}
\]
thanks to \eqref{borne Eg final}. Using a regularization of $g_{\A,\om_0}$ as a trial function for the functional  $\hat{\E}^{\mathrm{GP}}_{\om_0}$ as in the proof of Proposition \ref{reduction}, we get $ \hgpe_{\A,\omega_0} \geq \hgpe_{\omega_0} - \OO(\eps^{\infty}) $ and thus
\[
 \gpe \geq \hat{E}^{\mathrm{GP}}_{\om_0}  -C \frac{|\log \ep|^{3/2}}{\ep ^{1/2} \log |\log \ep|}.
\]
We conclude the proof of the lower bound recalling that, by definition, 
\[
 \hat{E}^{\mathrm{GP}} = \inf_{\om} \hat{E}^{\mathrm{GP}}_{\om} \leq \hat{E}^{\mathrm{GP}}_{\om_0}.
\]

For the upper bound 
\[
\gpe \leq \hat{E}^{\mathrm{GP}} = \hat{E}^{\mathrm{GP}}_{\om_{\mathrm{opt}}}
\]
we simply use $g_{\om_{\mathrm{opt}}}(r) \exp\{i ( [\Om] - \om_{\mathrm{opt}} )\vartheta \}$ as a trial function for $\gpf$.
\end{proof}

The proof of the absence of vortices requires an additional ingredient:

\begin{lem}[\textbf{Estimate for the gradient of $u _{\om _0}$}]
	\mbox{}	\\
	 Recall the definition of $u_{\om_0}$ in \eqref{function u}. There is a finite constant $C$ such that
	\begin{equation}\label{gradu}
	 	\lf\| \nabla u_{\om_0} \ri\|_{L^{\infty}(\at)} \leq C \frac{|\log \ep| ^{3/2}}{\ep ^{3/2}}.
	\end{equation}
\end{lem}

\begin{proof}
	We use the short-hand notation defined at the beginning of Section \ref{Sect est reduced energy} (in particular $u=u_{\om_0}$). Recall the variational equation \eqref{vareq u 2}
	\begin{equation}\label{varequGP}
	-\nabla \lf( g^2 \nabla u \ri) - 2i g ^2 \vec{B} \cdot \nabla u + \frac{2}{\eps^2} g^4 \lf( |u|^2 - 1 \ri) u = \lambda  g ^2 u. 
	\end{equation} 
	From this equation we get the pointwise estimate 
	\begin{equation}\label{deltau}
	 	\left| \Delta u  \right| \leq C \left(\frac{\left| \nabla g \right|}{g}|\nabla u |+ |B||\nabla u|+\frac{1}{\ep ^2}\left|g^2 \lf( |u|^2 - 1 \ri) u\right|+\left|\lambda \right| \left| u \right|\right),
	\end{equation}
	holding true on $\A$. Recalling that $|\gpm| \leq \ep^{-1/2} |\log \ep|^{-1/2}$ and 
	\begin{equation}\label{lowboundg recall}
	g \geq \frac{C}{ \ep^{1/2} |\log \ep|^{3/2} }\mbox{ on } \at,
	\end{equation}
	 we have 
	\begin{equation}\label{bound uGP}
	 |u | \leq C |\log \ep| \mbox{ on } \at.
	\end{equation}
	Thus, using \eqref{g estimates},
	\begin{equation}\label{nonlin term}
	 \frac{1}{\ep ^2}\left|g^2 \lf( |u|^2 - 1 \ri) u\right| \leq C\frac{|\log \ep|^2}{\ep ^3}.
	\end{equation}
	On the other hand, estimating the chemical potential with \eqref{scalelambda} and plugging in the results of Proposition \ref{pro:refinedbound}, we have
	\begin{equation}\label{scalelambdaGP}
	 | \lambda  u| \leq C |u|\left( \left| \E [u] \right| + \frac{1}{\ep ^{3/2}|\log \ep|^{1/2}}\F [u]^{1/2}\right) \leq \frac{C |\log\eps|^{7/4}}{\eps^{7/4} \log|\log\eps|}  \ll \frac{|\log \ep|^2}{\ep ^3}
	\end{equation}
	on $\at$.
	\newline
	Combining \eqref{g estimates}, \eqref{grad est}, \eqref{g low bound} and \eqref{Bbound} with \eqref{nonlin term} and \eqref{scalelambdaGP}, we deduce from \eqref{deltau}
	\begin{equation*}
	 	\left\Vert \Delta u \right\Vert_{L^{\infty} (\at)} \leq C \left( \frac{|\log \ep |^{3/4} }{\ep ^{5/4}}\left\Vert \nabla u  \right\Vert_{L^{\infty} (\at)} + \frac{|\log \ep| ^{2}}{\ep ^3}\right).
	\end{equation*}
	From Gagliardo-Nirenberg inequality \cite[Theorem at pg. 125]{N}, we deduce 
	\begin{equation}\label{deltauLinf}
	 	\left\Vert \Delta u \right\Vert_{L^{\infty} (\at)} \leq C \left( \frac{|\log \ep |^{3/4} }{\ep ^{5/4}}\left\Vert \Delta u  \right\Vert_{L^{\infty }(\at)}^{1/2}\left\Vert u  \right\Vert_{L^{\infty} (\at)}^{1/2} + \frac{|\log \ep| ^{2}}{\ep ^3}\right).
	\end{equation}
	Inserting \eqref{bound uGP}, we conclude
	\[
	 \left\Vert \Delta u \right\Vert_{L^{\infty} (\at)} \leq C \frac{|\log \ep|^{2}}{\ep ^3}
	\]
	and we get \eqref{gradu} by using \eqref{bound uGP} and the Gagliardo-Nirenberg inequality again. 
\end{proof}

We are now in position to complete the 

\begin{proof}[Proof of Theorem \ref{theo:vortex}]\mbox{} \\
The proof relies on a combination of \eqref{borne Fg final} and \eqref{gradu}, as in \cite{BBH1}.\\
Suppose that at some point $ \rv_0 $ such that 
\[
 \rtf + \half \ep|\log \ep|^{-1} \leq r_0 \leq 1,
\]  
we have 
	\[
 		\left| |u(\rv_0)| - 1 \right| \geq   \ep^{1/8} |\log \ep|^{3}.
	\]
Then, using \eqref{gradu}, there is a constant $C$ such that, for any $ \rv \in \B_0 : = \B(\rv_0, C \ep ^{13/8} |\log \ep | ^{3/2})$, we have 
	\[
	 \left| |u(\rv)| -1 \right| \geq  \half \ep^{1/8} |\log \ep|^{3} .
	\]
This implies (recall \eqref{g low bound})
	\[
	 \int_{\B_0} \diff \rv \: \frac{g^4}{\ep ^2 }\left(1-|u|^2\right) ^2 \geq \frac{C |\log \ep|^3}{\ep^{1/2}},
	\]
and thus (note that by the initial condition on $ r_0 $, $ \B_0 \subset \A $)
	\begin{equation}\label{Fginf}
	 	\F [u] \geq \frac{C |\log \ep|^3}{\ep^{1/2}},
	\end{equation}
which is a contradiction with \eqref{borne Fg final}. 
\newline
We have thus proven that (recall \eqref{bound uGP})
\begin{equation}\label{pointwise g psi}
\left| |\gpm|^2 - g^2 \right| \leq g^2  \left| |u|^2 - 1 \right| \leq C \frac{|\log \ep|^3}{\ep^{7/8}}
\end{equation}
on $\at = \left\lbrace \rtf + \ep|\log \ep|^{-1}\leq r \leq 1 \right\rbrace$. The result then follows by a combination of \eqref{pointwise bounds} and \eqref{pointwise g psi}. 
\end{proof}

\begin{rem}{\it (Absence of vortices in a larger domain)}.
	\mbox{}	\\
	By direct inspection of the proof of Theorem \ref{theo:vortex}, one can easily realize that we could have proven the main result, i.e., the absence of vortices, in a domain larger than $ \at $, i.e., there is some freedom in the choice of the bulk of the condensate.
	\newline
	More precisely the choice of a larger domain would have implied a worse lower bound on $ g^2 $ via \eqref{pointwise bounds} and in turn a worse remainder in \eqref{pointwise g psi}, but at the same time this would have allowed the extension of the no vortex result up to a distance of order $ \eps |\log\eps|^{-a} $ from $ \rtf $ for some power $ a > 1 $. 
	\newline
	We have however chosen to state the main result in $ \at $ for the sake of simplicity.
\end{rem}

The proof of the result about the degree of $ \gpm $ is a corollary of the main result proven above:

\begin{proof}[Proof of Theorem \ref{theo:degree}]
\mbox{}	\\
We first note that the pointwise estimate in \eqref{pointwise g psi} implies that $\gpm$ does not vanish on $\dd \B$, so that its degree is indeed well defined. We then compute
\bml{\label{compute degree}
2 \pi \deg \{\gpm, \dd \B\} = - i \int_{\dd \B} \diff \sigma \: \frac{|\gpm|}{\gpm} \dd_{\tau} \left( \frac{\gpm}{|\gpm|} \right) = - i \int_{\dd \B} \diff \sigma \: \frac{|u|}{u} \dd_{\tau} \left( \frac{u}{|u|} e^{i\left([\Om]-\om_0 \right)\vartheta}\right) e^{-i\left([\Om]-\om_0 \right)\vartheta} =
\\ 2\pi \left([\Om]-\om_0\right) - i \int_{\dd \B} \diff \sigma \:  \frac{|u|}{u} \dd_{\tau} \left( \frac{u}{|u|}\right).
}
Then 
\beq\label{error degree}
\bigg| \int_{\dd \B} \diff \sigma \: \frac{|u|}{u} \dd_{\tau} \left( \frac{u}{|u|}\right) \bigg|\leq \int_{\dd \B} \diff \sigma \:  \bigg| \dd_{\tau} \left( \frac{u}{|u|}\right) \bigg| \leq
 C \int_{\dd \B} \diff \sigma \: \left| \dd_{\tau}  u \right|,
\eeq
where we have used that $|u|$ is bounded below by a constant on $\dd \B$. 
\newline
Finally, combining \eqref{boundary} and the results of Proposition \ref{pro:refinedbound}, we obtain (recall that $g^ 2 \geq C \ep^{-1} |\log \ep|^{-1}$ on $\dd \B$) 
\begin{equation}\label{error degree 2}
\int_{\dd \B} \diff \sigma \: \left| \dd_{\tau}  u \right|^2 \leq \frac{C |\log \ep|^{3}}{\ep (\log|\log\eps|)^2}.   
\end{equation}
Using the Cauchy-Schwarz inequality, we thus conclude from \eqref{compute degree}, \eqref{error degree} and \eqref{error degree 2} that
\begin{equation}\label{degre gpm}
\deg  \{ \gpm, \dd \B \} = [\Om]-\om_0 + \OO \lf(\eps^{-1/2} |\log \ep|^{3/2} (\log|\log\eps|)^{-1} \ri).
\end{equation}
There only remains to recall that (see \eqref{est omega_0})
\[
 \om_{0} = \frac{2}{3 \sqrt{\pi}\ep} + \OO (\eps^{-1} |\log\eps|^{-1/2})
\]
and that an identical estimate applies to $ \omega_{{\rm opt}} $ (see \eqref{est omega_{opt}}).
\end{proof}

\begin{rem}{\it (Degree of a GP minimizer)}.
\mbox{}	\\
According to \eqref{degre gpm}, we could have stated the result \eqref{gpm degree} about the degree of $ \gpm $ in terms of $ \omega_0 $, i.e., the optimal giant vortex phase when the minimization problem is restricted to the annulus $ \A $, instead of $ \omega_{\rm opt} $, i.e., the optimal giant vortex phase in the whole of $ \B $. Moreover the remainder in  \eqref{degre gpm}, i.e., $ \OO(\eps^{-1/2}|\log\eps|^{3/2} (\log|\log\eps|)^{-1}) $, is much better than the one contained in the final result \eqref{gpm degree}, i.e., $ \OO(\eps^{-1} |\log\eps|^{-1/2}) $, which is inherited from \eqref{est omega_0} and \eqref{est omega_{opt}}. Note however that the latter remainder is the best precision to which one can estimate the giant vortex phase in terms of the explicit quantity $ 2/(3\sqrt{\pi}) \eps^{-1} $. For this reason and the fact that $ \omega_{\rm opt} $ occurs more naturally in the analysis, we have used it in \eqref{gpm degree}.
\end{rem}

\renewcommand{\thesection}{Appendix A}


\section*{Appendix A}
\addcontentsline{toc}{section}{Appendix A}

\renewcommand{\theequation}{A.\arabic{equation}}
\setcounter{equation}{0}
\renewcommand{\thesection}{A}

In this Appendix we discuss some useful properties of the TF-like functionals involved in the analysis as well as the critical angular velocity $ \Omega_c $ for the emergence of the giant vortex phase.

\subsection{The TF Functionals}

We start by considering the TF functional defined in \eqref{TFf}: 
\bdm
	\tff[\rho] : = \int_{\B} \diff \rv \: \lf\{ - \Omega^2 r^2 \rho + \eps^{-2} \rho^2 \ri\}.
\edm
Its minimizer over positive functions in $ L^1(\B) $ is unique and is given by the radial density
\beq
	\label{TFm}
	\tfm(r) : = \frac{1}{2} \lf[ \eps^2 \tfchem + \eps^2 \Omega^2 r^2 \ri]_+ = \frac{\eps^2 \Omega^2}{2} \lf[ r^2 - \rtf^2 \ri]_+,
\eeq		
where $ [ \:\: \cdot \:\: ]_+ $ stands for the positive part and the chemical potential is fixed by normalizing $ \tfm $ in $ L^1(\B) $, i.e.,
\beq
	\label{tfchem}
	\tfchem : = \tfe + \eps^{-2} \lf\| \tfm \ri\|^2_2 = - \Omega^2 \rtf^2.
\eeq 
Note that, if $ \Omega \geq 2/(\sqrt{\pi} \eps) $, the TF minimizer is a compactly supported function, since it vanishes outside $ \tfd $, i.e., for $ r \leq \rtf $, where
\beq
	\label{TFann}
	\rtf : = \sqrt{1 - \frac{2}{\sqrt{\pi} \eps \Omega}}.
\eeq
The corresponding ground state energy can be explicitly evaluated and is given by
\beq
	\label{TFe}
	\tfe = - \Omega^2 \lf( 1 - \frac{4}{3 \sqrt{\pi} \eps \Omega} \ri). 
\eeq
By \eqref{TFann} and \eqref{TFe} the annulus $ \tfd $ has a shrinking width of order $ \eps |\log\eps| $ and the leading order term in the ground state energy asymptotics is $ - \Omega^2 $, which is due to the convergence of $ \tfm $ to a delta function supported at the boundary of the trap.

In other sections of the paper we often consider the restrictions of the functionals to domains $ \D $ strictly contained inside $ \B $ (denoted by $ \tff_{\D} $). However in the case of the TF functional there is no need to make a distinction between $ \tff $ and $ \tff_{\D} $ since all the ground state properties are basically independent of the integration domain, provided $ \tfd \subset \D $.

Another important TF-like functional is defined in \eqref{hatTF} and includes the giant vortex energy contribution, i.e., 
\beq
 	\label{hTFf}
	\htff_{\omega}[\rho] : = \int_{\B} \diff \rv \: \lf\{ - \Omega^2 r^2 \rho + B_{\omega}^2(r) \rho + \eps^{-2} \rho^2 \ri\} = \int_{\B} \diff \rv \: \lf\{ \lf( [\Omega] - \omega \ri)^2 r^{-2} \rho + \eps^{-2} \rho^2 \ri\} - 2 \Omega [\Omega - \omega],
\eeq
where the potential $ \rmagnp $ is defined in \eqref{newB0}, $ \omega \in \Z $ and we have used the normalization in $ L^1(\B) $ of the density in the last term. 
\newline
The minimization is essentially the same as for \eqref{TFf}: The normalized minimizer is 
\beq
	\label{hTFm}
	\htfm_{\omega}(r) : = \frac{\eps^2}{2} \lf[ \htfchem_{\omega} - \lf( [\Omega] - \omega \ri)^2 r^{-2}\ri]_+,
\eeq 
and the normalization condition becomes
\beq
	\label{hat normalization}
	\frac{1 - \hrtf^2}{\hrtf^2} + \log \hrtf^2 = \frac{2}{\pi \eps^2 ([\Omega] - \omega)^2},
\eeq
where we have denoted
\beq
	\label{hat radius}
	\hrtf^2 : = \frac{([\Omega] - \omega)^2}{\htfchem_{\omega}}.
\eeq
With such a definition the minimizer \eqref{hTFm} can be rewritten in a form very close to the TF minimizer \eqref{TFm}, i.e.,
\bdm
	\ttfm_{\omega}(r) = \frac{\eps^2 ([\Omega] - \omega)^2}{2 \hrtf^2 r^2} \lf[ r^{2} - \hrtf^2 \ri]_+.
\edm
In order to make a comparison it would then be useful to evaluate the radius $ \hrtf $ but the equation \eqref{hat normalization} has no explicit solution. However, since the right hand side of \eqref{hat normalization} tends to zero as $ \eps \to 0 $, we can expand the left hand side assuming $ \hrtf^{-2} = 1 + \delta $ for some $ \delta \ll 1 $:
\bdm
 	\frac{1}{2} \delta^2 - \frac{1}{3} \delta^3 + \OO(\delta^4) = \frac{2}{\pi \eps^2 ([ \Omega] - \omega)^2},
\edm
which yields
\bdm
	\frac{1}{\hrtf^2} - 1 = \delta = \frac{2}{\sqrt{\pi} \eps ([\Omega] -\omega)} \lf[ 1 + \frac{2}{3 \sqrt{\pi} \eps ([\Omega] - \omega)} + \frac{1}{9 \pi \eps^2 ([\Omega]-\omega)^2} + \mathcal{O}(\eps^{-3} \Omega^{-3}) \ri].
\edm
We thus have
\beq
	\label{new radius}
	\hrtf^2 = 1 - \frac{2}{\sqrt{\pi} \eps ([\Omega] - \omega)} - \frac{4}{3 \pi \eps^2 ([\Omega] - \omega)^2} + \mathcal{O}(\eps^{-3}\Omega^{-3}) = \rtf^2 + \frac{2 \omega}{\sqrt{\pi} \eps \Omega^2} - \frac{4}{3 \pi \eps^2 \Omega^2} + \OO(\eps^3 |\log\eps|^3),
\eeq
and whether $ \hrtf $ is larger or smaller than $ \rtf $ depends in a crucial way on the phase $ \omega $: In particular in the case of the giant vortex phase $ \omega_0 $ (see Proposition \ref{optimal phase pro}), the sum of the two last terms in the above expression vanishes to the leading order (see \eqref{est omega_0}), i.e., it is much smaller than $ \OO(\eps^{2}|\log\eps|^2) $.
\newline
The ground state energy $ \htfe $ is easy to compute:
\begin{multline*}
	\htfe_{\omega} = - 2 \Omega \lf([\Omega] - \omega\ri)+ \frac{\pi \eps^2 ([\Omega] - \omega)^4}{4} \lf( \hrtf^{-2} - 1 \ri)^2 \\ 
	= - \Omega^2 + \frac{4 \Omega}{3 \sqrt{\pi} \eps} + \omega^2 - \frac{4 \omega}{3 \sqrt{\pi} \eps} + \frac{2}{3 \pi \eps^2} -2 \Omega \left( [\Omega] - \Omega \right) + \OO(\eps^{-2} |\log\eps|^{-2}),
\end{multline*}
and, assuming that $ |\omega| \leq \OO(\eps^{-1}) $, one can easily recognize that the leading term and the first remainder coincide with \eqref{TFe}, i.e., the energy $ \htfe_{\omega} $ is equal to $ \tfe $ up to second order corrections:
\beq
	\label{hTFe}
	\htfe_{\omega} = \tfe + \bigg[ \omega - \frac{2}{3\sqrt{\pi}\eps} \bigg]^2 + \frac{2}{9 \pi \eps^2}  -2 \Omega \left( [\Omega] - \Omega \right) + \OO(\eps^{-2} |\log\eps|^{-2}).
\eeq
This formula implies that $ \htfe_{\omega} $ is minimized by a phase which is given up to corrections of order $ \eps^{-1} |\log\eps|^{-1} $ by
\beq
	\label{TFphase}
	\optphtf : = \frac{2}{3\sqrt{\pi}\eps},
\eeq
and the same is true for the giant vortex phases $ \omega_0 $ (see Proposition \ref{optimal phase pro} and \eqref{est omega_0}) and $\om_{\mathrm{opt}}$ (see Proposition \ref{optimal phase omega_{opt} pro} and \eqref{est omega_{opt}}).

\subsection{The Critical Angular Velocity and the Vortex Energy}

The last part of this Appendix is devoted to the study of the critical velocity $ \Omega_c $, which is defined as the angular velocity at which vortices disappear from the bulk of the condensate. To estimate this velocity, according to the discussion in Section \ref{Sect est reduced energy}, we have to compare the vortex energy cost $ \frac{1}{2} g^2(r) |\log\eps| $ and the vortex energy gain $ |F(r)| $ (see \eqref{F}): This leads to the definition of the function
\beq
	\gain(r) : = \half g_{\A,\omega_0}^2(r) |\log\eps| - \lf|F_{\omega_0}(r)\ri|,
\eeq
which yields the overall energy contribution of a vortex at a radius $ \rv $ inside the bulk: If $ \gain $ is positive in some region, then a vortex is energetically unfavorable there, and, if this holds true in the whole of the bulk, the condensate is in the giant vortex phase.
\newline
Before studying the behavior of the above function $ \gain $, it is however convenient to obtain an explicit approximative value for the critical velocity and to this purpose we replace the density $ g^2 $ with $ \tfm $ and study the function
\beq
	\label{gainTF}
	\gaintf(r) : =\half |\log\eps| \tfm(r) - \lf|\costtf(r) \ri|,
\eeq
where the cost function $ \costtf $ is explicitly given by
\beq
	\label{exp TF potential}
	\costtf(r) : = 2 \int_{\rtf}^r \diff s \: \vec{B}_{\optphtf}(s) \cdot \vec{e}_{\vartheta} \tfm(s) = \eps^2 \Omega^2 \int_{\rtf}^r \diff s \: \lf[ \Omega s - \lf( [\Omega] - \optphtf \ri) s^{-1} \ri] (s^2 - \rtf^2), 
\eeq
with $ \optphtf $ defined in \eqref{TFphase}. 
\newline
In order to investigate the behavior of the infimum of $ \gaintf $ inside the bulk, it is convenient to rescale the quantities and set
\beq \label{rescalez}
	z : = \eps \Omega (r^2 - \rtf^2),
\eeq
so that $ z $ varies on a scale of order one, i.e., more precisely $ z \in [0, 2/\sqrt{\pi} ] $ (see \eqref{TFann}). With such a choice the gain function can be easily estimated:
\bmln{
 	\costtf(r) = \frac{\eps^2 \Omega^2}{2} \int_{0}^{r^2-\rtf^2} \diff t \: t  \lf[ \Omega (t + \rtf^2) - [\Omega] + \optphtf \ri] \lf(t + \rtf^2 \ri)^{-1} =	\\
	\frac{\eps^2 \Omega^2}{2} \int_{0}^{r^2-\rtf^2} \diff t \:  t \lf( \Omega t - \frac{4}{3\sqrt{\pi}\eps} + \OO(1) \ri) \lf(1 - \frac{2}{\sqrt{\pi} \eps \Omega} + t \ri)^{-1}= 	\\
	\frac{1}{2\eps} \int_{0}^{z} \diff s \:  s  \lf( s - \frac{4}{3\sqrt{\pi}} + \OO(\eps) \ri) \lf(1 - \frac{2}{\sqrt{\pi} \eps \Omega} + \frac{s}{\eps\Omega} \ri)^{-1} = 	\\
	\frac{1}{2\eps} \int_{0}^{z} \diff s \: s \lf( s - \frac{4}{3\sqrt{\pi}} \ri) + \OO(|\log\eps|) = \frac{z^2}{6 \eps} \lf( z - \frac{2}{\sqrt{\pi}} \ri) + \OO(|\log\eps|),
	}
where we have used the approximation $ [1 - \OO((\eps\Omega)^{-1})]^{-1} = 1 + \OO((\eps\Omega)^{-1}) $. 
\newline
Applying the same rescaling to the energy cost function, we thus obtain
\beq
	\label{Hrescaling}
	\gaintf(r) : = \frac{z \rgaintf(z)}{12\eps},
\eeq
where
\beq
	\label{Hrescaled}
	\rgaintf(z) = 3\Omega_0 - 2z \lf| z - \frac{2}{\sqrt{\pi}} + \OO(\eps|\log\eps|) \ri| = 3\Omega_0 - 2 z \lf( \frac{2}{\sqrt{\pi}} - z \ri) - \OO(\eps|\log\eps|),
\eeq
since $ z \leq 2/\sqrt{\pi} $ by the definition of the scaling. Now it is very easy to see that
\beq	\label{Hinf}
	\rgaintf(z) \geq 3\Omega_0 - 2\pi^{-1} - \OO(\eps|\log\eps|).
\eeq
	

The above considerations lead to the following 

\begin{pro}[\textbf{TF vortex energy}]
		\label{TF critical ang vel pro}
		\mbox{}	\\
		For any $ \Omega_0 > 2 (3\pi)^{-1} $ and $ \eps $ small enough, there exists a finite constant $ C $ such that 
		\[
		 \gaintf(r) \geq  C \eps^{-1} |\log \ep| ^{-2} > 0 
		 \]
		  for any $\vec{r}$ such that $ r \geq R_{\rm h} +  \ep |\log \ep| ^{-1} $.
	\end{pro}

\begin{proof}

It is sufficient to collect (\ref{Hrescaling}), (\ref{Hinf}) and recall the rescaling (\ref{rescalez}).  

\end{proof}

We now go back to the original function $ \gain $ and show that the above property holds true as well, i.e., $ \Omega_{c}  = 2 (3\pi)^{-1} \eps^{-2} |\log\eps|^{-1} $ is the critical velocity for the disappearance of vortices from the bulk of the condensate: 

	\begin{pro}[\textbf{Critical angular velocity}] 
		\label{critical ang vel pro}
		\mbox{}	\\
		If $ \Omega_0 > 2 (3\pi)^{-1} $ and $ \eps $ is small enough, there exists a finite constant $ C $ such that 
		\[ 
		\gain(r) \geq C \eps^{-1} |\log \ep| ^{-2} > 0 
		\] 
		for any $\vec{r}$ such that $ r \geq \rtf +  \ep |\log \ep| ^{-1} $.
	\end{pro}

	\begin{proof}
		The result basically follows from what is proven about $ \gaintf $: We are going to show that 
		\beq
			\label{difference energy functions}
			\sup_{\rv \in \at} \lf|  \gaintf(r) - \gain(r) \ri| \leq C \eps^{-1}|\log \ep| ^{-1}.
		\eeq
		\newline
		In order to prove the above inequality, we use the estimates \eqref{pointwise bounds} and \eqref{g estimates} to get
		\beq
			\label{difference densities}
			\sup_{\rv \in \at} \lf|  \tfm(r) - g_{\A,\omega_0}^2(r) \ri| \leq C \eps^{1/2} |\log\eps|^{7/2} \lf\| \tfm \ri\|_{\infty} \leq C \eps^{-1/2} |\log\eps|^{5/2},
		\eeq
		and
		\bml{
 			\sup_{\rv \in \at} \lf|  \costtf(r) - F_{\omega_0}(r) \ri| \leq 2 \int_{\rt}^{\rd} \diff s \: \lf| B_{\omega_0}(s) \ri| g_{\A,\omega_0}^2(s) +	\\
			C \eps^2 \Omega^2 \lf| \omega_0 - \optphtf \ri| \int_{\rtf}^1 \diff s \: s^{-1} (s^2 - \rtf^2) + 2 \sup_{\rv \in \at} \lf|  \tfm(r) - g_{\A,\omega_0}^2(r) \ri| \int_{\rtf}^1 \diff s \: \lf| B_{\omega_0}(s) \ri| \leq	\\
			C \lf[ \eps^{-1} \lf| \rd - \rt \ri| g^2_{\A,\omega_0}(\rd) + \eps^3 \Omega^2 |\log\eps| + \eps^{-1/2} |\log\eps|^{7/2} \ri] \leq	\\
			C \lf[ |\log\eps|^{-1} \tfm(\rd) + \eps^{-1} |\log\eps|^{-1} \ri] \leq C \eps^{-1} |\log\eps|^{-1},
		}
		where we have used \eqref{difference densities}, the monotonicity of $ g_{\A,\omega_0}(r) $ (see Proposition \ref{htminimization}) and the estimate \eqref{est omega_0}.
		\newline
		Hence one obtains \eqref{difference energy functions} and the final result follows from Proposition \ref{TF critical ang vel pro} if $r\geq \rtf +  \ep |\log \ep |^{1/2}$. Indeed, using (\ref{rescalez}), (\ref{Hrescaling}) and (\ref{Hinf}) we have  
		\[
		\gaintf(r)\geq C \ep ^{-1}|\log \ep| ^{-1/2} \]
		 in this case.\\
On the other hand, if $ \rtf + \ep |\log \ep| ^{-1} \leq r \leq \rtf + \ep |\log \ep| ^{1/2}$ it follows from (\ref{Fbound2}) that
\[
H(r) \geq \half |\log \ep| g_{\A,\om_0} ^2 (r) (1 - C |\log\eps|^{-2}) \geq C \frac{1}{\ep |\log \ep| ^2},
\] 
where the last inequality comes from (\ref{g low bound}).
	\end{proof}

\vspace{1cm}
\noindent{\bf Acknowledgements.} MC and NR gratefully acknowledge the hospitality of the {\it Erwin Schr\"{o}dinger Institute} (ESI). JY acknowledges the hospitality of the {\it Institute for Mathematical Sciences} (IMS) at the National University of Singapore. MC is partially supported by a grant {\it Progetto Giovani GNFM} and NR by {\it R\'{e}gion Ile-de-France} through a PhD grant. NR thanks Sylvia Serfaty for helpful discussions.

\end{document}